\documentclass[10pt]{article}
\usepackage[left=2cm,right=2cm,top=3.5cm,bottom=3.5cm]{geometry}
\usepackage[utf8]{inputenc}
\usepackage[T1]{fontenc}
\usepackage[labelfont=bf]{caption}
\usepackage{microtype}
\usepackage{fnpct}

\usepackage{authblk}
\makeatletter\renewcommand\AB@affilsepx{, \protect\Affilfont}\makeatother

\renewenvironment{abstract}
 {\small
  \begin{center}
  \bfseries \abstractname\vspace{-.5em}\vspace{0pt}
  \end{center}
  \list{}{%
    \setlength{\leftmargin}{5em}
    \setlength{\rightmargin}{\leftmargin}%
  }%
  \item\relax}
 {\endlist}

\usepackage{hyperref}
\hypersetup{colorlinks=true, citecolor=blue, linkcolor=blue}

\usepackage[title]{appendix}

\usepackage[inline]{enumitem}

\usepackage{amsthm}
\usepackage{stmaryrd}
\usepackage{bbm}
\usepackage{amsmath}
\usepackage{amssymb}
\usepackage{mathabx}

\newcommand{\powerset}{%
  \mathchoice{\raisebox{.1\baselineskip}{\Large$\displaystyle\wp$}}
             {\raisebox{.1\baselineskip}{\Large$\textstyle\wp$}}
             {\raisebox{.1\baselineskip}{\Large$\scriptstyle\wp$}}
             {\raisebox{.1\baselineskip}{\Large$\scriptscriptstyle\wp$}}}

\makeatletter
\newcommand{\circlearrow}{}
\DeclareRobustCommand{\circlearrow}{%
  \mathrel{\vphantom{\rightarrow}\mathpalette\circle@arrow\relax}%
}
\newcommand{\circle@arrow}[2]{%
  \m@th
  \ooalign{%
    \hidewidth$#1\circ\mkern1mu$\hidewidth\cr
    $#1\longrightarrow$\cr}%
}
\makeatother

\newcommand*{\defeq}{\mathrel{\vcenter{\baselineskip0.5ex\lineskiplimit0pt\hbox{\scriptsize.}\hbox{\scriptsize.}}}=}

\newcommand*{\lowerset}{{\downarrow}}

\newtheorem{theorem}{Theorem}
\newtheorem{proposition}{Proposition}
\newtheorem{corollary}{Corollary}

\theoremstyle{definition}
\newtheorem{definition}{Definition}
\newtheorem{example}{Example}

\theoremstyle{remark}
\newtheorem{notation}{Notation}
\newtheorem{remark}{Remark}

\newcommand{\Codom}{\operatorname{Codom}}
\newcommand{\Dom}{\operatorname{Dom}}
\newcommand{\rnk}{\operatorname{rnk}}

\newcommand{\supp}{\operatorname{supp}}

\newcommand{\dprime}{\prime \prime}

\def\Plus{\texttt{+}}
\def\Minus{\texttt{-}}

\def\AA{{\mathbb A}}
\def\NN{{\mathbb N}}
\def\RR{{\mathbb R}}
\def\VV{{\mathbb V}}
\def\indfun{{\mathbbm 1}}
\def\NNex{{\overline{\mathbb N}}}

\def\A{{\mathcal A}}
\def\B{{\mathcal B}}
\def\C{{\mathcal C}}
\def\D{{\mathcal D}}
\def\F{{\mathcal F}}

\def\Hh{{\mathcal H}}
\def\I{{\mathcal I}}
\def\L{{\mathcal L}}
\def\M{{\mathcal M}}
\def\N{{\mathcal N}}
\def\Q{{\mathcal Q}}
\def\R{{\mathcal R}}
\def\S{{\mathcal S}}
\def\T{{\mathcal T}}
\def\U{{\mathcal U}}
\def\W{{\mathcal W}}
\def\Y{{\mathcal Y}}
\def\Z{{\mathcal Z}}

\usepackage{tikz}
\usetikzlibrary{graphs}
\usetikzlibrary{arrows}
\usetikzlibrary{positioning,chains,fit,shapes,calc}

\title{Abstract relational structures in models of biology}

\author[a,b,1]{L\'eo Diaz}
\author[a,b,c,1]{Sean T. Vittadello}
\author[a,b,d,2]{Michael P.H. Stumpf}

\affil[a]{School of BioSciences, University of Melbourne, Melbourne, Australia}
\affil[b]{School of Mathematics and Statistics, University of Melbourne, Melbourne, Australia}
\affil[c]{ARC Centre of Excellence for the Mathematical Analysis of Cellular Systems (MACSYS), Melbourne, Australia}
\affil[d]{Cell Bauhaus, Melbourne, Australia}

\date{May 22, 2026}

\begin{document}

\maketitle

\begin{abstract}
    The mathematical formalisms used to model biological systems induce both latent and ambiguous assumptions that can limit or distort their representational capabilities. Developing formalisms that can represent systems more precisely is fundamental to comprehending their intricacies and complexities. Here we introduce the \emph{systems hypergraph}, a general and extendable formalism for representing abstract relational systems. A systems hypergraph combines a hypergraph, representing multidimensional relations among objects, with a hierarchical system of attributes representing system properties and their interdependencies. The attribute structure ensures that dependencies between system properties are patent and unambiguous, thereby clarifying assumptions and avoiding redundancy in data association. As an application we consider two formalisms widely used in systems biology -- chemical reaction networks and stochastic Petri nets -- and study their natural representation as systems hypergraphs. This allows us to relate the two formalisms rigorously, demonstrating in particular that stochastic Petri nets are strictly more general than chemical reaction networks in contrast to their commonly assumed equivalence. More broadly our work demonstrates the power of abstraction, and in particular its role in mediating between objects and relations in mathematical representations of biological complexity.
\end{abstract}

\footnotetext[1]{L.D. and S.T.V. contributed equally to this work.}
\footnotetext[2]{To whom correspondence should be addressed. E-mail: mstumpf@theosysbio.org.}

\section{Introduction}

Biological systems are inherently multifaceted, so favourably modelled from multiple conceptual perspectives. Systems are generally characterised by integration, organisation, interdependence, and interaction, for which relational approaches offer a judicious modelling foundation. Importantly, a relational perspective can make legible system-level properties not evident from the constituent parts considered in isolation -- a distinguishing characteristic of biological systems, and the premise of systems biology~\cite{FoxKeller2005Century}.

Biological systems are often viewed from a relational perspective using network metaphors, formalised as graphs: parts of a biological system become vertices, and their pairwise relationships become edges. The representation appears evident and graphs are therefore often described as providing a common language for representing complex systems~\cite{Mesarovic2004Search,Torres2021Why,Brown2025}. For molecular systems at the subcellular level, two graph-based formalisms have gained prominence: chemical reaction networks (CRNs)~\cite{Feinberg2019Foundations} and stochastic Petri nets (SPNs)~\cite{Koch2010Modeling}.

Designing a formalism for a specific class of systems provides a template that makes certain phenomena representable. For instance, CRNs were originally designed to represent chains of chemical transformations where reactants are typically converted into products by enzymes. This specificity simplifies model development through domain-specific assumptions, but those same assumptions also constrain the mathematical adaptability of the models. Moreover, alternative interpretations of a system may produce different formalisms, which in turn determine a model's accuracy, scope of applicability, limitations, and constructive connections with other areas of mathematics.

The efficacy of modelling a system from multiple perspectives is greatly improved when the interrelationships between the models are known~\cite{Vittadello2021Model}. Assumptions specific to each model can, however, obscure these interrelationships~\cite{Vittadello2022Group}. For instance, both CRNs and SPNs are used to represent chemical reaction systems, but they are based on different representational assumptions leading to distinct mathematical structures. While equivalence of the two formalisms is often purported, this connection is not immediate as formal comparison requires proper consideration of their distinct mathematical forms.

Abstraction offers a powerful way to relate and integrate different formalisms. When multiple formalisms represent the same kind of system we can identify their common underlying properties and produce a more abstract formalism with the originals as special cases. This helps focus attention on the essential properties, and can reveal subtle relationships between alternative formalisms.

Here we use hypergraphs to represent relational structures, as a natural abstraction and generalisation of graphs. While graphs represent pairwise relations, hypergraphs explicitly represent multidimensional relations -- a capability well matched to the multidimensional nature of biological systems~\cite{Klamt2009Hypergraphs,Taylor2015Higherorder}. We go further by introducing \emph{systems hypergraphs}, a general formalism that combines a hypergraph with a dependency hierarchy of \emph{attributes} on the hypergraph structure. Systems hypergraphs encompass both CRNs and SPNs -- facilitating the identification of equivalent data across the two formalisms, and demonstrating that their interchangeability only holds under certain conditions -- as well as related work representing CRNs as specialised hypergraphs~\cite{Jost2019Hypergraph}.

\section{Background}
Here we provide the general notation and definitions that are foundational for this article.

\subsection*{Notation}
$\NN$ (resp. $\NN_+$) is the set of nonnegative (resp. positive) integers and $\RR$ (resp. $\RR_{\ge 0}$) is the set of real (resp. nonnegative real) numbers. The power set of any set $V$ is $\powerset(V)$, which includes the empty set $\emptyset$, and for $W \subseteq \powerset(V)$ the downward closure of $W$ in $\powerset(V)$ with respect to inclusion is $\lowerset W \subseteq \powerset(W)$. In particular we refer to $\I \defeq \powerset\big(\{-1,+1\}\big) \setminus \emptyset$ as the \emph{orientation set}.

For a sequence of nonempty sets $(A_i)_{i = 1}^n$ with $n \in \NN_+$ the Cartesian product $A_1 \times \cdots \times A_n$ is denoted $\bigtimes_{i = 1}^n A_i$. For a family of nonempty sets $\{A_i\}_{i \in I}$ with nonempty finite index set $I$ the \emph{unordered Cartesian product} $\prod_{i \in I} A_i$ is the set of all functions $f \colon I \to \bigcup_{i \in I} A_i$ such that $f(i) \in A_i$ for all $i \in I$, and each such $f$ is an \emph{unordered $I$-tuple}. If the sets in $\{A_i\}_{i \in I}$ are pairwise disjoint then for notational simplicity we identify $f \in \prod_{i \in I} A_i$ with $\{f(i)\}_{i \in I}$.

For two sets $A$ and $B$ the set of all functions $f$ with domain $\Dom(f) \defeq A$ and codomain $\Codom(f) \defeq B$ is $\{A \rightarrow B\}$, and if $B \subseteq \RR$ then the support of $f$ is $\supp(f) \defeq \{\, a \in A \mid f(a) \ne 0 \,\}$. The set of all partial functions $f$ from $\Dom(f) \subseteq A$ to $B$, including the empty function, is $\{A \rightharpoonup B\}$. The set of all functions with domain in a family of sets $\A$ and codomain $B$ is $\{\A \rightrightarrows B \}$.

A \emph{multiset} (resp. \emph{positive multiset}) is a pair $(S,\mu)$ where $S$ is the \emph{underlying set}, and $\mu \colon S \to \NN$ (resp. $\mu \colon S \to \NN_+$) is the \emph{multiplicity function} which indicates the \emph{multiplicity}, or integer number of instances, of each element in $S$. If $\mu$ is a partial function then $(S,\mu)$ is a \emph{partial multiset} (resp. \emph{positive partial multiset}), and the support of $(S,\mu)$ is $\supp(S,\mu) \defeq \supp(\mu)$. For notational simplicity we may identify a multiset $(S,\mu)$ with its multiplicity function $\mu \colon S \to \NN$.

\begin{remark}
    In a multiset the instances of an element are considered indistinguishable with regard to intrinsic properties, however we assume that the instances can be individuated based on an extrinsic property, such as spatial position. Therefore, the instances are individual objects to which we can assign labels.
\end{remark}

\subsection*{Modelling considerations}
The two concepts \emph{absence} and \emph{abstraction} require careful treatment in this paper.

\begin{definition}[Absence]\label{def:absence}
    We want to distinguish between two cases of absence: \emph{absence of data} indicates a lack of knowledge, whereas \emph{data of absence} indicates knowledge of either nonexistence or nonoccurrence. While these two cases have distinct interpretations they may in practice be represented by the same data value and become indistinguishable, thereby requiring strategies for disambiguation.
\end{definition}

\begin{remark}
    The consideration of absence motivates our use of \emph{positive partial} multisets. Multisets are typically defined as total functions into $\NN$, but a zero multiplicity is then ambiguous: it could represent either the absence of data or data of absence. Positive partial multisets resolve this by encoding absence of data through partiality -- if a function is undefined on an element, no data is available -- while reserving explicit values for known quantities.
\end{remark}

Reality is often represented as a system of objects and relations: a relation attributes properties to objects, while relations do not themselves possess properties. Note that unary relations are often called \emph{properties} in first-order logic. Objectification allows us to logically transform a relation into an object, which we can then investigate and to which properties can be predicated. Objectification is formally described by the logical operation of \emph{hypostatic abstraction}, also called subjectal abstraction or positive reification. Charles Sanders Peirce\footnote[2]{Peirce made contributions across mathematics, philosophy, logic, and science.} considered hypostatic abstraction to be the principal operation in mathematics, and observed that hypostatic abstraction presupposes the distinct logical operation of \emph{prescissive abstraction}, which associates a relation to an object so that only properties of interest are attributed while ignoring the rest~\cite{Zeman1982Peirce}. 

\begin{definition}[Hypostatic abstraction]
    A \emph{hypostatic abstraction} transforms a relation $R$ over a domain $D$ into another logically equivalent relation $R^{\ast}$ over a domain $D^{\ast}$ such that: $D^{\ast}$ consists of the original objects in $D$ and a new distinct object $h(R)$, the \emph{hypostatic object}, which is the hypostatised form of $R$ representing the properties attributed by $R$; the information of $R$ is preserved in $R^{\ast}$, which has arity one greater than $R$. Hypostatic abstraction can result in two different hypostatic objects, whereby the relation given by the predicate `is an X' transforms into the relation given by the predicate: `is a member of the set of Xs' equivalently 'has the property of being X'; or, `possesses the character of being an X' equivalently `has X-ness'. Note that the hypostatic object is not directly defined: it's properties are indirectly characterised by the original predicate. We generally employ hypostatic abstraction in this paper without further comment.
\end{definition}

\begin{remark}
    The following two examples illustrate hypostatic abstraction. \begin{enumerate*}[label=(\roman*)] \item The non-symmetric binary predicate `is less than', with property `less than', is transformed to the (non-symmetric) ternary predicate `has less-than-ness', with hypostatic object `less-than-ness': so, the proposition `2 is less than 3' is transformed to the proposition `2 has less-than-ness in relation to 3', or `The relationship between 2 and 3 has the property of less-than-ness'.
    \item The symmetric ternary predicate `are collinear', with property `collinear', is transformed to the (non-symmetric) quaternary predicate `have collinearity', with hypostatic object `collinearity': so, the proposition `$x$, $y$, and $z$ are collinear' is transformed to the proposition `$x$, $y$, and $z$ share the property of collinearity'.
    \end{enumerate*}

    \medskip

    Hypostatic abstraction is foundational for objectifying aspects of experience, including conceptualisation, generalisation, and in scientific inquiry. The operation does not directly result in explanation, rather providing a means toward explanation by transforming relations into objects to which properties can then be predicated, thereby connecting first-order to second-order logic. Hypostatic abstraction requires that we accept the existence of the hypostatic object, which was unnecessary before the operation, so represents either physical reality or physical (not logical) impossibility. Even in the case of physical impossibility the hypostatic object can be of analytical value. It is important to note that logical consistency should be maintained throughout the process and subsequent analysis to arrive at conclusions relevant for the original relation.
\end{remark}

\subsection*{(Hyper)graphs}
Systems biology characterises biological systems primarily in terms of their \emph{structure}: the way a system is organised in space and time by the interactions of specific combinations of its parts. This can be represented by the relational structure of graphs or hypergraphs, whereby a set of objects represents the parts of the system and relations on those objects represent properties among the parts. Graphs are sufficient when binary relations provide an appropriate level of detail, whereas hypergraphs are required for a direct representation of multidimensional structure based on the notion of a hyperedge. Since the number of instances of a given hyperedge -- its multiplicity -- is also important, we work with \emph{multihypergraphs} throughout.

\begin{definition}[Graph, directed graph, bipartite graph, directed bipartite graph]
    An \emph{undirected graph}, or simply a \emph{graph}, $G \defeq (V,E)$ is a pair of sets where the \emph{edge set} $E$ is a set of $2$-subsets of the nonempty finite \emph{vertex set} $V$. The elements of $V$ are \emph{vertices} and the elements of $E$ are \emph{edges}. A \emph{directed graph} $D \defeq (V,F)$ is a pair of sets where $V$ is the nonempty finite vertex set and $F \subseteq V \times V$ is the \emph{arc set}. The elements of $F$ are \emph{directed edges} or \emph{arcs}. By definition there are no multiple edges in undirected/directed graphs. For two vertices $v, w \in V$ in an undirected (resp. directed) graph we say that $v$ is \emph{incident on} $w$ if $\{v,w\} \in E$ (resp. if $(v,w) \in F$). This definition of incidence carries to further relational structures defined below.

    A \emph{bipartite graph} is a triple $G \defeq (V,W,E)$ where $V$ and $W$ are two disjoint sets of \emph{vertices} referred to as the \emph{first part} and the \emph{second part} of the graph, respectively, and $E \subseteq \big\{\, \{v,w\} \mid \text{$(v,w) \in V \times W$} \,\big\}$ is a set of \emph{edges}. A \emph{directed bipartite graph} $D$ is a triple $D \defeq (V,W,F)$ with first part $V$, second part $W$, and \emph{arc set} $F \subseteq (V \times W) \sqcup (W \times V)$. An \emph{arc weight function} $\omega$ for $D$ satisfies $\omega \colon F \to \AA$ where $\AA \subseteq \RR$ is nonempty and closed under addition, and a \emph{second-part vertex weight function} $\delta$ for $D$ satisfies $\delta \colon W \to \VV$ where $\VV \subseteq \RR$ is nonempty. If $D$ has arc weight function $\omega$ and second-part vertex weight function $\delta$ then we write $D \defeq (V,W,F,\omega,\delta)$, or $D \defeq (V,W,F,(\omega,\AA),(\delta,\VV))$ to specify the codomains.
\end{definition}

\begin{definition}[Hypergraph, vertex labelling, multihypergraph]
    A \emph{hypergraph} $X \defeq (V,E)$ consists of a nonempty finite \emph{vertex set} $V$ and a \emph{hyperedge set} $E \subseteq \powerset(V)$ of \emph{hyperedges}, noting we allow the empty hyperedge $\emptyset$. Note that a graph is a hypergraph where the hyperedge set consists of 2-subsets of $V \!$. A \emph{vertex labelling} of $X$ is an injection from $V$ to a set of labels, and in this work we always assume that hypergraphs are vertex-labelled (henceforth simply labelled) so we can identify specific vertices.
    
    A \emph{multihypergraph} $M \defeq (V,E,\mu_E)$ is a hypergraph $(V,E)$ together with a hyperedge multiplicity function $\mu_E \colon E \to \NN_+$, and we say that $e \in E$ is a \emph{multihyperedge} with multiplicity $\mu_E(e)$. We assume that the hyperedges within a multihyperedge have their own identity, so the multiplicity of a multihyperedge serves only to indicate the number of hyperedges on the same subset of vertices. A general hyperedge in $M$ is therefore any hyperedge of a multihyperedge, with a total of $\sum_{e \in E} \mu_E (e)$ hyperedges in $M$. Depending on context we will refer to $E$ as either the set of multihyperedges or the set of (all) hyperedges: if required for clarity we will refer to $E$ as the multihyperedge set and the set of all hyperedges in $M$ as $(E,\mu_E)$.
\end{definition}

\begin{remark}
    Given a multihypergraph $(V,E,\mu_E)$, the set of hyperedges $(E,\mu_E)$ can represent different relational structures on $V$, specifically: \begin{enumerate*}[label={(\arabic*)}] \item a single multigrade relation that is unordered, homogeneous, and finitary;\label{rel_1} \item a (possibly nonunique) disjoint union of distinct relations each of the form in \ref*{rel_1}, whereby $E$ represents more than one form of relationship among the elements of $V$. Note that a single relation can be decomposed into a disjoint union of relations representing the same form of relationship: in particular a single relation can be decomposed into a disjoint union of unigrade relations where each contains a single element corresponding to a hyperedge (each of these unigrade relations is therefore a singleton). For a given multihypergraph these relational structures are all logically equivalent, so represent the same information.
    \end{enumerate*}
\end{remark}

\subsection*{Relationship between multihypergraphs and bipartite graphs}
From a mathematical perspective, multihypergraphs and bipartite graphs have a canonical relationship: every multihypergraph has an associated incidence (or Levi) graph, which is bipartite, and every bipartite graph is the incidence graph of some multihypergraph. The mathematical details are provided in the Appendix (Section~\ref{app:hypergraphs-bipartite-graphs}, Corollary~\ref{corr:multihypergraphs-bipartite}). This canonical relationship provides the symbolic connection between the two abstract structures, however for any application we also require a foundational logical connection which can be obtained through hypostatic abstraction.

So let $(V,E,\mu_E)$ be a multihypergraph where we assume (without loss of generality) that the relation structure is a disjoint union of (unigrade) singleton relations (each containing a single hyperedge). For a hyperedge $e$ let $R_e \defeq \{e\}$ be the associated relation, then under hypostatic abstraction the relation $R_e$ is transformed to the singleton relation $R^{\ast}_e \defeq \{e \cup \{h(R_e)\}\}$ where $h(R_e)$ is the hypostatic object which possesses the properties attributed to the vertices in $e$ by $R_e$. The relation $R^{\ast}_e$ attributes the properties possessed by $h(R_e)$ to each vertex in $e$, and since $R_e$ is an unordered relation it reduces to the binary relation $R^{\ast\ast}_e \defeq \{\, (v,h(R_e)) \mid v \in e \,\}$ with domain $e \times \{h(R_e)\}$, where $(v,h(R_e)) \in R^{\ast\ast}_e$ means that $v$ has the properties of $h(R_e)$. Importantly, the relation $R \defeq \bigcup_{e \in (E,\mu_E)} R^{\ast\ast}_e$ preserves all of the information in $R_e$, maintaining the logical connection. Note that $R$ corresponds to the (bipartite) incidence graph of the multihypergraph, where one part consists of vertices in $V$ and the other part consists of hypostatised hyperedges which have the properties predicated by the original hyperedges. While the hyperedges can now be studied as objects, it is essential to keep in mind that the physical impossibility of the hypostatic objects can lead to conclusions that are logically erroneous in applications. Moreover, if the original properties of the hypostatic objects are modified or there are logical inconsistencies among the original properties and any new attributed properties then any conclusions may be irrelevant for the original system.

This analysis shows that multihypergraphs and bipartite graphs will provide equivalent representations as long as the logical connections (including those from the application context) are preserved. The bipartite graph can be regarded as a `flattened' multihypergraph, where the multidimensional structure becomes more latent -- it is distributed across neighbourhoods rather than represented directly -- and there is typically a substantial increase in edge density from a multihypergraph to the corresponding bipartite graph. Hypergraphs allow a more faithful and direct representation of higher-order relations than graphs, however the associated bipartite graphs can facilitate analysis of the hypergraphs and therefore the represented systems.

It is important to note that the relation reduction from multihypergraphs to bipartite graphs is straightforward due to the unordered nature of the relations in the multihypergraph. If these relations were ordered then a reduction to binary relations may not be logically possible.

\subsection*{Chemical reaction networks and stochastic Petri nets}
Our definitions of chemical reaction networks and stochastic Petri nets are standard.

\begin{definition}[Chemical reaction network, kinetics, kinetic system]
    A \emph{chemical reaction network} (CRN) is a triple $(\S,\C,\R)$ of finite sets where: the \emph{species set} $\S$ consists of molecular \emph{species} that are capable of reacting chemically; the \emph{complex set} $\C$ consists of positive partial multisets with underlying set $\S$, where each element of $\C$ is a \emph{complex} and the multiplicity of a species in a complex is its \emph{stoichiometric coefficient}; and the \emph{reaction set} $\R \subseteq \C \times \C$ is an irreflexive binary relation on $\C$ which specifies a set of \emph{reactions} between complexes. Each reaction is written $(c,c^{\prime})$, or $c \to c^{\prime}$, with $c$ the \emph{reactant complex} and $c^{\prime}$ the \emph{product complex}. Species in the complexes $c$ and $c^{\prime}$ are called \emph{reactants} and \emph{products}, respectively. We say that $(\S,\C,\R)$ is \emph{connected} when $\C = \bigcup_{(c,c^{\prime}) \in \R} \{c,c^{\prime}\}$, that is there are no \emph{isolated} complexes. A \emph{kinetics} $k$ for a CRN $(\S,\C,\R)$ is a map $k \in \{\R \to \RR_{\ge 0}\}$ such that $k((c,c^{\prime}))$ is the \emph{stochastic rate constant} of the reaction $(c,c^{\prime}) \in \R$. The quadruple $(\S,\C,\R,k)$ is a (chemical) \emph{kinetic system}.
\end{definition}

Complexes in chemical reaction networks are typically represented as ordered formal linear combinations, unordered formal linear combinations, or multisets of chemical species. We use multisets for complexes, and for completeness show the relationships between the three representational forms in Appending~\ref{app:complexes}.

\begin{remark}\label{rem:CRNZeros}
   CRNs are sometimes defined with stoichiometries in $\NN$, leading to an understanding of complexes as vector spaces. While mathematically convenient, this introduces an interpretational ambiguity between two distinct cases of absence (Definition~\ref{def:absence}): zero stoichiometry can indicate either the absence of a species from a complex or its presence with zero instances. A related ambiguity arises with \emph{net change} in CRN dynamics, where the species-wise difference in stoichiometry between product and reactant complexes can be zero for species that nonetheless participate in the reaction on both sides. We explore below how to resolve these ambiguities in interpreting data defined at the formalism level.
\end{remark}

\begin{definition}[Petri net, stochastic Petri net]
    A \emph{Petri net} is a 4-tuple $N \defeq (P,T,F,\omega)$ where $P \defeq \{P_1, \ldots, P_m\}$ is a finite set of \emph{places}; $T \defeq \{T_1, \ldots, T_n\}$ is a finite set of \emph{transitions} with $P \cap T = \emptyset$; $F \subseteq (P \times T) \cup (T \times P)$ is the \emph{flow relation} defining a set of directed \emph{arcs}; and $\omega \colon F \to \AA$ is a \emph{weight function}, where $\AA \in \{\NN,\NN_+\}$. We make the standard assumption that there are no isolated places in a Petri net. Note that $(P,T,F)$ is a directed bipartite graph with labelled vertices. A \emph{stochastic Petri net} (SPN) is a 5-tuple $N \defeq (P,T,F,\omega,r)$ where $(P,T,F,\omega)$ is a Petri net, and $r \colon T \to \RR_{\ge 0}$ is a function which assigns a \emph{stochastic rate constant} to each transition.
\end{definition}
    
\begin{remark}
    We define the weight function as taking values in $\NN$ for generality, and consider values in $\NN_+$ when relating SPNs to CRNs, in which complexes are specifically defined as positive partial multisets.
    
    We will assume throughout that CRNs are connected, since an isolated complex is redundant with respect to the reactions in the CRN. Moreover, an isolated complex has no representation in SPNs, since relations are defined on places directly and not via the intermediary concept of a complex. The connectedness of CRNs corresponds in part to the absence of isolated places in SPNs.
\end{remark}

\section{Systems hypergraphs}

Systems are primarily understood through their \emph{structure}: the interactions between parts that constitute, maintain, and transform a system over time. The basis of a representation is therefore a naming of parts, together with incidence data recording which parts interact. Of course, parts are systems themselves and choosing a label is simply a choice of reference point relative to a focus of analysis. This reduction is in fact necessary for analysis, but importantly is not final. It is in fact the purpose of a representation to identify parts on the basis of the existence or relevance of some relations on those parts.

Hypergraphs describe incidence data in terms of unordered relations on a set of vertices. In that setting the duality between part and system finds an analogue in the ones between object and relation, vertex and hyperedge. Hypergraphs themselves are, however, not sufficient to represent real systems in enough detail for modelling purposes. For this we need to augment hypergraphs with additional \emph{attributes} -- data associated with objects such as vertices and hypostatised hyperedges -- that encode the particular properties of the system being modelled.

Here we define \emph{systems hypergraphs}, a general class of augmented hypergraphs. We develop the framework in general terms, while illustrating key concepts through the special case of systems of \emph{abstract molecular transition (AMT)} type, which cover the kinds of systems traditionally represented by CRNs and SPNs.

The foundation of our framework is a labelled multihypergraph $(V,E,\mu_E)$. The vertex set $V$ represents the \emph{ground set} of the system constituted of elements treated as atomic, i.e. whose internal relational structure is considered fixed to define a minimum structural level at which relations can be defined. The hyperedge set $E$ then represents relations among any number of elements, with each hyperedge indicating a particular connection and distinct hyperedges potentially representing distinct types of connection.

We then augment this basic structure with \emph{attributes} to provide additional data beyond incidence. Attributes are associated to multihypergraphs via \emph{attribute functions}, with the scope of the attribute -- the specific substructures it annotates -- specified by \emph{domain functions}. Crucially, the way attributes are associated to substructures, and how attributes may depend on other attributes, provides important structural information. Attributes and their dependency relations are therefore key in understanding systems, allowing for additional data beyond the traditional relational structure defined by graphs or hypergraphs. This culminates in our definition of a systems hypergraph as a framework to collect this extended structural data in an integrated way.

\begin{definition}[Domain function]
    Let $X \defeq (V,E)$ be a hypergraph. A \emph{domain function} $\lambda$ for $X$ is a set-valued function with $\Dom(\lambda) \defeq E$ such that, for each $e \in E$, the set $\lambda(e) \subseteq \powerset(e) \sqcup \{\{e\}\}$ contains substructures of $e$ from either the vertex level (so subsets of vertices) or from the hyperedge level (so $e$ considered as a single object).
\end{definition}

We define domain functions on the hyperedge set since we want to associate data with objects that have a relationship with other objects: isolated objects have no interaction with the rest of the system. Our definition of a domain function can be extended to any level of structure (hypostatised if necessary), such as vertices, or subhypergraphs (including isolated vertices) of a given hypergraph or families of hypergraphs.

We only consider relations involving at least two objects. We reserve the singleton notation $\{x\}$ to denote the objectification (via hypostatic abstraction) of a relation $x$, so that $\{x\}$ is an object. For notational convenience we allow for the objectification of an object $y$, in which case we can logically identify $y$ and $\{y\}$.

\begin{example}\label{ex:domain}[Domain functions for AMT-type systems]
    For a hypergraph $X \defeq (V,E)$ the domain functions relevant for AMT-type systems are:
    \begin{enumerate*}[label={(\arabic*)}]
        \item The \emph{singleton function} from $E$ into $\powerset(E)$ where $e \mapsto \{e\}$ for $e \in E$.
        \item The \emph{powerset function} from $E$ into $\lowerset E$ where $e \mapsto \powerset(e)$ for $e \in E$.
        \item The \emph{two-fold nested singleton function} from $E$ into $\powerset(\powerset(E))$ where $e \mapsto \{\{e\}\}$ for $e \in E$, which we call the \emph{objectification function} for simplicity.
    \end{enumerate*}
\end{example}

The appropriate domain function depends on the substructure to which the type of attribute is associated: the singleton function when the same type of attribute is assigned to all vertices in a hyperedge; the powerset function when any subset of a hyperedge has a particular type of attribute assigned to all vertices in the subset; and the objectification function when the attribute is assigned to a hypostatised hyperedge.

\begin{remark}\label{rem:domains-context}
    Domain functions enable the annotation of each individual vertex in a hyperedge, based on incidence pairs in $V \times E$. In particular, the powerset function enables the association of different types of attributes to distinct subsets of vertices in a given hyperedge. The powerset function acts at an intermediate level of structure between the singleton function, which acts uniformly on all vertices of a hyperedge, and the objectification function, which acts on a hypostatised hyperedge.

    Domain functions ensure that multidimensional substructures are directly available for the precise association of properties in relation to context. Importantly, domain functions help to maintain the integrity of multidimensional data by preventing information distortion through loss of context.
\end{remark}

\begin{definition}[Attribute set, attribute function]
    Let $X \defeq (V,E)$ be a hypergraph, and let $\lambda$ be a domain function for $X$. An \emph{attribute set} $\AA$ is a nonempty set of attributes. For $e \in E$, each $\pi \in \{ \lambda(e) \rightrightarrows \AA \}$ is an \emph{$\AA$-attribute function}, or simply \emph{attribute function}, which associates attributes in $\AA$ to the substructures in $\lambda(e)$.
\end{definition}

Attributes represent any data needed to interpret the meaning of a structure in practice -- of particular relevance here since structures are used to represent physical systems. Moreover, the specific value of an attribute on a given structure can be understood as a \emph{labelling} of that structure to distinguish structurally equivalent objects with distinct properties.

\begin{example}\label{ex:attribute}[Attribute functions for AMT-type systems]
    Suppose $X \defeq (V,E)$ is a hypergraph, $\lambda$ is a domain function for $X$, and $\AA$ is an attribute set. Continuing from Example~\ref{ex:domain}, $\{ \lambda(e) \rightrightarrows \AA \}$ becomes:
    \begin{enumerate*}[label={(\arabic*)}]
        \item $\{ e \rightarrow \AA \}$ when $\lambda$ is the singleton function;
        \item $\{ e \rightharpoonup \AA \}$ when $\lambda$ is the powerset function;
        \item and $\{ \{e\} \rightarrow \AA \}$ when $\lambda$ is the objectification function.
    \end{enumerate*}
\end{example}

The multiplicity of hyperedges in a multihypergraph requires that attributes are assigned to substructures in a combinatorial manner. The corresponding multiplicities of attribute functions motivates our notion of an $n$-distribution function.

\begin{definition}[$n$-distribution function, induced distribution function, independence]\label{def:nDistFn}
Suppose $X \defeq (V,E,\mu_E)$ is a multihypergraph, $n \in \NN_+$, $(\lambda_i)_{i=1}^n$ is a sequence of domain functions for $X$, and $(\AA_i)_{i=1}^n$ is a sequence of attribute sets. An \emph{$n$-distribution function}, or simply a \emph{distribution function}, for $X$ with respect to $(\lambda_i)_{i=1}^n$ and $(\AA_i)_{i=1}^n$ is given by $\widehat{\pi} \colon E \to \bigcup_{e \in E} \big\{ \bigtimes_{i=1}^n \{ \lambda_i(e) \rightrightarrows \AA_i \} \rightarrow \NN \big\}$ such that for $e \in E$, denoting $\widehat{\pi}(e)$ by $\widehat{\pi}_e$, we have $\widehat{\pi}_e \colon \bigtimes_{i=1}^n \{ \lambda_i(e) \rightrightarrows \AA_i \} \rightarrow \NN$ and $\sum_{\pi \in \bigtimes_{i=1}^n \{ \lambda_i(e) \rightrightarrows \AA_i \}} \widehat{\pi}_e 
(\pi) = \mu_E (e)$. If the domain functions are clear from context then we may specify the attribute sets $(\AA_i)_{i=1}^n$ associated with the $n$-distribution function $\widehat{\pi}$ by writing $(\widehat{\pi},(\AA_i)_{i=1}^n)$.

If $P \subseteq [n]$ is nonempty with complement $P^{\prime}$ (both with induced linear orders) then the function $\widehat{\psi} \colon E \to \bigcup_{e \in E} \big\{ \bigtimes_{i \in P} \{ \lambda_i(e) \rightrightarrows \AA_i \} \rightarrow \NN \big\}$ such that $\widehat{\psi}_e \big((\pi_i)_{i \in P}\big) \defeq \sum_{(\pi_i)_{i \in P^{\prime}} \in \bigtimes_{i \in P^{\prime}} \{ \lambda_i(e) \rightrightarrows \AA_i \}} \widehat{\pi}_e \big((\pi_i)_{i=1}^n\big)$ for $e \in E$ and $(\pi_i)_{i \in P} \in \bigtimes_{i \in P} \{ \lambda_i(e) \rightrightarrows \AA_i \}$ is a $\left| P \right|$-distribution function, referred to as the \emph{$P$-induced $\left| P \right|$-distribution function}.

Suppose $m$, $n \in \NN_+$, $\widehat{\pi} \colon E \to \bigcup_{e \in E} \big\{ \bigtimes_{i=1}^m \{ \lambda_i(e) \rightrightarrows \AA_i \} \rightarrow \NN \big\}$ is an $m$-distribution function for $X$ with respect to the sequence of domain functions $(\lambda_i)_{i=1}^m$ and the sequence of attribute sets $(\AA_i)_{i=1}^m$, and $\widehat{\pi}^{\prime} \colon E \to \bigcup_{e \in E} \big\{ \bigtimes_{i=1}^n \{ \lambda^{\prime}_i(e) \rightrightarrows \AA^{\prime}_i \} \rightarrow \NN \big\}$ is an $n$-distribution function for $X$ with respect to the sequence of domain functions $(\lambda^{\prime}_i)_{i=1}^n$ and the sequence of attribute sets $(\AA^{\prime}_i)_{i=1}^n$. Then $\widehat{\pi}$ is \emph{independent} from $\widehat{\pi}^{\prime}$ if and only if for every $e \in E$ it holds that for $\pi \in \widehat{\pi}_e^{-1}(\NN_+)$, $\pi^{\prime} \in {\widehat{\pi}_e}^{\prime \, -1}(\NN_+)$, $j \in [m]$, and $k \in [n]$ the $\AA_j$-attribute function $\pi_j$ is independent from the $\AA^{\prime}_k$-attribute function $\pi^{\prime}_k$; otherwise $\widehat{\pi}$ is \emph{dependent} on $\widehat{\pi}^{\prime}$. If a distribution function $\widehat{\pi}$ is dependent on the $r \in \NN_+$ distribution functions $\{ \widehat{\pi}_1, \widehat{\pi}_2, \ldots, \widehat{\pi}_r \}$ then we denote the dependence as $\widehat{\pi} = \widehat{\pi}[\widehat{\pi}_1;\widehat{\pi}_2; \ldots;\widehat{\pi}_r]$.
\end{definition}

\begin{remark}\label{rem:multiplicity-matching}
    With reference to Definition~\ref{def:nDistFn}, $\widehat{\pi}_e \colon \bigtimes_{i=1}^n \{ \lambda_i(e) \rightrightarrows \AA_i \} \rightarrow \NN$ is a multiplicity function on the set of $n$-tuples $\bigtimes_{i=1}^n \{ \lambda_i(e) \rightrightarrows \AA_i \}$ of attribute functions which are defined on certain substructures of the hyperedge $e$. The condition $\sum_{\pi \in \bigtimes_{i=1}^n \{ \lambda_i(e) \rightrightarrows \AA_i \}} \widehat{\pi}_e (\pi) = \mu_E (e)$ ensures that the total number of $n$-tuples of attribute functions equals the multiplicity of $e$.
\end{remark}

\begin{example}[$n$-distribution functions for AMT-type systems]\label{ex:DistFun}\leavevmode
    Continuing from Example~\ref{ex:attribute}, let $X \defeq (V,E,\mu_E)$ be a multihypergraph. The first three distribution functions are the \emph{orientation}, \emph{multiset}, and \emph{hyperedge label functions}, which are 1-distribution functions, pairwise independent, and provide a multihypergraph with the properties of vertex orientation, vertex multiplicity, and hyperedge label, respectively. The fourth distribution function is the \emph{negative/positive-orientation multiset function}, which is a 2-distribution function, is dependent on the orientation and multiset functions, and is used to partition vertex multiplicity into negative and positive parts according to the specific orientation of a vertex in a hyperedge.
    \begin{enumerate}[label={(\arabic*)},topsep=3pt,itemsep=1pt,leftmargin=16pt]
        \item \emph{Multihypergraph orientation function, $\widehat{\sigma}$}: Let the domain function $\lambda$ for $X$ be the singleton function and the attribute set be the orientation set $\I$. A \emph{multihypergraph orientation function} $\widehat{\sigma} \colon E \to \bigcup_{e \in E} \{ \{ e \rightarrow \I \} \rightarrow \NN \}$ for $X$ with respect to $\lambda$ and $\I$ satisfies the conditions $\widehat{\sigma}_e \colon \{ e \rightarrow \I \} \rightarrow \NN$ and $\sum_{\sigma \in \{ e \rightarrow \I \}} \widehat{\sigma}_e (\sigma) = \mu_E (e)$ for $e \in E$. Note that if $e = \emptyset$ then $\{ e \rightarrow \I \}$ consists of only the empty function to $\I$. A multihypergraph with an orientation function $(V,E,\mu_E,\widehat{\sigma})$ is an \emph{oriented multihypergraph}. \label{ex:Orientation}
        \item \emph{Multihypergraph multiset function, $\widehat{\mu}$}: Let the domain function $\lambda$ for $X$ be the singleton function and the attribute set be $\AA \subseteq \RR$ which is nonempty and closed under addition. A \emph{multihypergraph multiset function} $\widehat{\mu} \colon E \to \bigcup_{e \in E} \{ \{e \rightarrow \AA\} \rightarrow \NN \}$ for $X$ with respect to $\lambda$ and $\AA$ satisfies the conditions $\widehat{\mu}_e \colon \{e \rightarrow \AA\} \to \NN$ and $\sum_{\mu \in \{e \rightarrow \AA\}} \widehat{\mu}_e(\mu) = \mu_E (e)$ for $e \in E$. Note that if $e = \emptyset$ then $\{e \rightarrow \AA\}$ consists of only the empty function to $\AA$. A multihypergraph with a multiset function $(V,E,\mu_E,\widehat{\mu})$ is a \emph{multiset multihypergraph}.
        \item \emph{Multihypergraph hyperedge label function, $\widehat{\rho}$}: Let the domain function $\lambda$ for $X$ be the objectification function and the attribute set be $\VV \subseteq \RR$ which is nonempty. A \emph{multihypergraph hyperedge label function} $\widehat{\rho} \colon E \to \bigcup_{e \in E} \{ \{ \{e\} \rightarrow \VV \} \rightarrow \NN \}$ for $X$ with respect to $\lambda$ and $\VV$ satisfies the conditions $\widehat{\rho}_e \colon \{ \{e\} \rightarrow \VV \} \to \NN$ and $\sum_{\rho \in \{ \{e\} \rightarrow \VV \}} \widehat{\rho}_e(\rho) = \mu_E (e)$ for $e \in E$. A multihypergraph with a hyperedge label function $(V,E,\mu_E,\widehat{\rho})$ is a \emph{hyperedge-labelled multihypergraph}.
        \item \emph{Multihypergraph negative/positive-oriented multiset function, $\widehat{\mu}^{\,\Minus / \Plus}$}: Let the domain function $\lambda$ for $X$ be the powerset function and the attribute set be $\AA \subseteq \RR$ which is nonempty and closed under addition. A \emph{multihypergraph negative/positive-oriented multiset function} $\widehat{\mu}^{\,\Minus / \Plus} \colon E \to \bigcup_{e \in E} \{ \{ e \rightharpoonup \AA \} \times \{ e \rightharpoonup \AA \} \rightarrow \NN \}$ for $X$ with respect to $(\lambda,\lambda)$ and $(\AA,\AA)$ satisfies the conditions $\widehat{\mu}^{\,\Minus / \Plus}_e \colon \{ e \rightharpoonup \AA \} \times \{ e \rightharpoonup \AA \} \to \NN$ and $\sum_{\text{$(\mu^\Minus,\mu^\Plus) \in \{ e \rightharpoonup \AA \} \times \{ e \rightharpoonup \AA$} \}} \widehat{\mu}^{\,\Minus / \Plus}_e \big((\mu^\Minus,\mu^\Plus)\big) = \mu_E (e)$ for $e \in E$, and $\widehat{\mu}^{\,\Minus / \Plus} = \widehat{\mu}^{\,\Minus / \Plus}[\widehat{\mu};\widehat{\sigma}]$. Note that if $e = \emptyset$ then $\{ e \rightharpoonup \AA \}$ consists of only the empty function to $\AA$.
        
        The $\{1\}$-induced 1-distribution function $\widehat{\mu}^{\,\Minus} \colon E \to \bigcup_{e \in E} \{ \{ e \rightharpoonup \AA \} \rightarrow \NN \}$ for $X$ with respect to $\lambda$ and $\AA$ such that $\widehat{\mu}^{\,\Minus}_e (\mu^{\Minus}) \defeq \sum_{\mu^{\Plus} \in  \{ e \rightharpoonup\AA \}} \widehat{\mu}^{\,\Minus / \Plus}_e \big((\mu^{\Minus},\mu^{\Plus})\big)$, for $e \in E$ and $\mu^{\Minus} \in \{ e \rightharpoonup \AA \}$ is the \emph{multihypergraph negative-oriented multiset function}. Similarly, the $\{1\}$-induced 1-distribution function $\widehat{\mu}^{\,\Plus} \colon E \to \bigcup_{e \in E} \{ \{ e \rightharpoonup \AA \} \rightarrow \NN \}$ for $X$ with respect to $\lambda$ and $\AA$ such that $\widehat{\mu}^{\,\Plus}_e (\mu^{\Plus}) \defeq \sum_{\mu^{\Minus} \in  \{ e \rightharpoonup \AA \}} \widehat{\mu}^{\,\Minus / \Plus}_e \big((\mu^{\Minus},\mu^{\Plus})\big)$, for $e \in E$ and $\mu^{\Plus} \in \{ e \rightharpoonup \AA \}$ is the \emph{multihypergraph positive-oriented multiset function}. \label{ex:NegPos-Oriented-Multiset}
    \end{enumerate}
\end{example}

The conditions that distribution functions must satisfy specify the precise way in which attributes may be dependent on certain system properties. Following Remark~\ref{rem:multiplicity-matching}, matching hyperedge multiplicity data is one condition that applies to all distribution functions. More complex attribute dependency patterns result in additional conditions: in Example~\ref{ex:DistFun}~\ref{ex:NegPos-Oriented-Multiset}, for instance, the particular instances of $\widehat{\mu}$ and $\widehat{\sigma}$ determine $\widehat{\mu}^{\,\Minus / \Plus}$, ensuring that the partitioning into negative and positive multiplicities is consistent with both vertex multiplicities and orientations.

A \emph{distribution hierarchy} formalises the dependencies between the relevant distribution functions as a graded poset.

\begin{definition}[Distribution hierarchy]\label{def:DistHierarchy}
Let $X \defeq (V,E,\mu_E)$ be a multihypergraph. A \emph{distribution hierarchy} for $X$ is a graded poset $Q \defeq (S,<,\rnk)$ where $S$ is a finite set of distribution functions for $X$, $<$ is a strict partial order on $S$, and $\rnk \colon S \to \NN$ is a rank function, such that:
    \begin{enumerate*}[label={(\arabic*)}]
        \item $\widehat{\pi} < \widehat{\pi}^{\prime}$ if and only if $\widehat{\pi}^{\prime}$ depends on $\widehat{\pi}$, for all $\widehat{\pi}$, $\widehat{\pi}^{\prime} \in S$;
        \item If $\widehat{\pi}^{\prime} \in \rnk^{-1}(i+1)$ then there exists $\widehat{\pi} \in \rnk^{-1}(i)$ such that $\widehat{\pi} < \widehat{\pi}^{\prime}$, for all $i \in \NN$;
        \item All minimal elements in $(S,<)$ have rank 0.
    \end{enumerate*}
\end{definition}

\begin{remark}
    With reference to Definition~\ref{def:DistHierarchy}, and for $i \in \NN$: the distribution functions in $\rnk^{-1}(i)$ all have rank $i$ so are pairwise incomparable, hence pairwise independent; if $\widehat{\pi} \in \rnk^{-1}(i)$ and $\widehat{\pi}^{\prime} \in \rnk^{-1}(i+1)$ then $\rnk(\widehat{\pi}^{\prime}) = \rnk(\widehat{\pi}) + 1$, so $\widehat{\pi}^{\prime} \nless \widehat{\pi}$, hence $\widehat{\pi}$ is independent from $\widehat{\pi}^{\prime}$; there exists a unique $n \in \NN$ such that $\rnk^{-1}(j) \ne \emptyset$ if and only if $j \in \NN$ and $j \le n$.
\end{remark}

\begin{remark}
    While the distribution hierarchy may be a nonlinear poset, which efficiently encodes the necessary information, there may be occasion where a linear ordering of the distribution functions is required. Proposition~\ref{prop:dist-hier} in Appendix~\ref{app:dist-hier} shows that such a linear order always exists, and is in fact a well order. The well ordering of the distribution hierarchy is, of course, not necessarily unique, however the linear ordering maintains the original dependencies.
\end{remark}

\begin{remark}\label{rem:Multihyperedges}
    While our foundational structure is a multihypergraph, we could instead begin with a hypergraph and arrive at a multihypergraph by associating the attribute of multiplicity to hyperedges. These two approaches are logically equivalent since, in this case, objectification does not result in information distortion as hyperedge multiplicity data is neither contextual (Remark~\ref{rem:domains-context}) nor dependent on other attributes. Moreover, we could begin with just a ground set of objects and then associate hyperedges as attributes, another logically equivalent approach. In this work we begin with a multihypergraph since it is the structure underlying our applications, and thereby allows a simplification of the attribute structure.
\end{remark}

\begin{example}[Distribution hierarchy for AMT-type systems]\label{ex:DistHierarchy}
    Let $X \defeq (V,E,\mu_E)$ be a multihypergraph. Continuing from Example~\ref{ex:DistFun}, the corresponding distribution hierarchy for $X$ is $Q \defeq (S,<,\rnk)$, where $S \defeq \{(\widehat{\sigma},\I),(\widehat{\mu},\AA),(\widehat{\rho},\VV),\allowbreak(\widehat{\mu}^{\,\Minus / \Plus},\AA)\}$, the comparable distribution functions are $\widehat{\sigma} < \widehat{\mu}^{\,\Minus / \Plus}$ and $\widehat{\mu} < \widehat{\mu}^{\,\Minus / \Plus}$, and $\rnk(\widehat{\sigma}) = \rnk(\widehat{\mu}) = \rnk(\widehat{\rho}) = 0$ and $\rnk(\widehat{\mu}^{\,\Minus / \Plus}) = 1$.
\end{example}

Distribution functions specify the multiplicity (or distribution) of attribute functions for each hyperedge, based on the attribute functions that are possible given the domain functions and attribute sets. A distribution hierarchy represents the dependency structure among the distribution functions of interest. To specify a particular attribute structure, out of all possible combinations, that is relevant for each instance of each hyperedge we introduce the notion of an \emph{attribute hierarchy}, which is induced by the distribution hierarchy. Finally, we define the \emph{association function} which specifies all of the required associations between the various attributes based on the attribute hierarchies.

\begin{definition}[Attribute hierarchy]\label{def:AttHierarchy}
    Let $X \defeq (V,E,\mu_E)$ be a multihypergraph, and let $Q \defeq (S,<,\rnk)$ be a distribution hierarchy for $X$. For each $e \in E$ denote by $Q_e$ the set of all graded posets such that $(S_e,<_e,\rnk_e) \in Q_e$ implies:
        \begin{enumerate*}[label={(\arabic*)}]
            \item The set $S_e$ is an element of the unordered Cartesian product $\prod_{\hat{\pi} \in S} \Dom(\widehat{\pi}_e)$;
            \item $\pi <_e \pi^{\prime}$ if and only if $\widehat{\pi} < \widehat{\pi}^{\prime}$, for all $\pi$, $\pi^{\prime} \in S_e$ where $\widehat{\pi}$, $\widehat{\pi}^{\prime} \in S$, $\pi \in \Dom(\widehat{\pi}_e)$, and $\pi^{\prime} \in \Dom(\widehat{\pi}^{\prime}_e)$;
            \item $\rnk_e(\pi) = \rnk(\widehat{\pi})$, for all $\pi \in S_e$ where $\widehat{\pi} \in S$ and $\pi \in \Dom(\widehat{\pi}_e)$.
        \end{enumerate*}
        Each $(S_e,<_e,\rnk_e) \in Q_e$ is an \emph{attribute hierarchy}.
\end{definition}

\begin{remark}
    With reference to Definition~\ref{def:AttHierarchy}, for each $e \in E$ and $(S_e,<_e,\rnk_e) \in Q_e$:
    \begin{enumerate*}[label={(\arabic*)}]
        \item It follows from the dependency structure of the distribution hierarchy that the factors in the unordered Cartesian product $\prod_{\hat{\pi} \in S} \Dom(\widehat{\pi}_e)$ are pairwise disjoint sets, hence $\left| S_e \right| = \left| S \right|$.
        \item There is only one graded poset in $Q_e$ with the ground set $S_e$, so the notation $(S_e,<_e,\rnk_e)$ is unambiguous.
        \item The map $f \colon S \to S_e$ such that $f(\widehat{\pi}) = \pi$, where $\Dom(\widehat{\pi}_e) \cap S_e = \{ \pi \}$, is a well-defined order isomorphism of the graded posets $Q$ and $(S_e,<_e,\rnk_e)$.
    \end{enumerate*}
\end{remark}

\begin{example}[Attribute hierarchy for AMT-type systems]\label{ex:AttHierarchy}
Continuing from Example~\ref{ex:DistHierarchy}, let $X \defeq (V,E,\mu_E)$ be a multihypergraph and let $Q \defeq (S,<,\rnk)$ be an AMT-type distribution hierarchy for $X$. For $e \in E$ and $(S_e,<_e, \linebreak[4] \rnk_e) \in Q_e$:
    \begin{enumerate*}[label={(\arabic*)}]
        \item $S_e = \{\sigma,\mu,\rho,(\mu^{\Minus},\mu^{\Plus})\}$, where $\sigma \in \Dom(\widehat{\sigma}_e)$, $\mu \in \Dom(\widehat{\mu}_e)$, $\rho \in \Dom(\widehat{\rho}_e)$, and $(\mu^{\Minus},\mu^{\Plus}) \in \Dom(\widehat{\mu}^{\,\Minus / \Plus}_e)$.
        \item $\rnk_e(\sigma) = \rnk_e(\mu) = \rnk_e(\rho) = 0$ and $\rnk_e((\mu^{\Minus},\mu^{\Plus})) = 1$.
        \item The comparable attribute function sequences are $\sigma <_e (\mu^{\Minus},\mu^{\Plus})$ and $\mu <_e (\mu^{\Minus},\mu^{\Plus})$.
    \end{enumerate*}
\end{example}

\begin{definition}[Association function]
Let $X \defeq (V,E,\mu_E)$ be a multihypergraph, and let $Q \defeq (S,<,\rnk)$ be a distribution hierarchy for $X$. An \emph{association function} $\chi \colon E \to \bigcup_{e \in E} \{ Q_e \rightarrow \NN \}$ for $X$ with respect to $Q$ satisfies, for $e \in E$:
    \begin{enumerate*}[label={(\arabic*)}]
        \item
        $\chi_e \colon Q_e \to \NN$.
        \item $\widehat{\pi}_e(\pi) = \sum_{\{\, (S_e,<_e,\rnk_e) \in Q_e \mid \pi \in S_e \,\}} \chi_e((S_e,<_e,\rnk_e))$ for $\widehat{\pi} \in S$ and $\pi \in \Dom(\widehat{\pi}_e)$. 
    \end{enumerate*}
\end{definition}

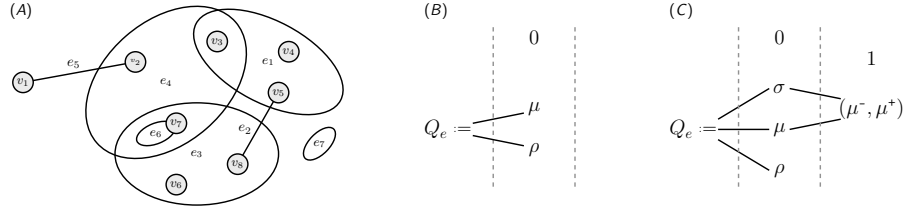
\begin{figure}[!t]
    \centering
    \resizebox{.68\textwidth}{!}{%
\begin{tikzpicture}

\node [font=\LARGE] at (-3.55,10.5) {\normalfont\sffamily\fontsize{11}{9}\selectfont(\emph{A})};
\draw[line width=1pt, rotate around={41:(0,8.75)}] (0,8.75) ellipse (2.25cm and 1.5cm) node {\small $e_4$};
\draw[line width=1pt] (0.75,7) ellipse (2cm and 1.25cm) node {\small $e_3$};
\draw[line width=1pt, rotate around={-28:(2.5,9.25)}] (2.5,9.25) ellipse (2cm and 1cm) node {\small $e_1$};
\draw[line width=1pt, rotate around={22:(-0.25,7.5)}] (-0.25,7.5) ellipse (0.5cm and 0.25cm) node {\small $e_6$};
\draw[line width=1pt] (-0.75,9.25) -- (-3.5,8.75)node[pos=0.55,above]{$e_5$};
\draw[line width=1pt] (2.75,8.5) -- (1.75,6.75)node[pos=0.65,above,xshift=-.17cm]{$e_2$};
\draw[fill={rgb,255:red,235;green,235;blue,235}, line width=1pt] (2.75,8.5) circle (0.25cm) node {\small $v_5$};
\draw[fill={rgb,255:red,235;green,235;blue,235}, line width=1pt] (1.75,6.75) circle (0.25cm) node {\small $v_8$};
\draw[fill={rgb,255:red,235;green,235;blue,235}, line width=1pt] (-3.5,8.75) circle (0.25cm) node {\small $v_1$};
\draw[fill={rgb,255:red,235;green,235;blue,235}, line width=1pt] (-0.75,9.25) circle (0.25cm) node {\tiny $v_2$};
\draw[fill={rgb,255:red,235;green,235;blue,235}, line width=1pt] (0.25,7.75) circle (0.25cm) node {\small $v_7$};
\draw[fill={rgb,255:red,235;green,235;blue,235}, line width=1pt] (1.25,9.75) circle (0.25cm) node {\small $v_3$};
\draw[fill={rgb,255:red,235;green,235;blue,235}, line width=1pt] (3,9.5) circle (0.25cm) node {\small $v_4$};
\draw[fill={rgb,255:red,235;green,235;blue,235}, line width=1pt] (0.25,6.25) circle (0.25cm) node {\small $v_6$};
\draw[line width=1pt, rotate around={-134:(3.75,7.25)}] (3.75,7.25) ellipse (0.5cm and 0.25cm) node {\small $e_7$};

\node [font=\LARGE] at (6.6,10.5) {\normalfont\sffamily\fontsize{11}{9}\selectfont(\emph{B})};
\node[font=\Large] at (6.9,7.6) {$Q_e \defeq$};
\node[font=\Large] at (9,7.1) {$\rho$};
\draw[color={rgb,255:red,146;green,146;blue,146}, line width=1pt, dashed] (8,10.1) -- (8,6.1);
\draw[color={rgb,255:red,146;green,146;blue,146}, line width=1pt, dashed] (10,10.1) -- (10,6.1);
\node[font=\Large] at (9,9.85) {$0$};
\node[font=\Large] at (9,8.1) {$\mu$};
\draw[line width=1pt] (7.5,7.45) -- (8.75,7.2);
\draw[line width=1pt] (8.75,8) -- (7.5,7.75);

\node [font=\LARGE] at (12.6,10.5) {\normalfont\sffamily\fontsize{11}{9}\selectfont(\emph{C})};
\node[font=\Large] at (12.9,7.6) {$Q_e \defeq$};
\node[font=\Large] at (15,8.6) {$\sigma$};
\node[font=\Large] at (15,7.6) {$\mu$};
\node[font=\Large] at (15,6.6) {$\rho$};
\node[font=\Large] at (17.25,8.1) {$(\mu^{\Minus}, \mu^{\Plus})$};
\draw[line width=1pt] (13.5,7.85) -- (14.75,8.6);
\draw[line width=1pt] (13.5,7.6) -- (14.75,7.6);
\draw[line width=1pt] (13.5,7.35) -- (14.75,6.6);
\draw[line width=1pt] (15.25,8.6) -- (16.5,8.35);
\draw[line width=1pt] (16.5,7.85) -- (15.25,7.6);
\draw[color={rgb,255:red,146;green,146;blue,146}, line width=1pt, dashed] (14,10.1) -- (14,6.1);
\draw[color={rgb,255:red,146;green,146;blue,146}, line width=1pt, dashed] (16,10.1) -- (16,6.1);
\node[font=\Large] at (15,9.85) {$0$};
\node[font=\Large] at (17.25,9.35) {$1$};

\end{tikzpicture}}%
    \vspace{-.6cm}
    \caption{Attribute hierarchy for AMT-type systems and extensibility of the attribute system. (\emph{A}) Multihypergraph $X \defeq (V, E, \mu_E)$ representing multidimensional relations $E \defeq \{e_1,\ldots,e_7\}$ on ground set $V \defeq \{v_1,\ldots,v_8\}$, with $\mu_E$ arbitrarily defined. (\emph{B}) Attribute hierarchy $Q_e$ for $X$ with $S_e \defeq (\mu,\rho)$, $<$ given by the solid lines, and $\rnk$ for each distribution function delimited by the dashed lines, defining a labelled multiset multihypergraph. (\emph{C}) Attribute hierarchy for $X$ defining a labelled multihypergraph with the added properties of orientation and negative/positive oriented multiplicity, corresponding to the attribute hierarchy in Example~\ref{ex:AttHierarchy}.}
    \label{fig:sys-hyp}
\end{figure}

Figure~\ref{fig:sys-hyp} illustrates the flexibility of the attribute system: augmenting a structure with a richer description requires only a change of distribution hierarchy while leaving the underlying relational structure of the multihypergraph unchanged. The distribution hierarchy in Figure~\ref{fig:sys-hyp}~(\emph{B}) is extended to that of Figure~\ref{fig:sys-hyp}~(\emph{C}), giving two different interpretations of the same underlying structure shown in Figure~\ref{fig:sys-hyp}~(\emph{A}).

\begin{notation}
    In Example~\ref{ex:AssocFun} we employ the standard notion of a \emph{composition} of an integer $n \in \NN_+$, which is a finite sequence in $\NN_+$ such that the elements, called \emph{parts}, sum to $n$. For $m$, $n \in \NN_+$ with $m \le n$ we denote by $C(n,m)$ the set of all compositions of $n$ with $m$ parts.
\end{notation}

\begin{example}[Association functions for AMT-type systems]\label{ex:AssocFun}
Continuing from Example~\ref{ex:AttHierarchy}, let $X \defeq (V,E,\mu_E)$ be a multihypergraph and let $Q \defeq (S,<,\rnk)$ be an AMT-type distribution hierarchy for $X$. The relevant association functions $\chi \colon E \to \bigcup_{e \in E} \{ Q_e \rightarrow \NN \}$ for $X$ with respect to $Q$ describe how the particular instances of $\widehat{\mu}$ and $\widehat{\sigma}$ determine the particular instances of $\widehat{\mu}^{\,\Minus / \Plus}$. Specifically, for $e \in E$ and $(S_e,<_e,\rnk_e) \in \chi_e^{-1}(\NN_+)$, where $S_e \defeq \{\sigma,\mu,\rho,(\mu^{\Minus},\mu^{\Plus})\}$, the vertex multiplicity $\mu$ is partitioned into negative $\mu^{\Minus}$ and positive $\mu^{\Plus}$ parts according to the vertex orientation $\sigma$ as follows:
    \begin{itemize}[topsep=3pt,itemsep=1pt,leftmargin=12pt]
        \item $\Dom(\mu^\Minus) = \{\, v \in e \mid -1 \in \sigma(v) \,\}$ and $\Dom(\mu^\Plus) = \{\, v \in e \mid +1 \in \sigma(v) \,\}$.
        \item $\mu^\Minus(v) = \mu(v)$ when $v \in e$ and $\sigma(v) = \{-1\}$.
        \item $\mu^\Plus(v) = \mu(v)$ when $v \in e$ and $\sigma(v) = \{+1\}$.
        \item $\big(\mu^\Minus(v),\mu^\Plus(v)\big) \in C(\mu(v),2)$ when $v \in e$ and $\sigma(v) = \{-1,+1\}$.
    \end{itemize}
\end{example}

\begin{remark}\label{rem:RelData}
    Example~\ref{ex:AssocFun} illustrates how defining attributes relationally increases the level of resolution on structures. Here vertex multiplicity is first defined globally via $\mu$, before being partitioned into oriented vertex multiplicities defined relationally to $\sigma$ via $\mu^{\Minus}$ and $\mu^{\Plus}$. This second, local level of resolution therefore provides a finer-grained way to associate data to the same structure. This is relevant when the same data can be encoded in multiple ways. The same global multiplicity value can for instance arise from multiple compositions of negative- and positive-oriented multiplicities, and multiplicity data therefore needs to be defined relationally to provide enough context to distinguish between distinct ways of reaching the same value.
\end{remark}

\begin{definition}[Systems hypergraph]
    Let $X \defeq (V,E,\mu_E)$ be a multihypergraph, let $Q \defeq (S,<,\rnk)$ be a distribution hierarchy for $X$, and let $\chi \colon E \to \bigcup_{e \in E} \{ Q_e \rightarrow \NN \}$ be an association function for $X$ with respect to $Q$. Then we call $(V,E,\mu_E,Q,\chi)$ a \emph{systems hypergraph}.
\end{definition}

Systems hypergraphs represent all possible ways to annotate a multihypergraph with a collection of attributes. This lets us define \emph{AMT-type systems hypergraphs}, a particular class of systems hypergraphs designed to represent systems of AMT type. 

\begin{example}[AMT-type systems hypergraph]\label{ex:SysHyp}
    A systems hypergraph $(V,E,\mu_E,Q,\chi)$ is AMT-type when $Q$ and $\chi$ are AMT-type.
\end{example}

Systems hypergraphs are abstract relational systems designed to represent attributed relational structures with full generality. A systems hypergraph integrates a multihypergraph, representing the foundational structure, with a dependency structure of attributes that associates properties to the multihypergraph. This affords maximum independence, generality, efficiency, and precision for attribute definition and association. Further, it facilitates the avoidance of redundancy: an example of redundancy would be defining the same attribute individually on all separate instances of a structure, rather than a single global definition augmented with multiplicity data. Systems hypergraphs also incorporate the flexibility to represent independent attributes in either of two natural and logically equivalent ways: within the attribute structure; or directly associated with the underlying multihypergraph.

\section{Stochastic Petri nets generalise chemical reaction networks}\label{sec:spns-crns}

Before presenting our main results we precisely establish the relationship between CRNs and SPNs. These two formalisms arose in the context of specific applications, and certain aspects of their mathematical notation are arbitrarily defined for representational convenience while others are ambiguous. Formalisms can be difficult to generalise and therefore to integrate when the assumptions behind such choices are left implicit. Here we carefully define and compare the two formalisms, showing in particular that SPNs are in fact strictly more general than CRNs.

Our analysis reveals the similarities and differences between the two formalisms, motivating the abstraction and generalisation of systems hypergraphs. This exemplifies the fact that the particular mathematical form of a representation can constrain its ability to represent certain properties of systems.

To establish the connection between CRNs and SPNs we need to consider isomorphism classes of bipartite graphs, since they allow for multiple second part vertices to be incident on the same set of first part vertices.

\begin{definition}[Type-1 isomorphism and isomorphism classes of directed bipartite graphs]\label{def:BGiso}
    Let $\D$ be the set of all labelled directed bipartite graphs with arc weight functions and second-part vertex weight functions, and let $D \defeq (V,W,F,\omega,\delta)$ and $D^{\prime} \defeq (V^{\prime},W^{\prime},F^{\prime},\omega^{\prime},\delta^{\prime})$ be in $\D$. A type-$1$ bipartite graph isomorphism $\phi$ is a directed graph isomorphism $\phi \colon D \to D^{\prime}$ such that: \begin{enumerate*}[label=(\arabic*)] \item $\phi(V) = V^{\prime}$ and $\phi(W) = W^{\prime}$; \item $\phi$ preserves vertex labels of the first part; \item $\Codom(\omega) = \Codom(\omega^{\prime})$ and $\omega^{\prime}\big(\phi(e)\big) = \omega(e)$ for all $e \in F$; \item $\Codom(\delta) = \Codom(\delta^{\prime})$, $\phi\big(\Dom(\delta)\big) = \Dom(\delta^{\prime})$, and $\delta^{\prime}\big(\phi(v)\big) = \delta(v)$ for all $v \in \Dom(\delta)$. \end{enumerate*} If $D$ and $D^{\prime}$ are type-$1$ isomorphic then we write $D \cong_1 D^{\prime}$.

    Let $R_1$ be the equivalence relation on $\D$ such that $D \mathrel{R_1} D^{\prime}$ if and only if there exists a type-$1$ bipartite graph isomorphism $\pi \colon D \to D^{\prime}$, for all $D$, $D^{\prime} \in \D$. Then the quotient set $\widetilde{\D} \defeq \D/R_1 = \{\, [D]_1 \mid D \in \D \,\}$ consists of the type-$1$ isomorphism classes of $\D$. Each type-1 equivalence class is a directed bipartite graph that is arc weighted, vertex weighted, with labelled vertices in the first part and unlabelled vertices in the second part.
\end{definition}

\begin{definition}[Symmetric vertices, symmetric arc]\label{def:symm}
    Let $G \defeq (V,E)$ be a graph, either undirected or directed, which may also have attributes associated with the vertices and edges. Two vertices $u$, $v \in V$ are \emph{symmetric} if and only if the permutation $(u,v)$ of $V$ is an automorphism of $G$. If $G$ is directed then an arc $uv \in E$ is \emph{symmetric} if and only if $vu \in E$ with the same attributes as $uv$.
\end{definition}

Proposition~\ref{prop:CRN-SPN} provides an informal statement of the relationship between CRNs and SPNs: the formal statement and complete proof are in the Appendix (Proposition~\ref{prop:CRN-SPN-supp}).

\begin{proposition}[(\emph{Informal version}) Correspondence between chemical reaction networks and stochastic Petri nets] \label{prop:CRN-SPN}
     Let $\Z$ be the set of all connected chemical reaction networks with a kinetics $Z \defeq (\S,\C,\R,k)$. Let $\U$ be the set of all stochastic Petri nets $U \defeq (P,T,F,(\omega,\NN_+),(r,\RR_{\ge 0}))$ such that all distinct pairs of vertices in $T$ are asymmetric, and for each vertex $z \in T$ there exists an asymmetric arc incident on $z$. Let $\widetilde{\U}$ be the set of all type-1 isomorphism classes of $\U$. Then there exists a bijection $\kappa \colon \Z \to \widetilde{\U}$ such that:
        \begin{enumerate}[label=(\arabic*),topsep=3pt,itemsep=1pt,leftmargin=16pt]
            \item $\kappa\big((\S,\C,\R,k)\big) = \big[(\S,\R,F,\omega,k)\big]_1$ where $F = F(\S,\R)$ and $\omega = \omega(\S,\R)$.
            \item $\kappa^{-1}\big(\big[(P,T,F,\omega,r)\big]_1\big) = (P,\C,\R,k)$ where $\R = \R(P,T,F,\omega)$, $\C = \C(P,T,F,\omega)$, and $k = k(T,r)$.
        \end{enumerate}
\end{proposition}

Understanding the relationship between CRNs and SPNs invites complexity: it requires us to consider how multiple data interrelate when moving between the two formalisms. This information is most naturally represented by multidimensional relations -- the individual maps that make up the bijection -- on the structure of each formalism, rather than by relating structures pairwise only or CRNs and SPNs themselves directly. This does not mean that the mapping cannot be decomposed into binary incidences -- we show this in Figure~\ref{fig:overview}~(\emph{B}) -- but that the relationship between CRNs and SPNs is only understood through the overall incidence structure of the mappings themselves. In practice this means that conceptually-related structures of the two formalisms may not correspond symmetrically under the map, rather the data a structure represents can be distributed across overlapping but distinct structures through each part of the overall mapping, whose inverse is therefore not straightforwardly the inverse of all its parts.

The asymmetry between CRNs and SPNs is concrete. For instance, constructing a complex set from an SPN requires using data from all parts of the underlying Petri net, whereas in the other direction the data which is represented by the complex set is distributed across the bipartite incidence structure: the supports of positive partial multisets (from the complex sets directly) are represented as transitions; and arc weights (stoichiometric coefficients) represent data from both the species and reactions sets. This asymmetry is also represented in our use of two intermediary notions when establishing the correspondence formally: neighbourhood maps, to associate complexes with sets of arcs; and twin-vertices maps, to recover arc directionality (see Definitions~\ref{def:neighbourhood-map} and~\ref{def:twin-vertices-maps} in Appendix~\ref{app:CRN-SPN}). These simplify the overall mapping by describing how certain local data must be grouped to map across formalisms. In particular, the neighbourhood map illustrates the relevance of context: the species making up a complex must be distributed over arcs that are only defined via local incidences of the relevant vertices.

Crucially, a precise statement of this relationship shows that a bijection between CRNs and SPNs only holds under certain conditions. These conditions, described in Definitions~\ref{def:BGiso} and~\ref{def:symm} and necessary in establishing Proposition~\ref{prop:CRN-SPN}, mean that SPNs make available more structure than CRNs. Type-1 isomorphism classes specifically account for label permutations at the transition level that produce distinct bipartite graphs which share the same relational structure; and symmetry requirements ensure structural uniqueness of transitions, by excluding multiple instances of a transition on the same set of places and by requiring the existence of at least one outgoing and one incoming arc for all transitions.

Having to account for extra structure is evidence of the higher generality of SPNs. We understand generality as precision here, since we are interested in representing systems accurately. Enforcing structural conditions then means that some information represented in SPNs becomes indistinguishable in CRNs. In practical terms, symmetry requirements translate to the inability of CRNs to represent the data corresponding to distinct transitions on the same set of places. Where multiple such transitions exist, they cannot be labelled individually in a CRN and therefore cannot be represented. Type-1 isomorphism classes account for the greater capacity for permutation symmetry of bipartite graphs.

The correspondence also motivates our use of positive partial multisets for representing complexes. When complexes represent vector spaces (Remark~\ref{rem:CRNZeros}), a zero stoichiometry may ambiguously translate to either a zero arc weight or an absent arc in the SPN representation -- a practical instance of Definition~\ref{def:absence}, where data about absence and absence of data coincide. The extra structure of the bipartite representation distinguishes the two meanings by construction, and defining complexes as positive partial multisets achieves the same distinction within CRNs by representing absence of data through function partiality.

\paragraph{CRNs and SPNs: differently expressive frames on the same systems.}

Despite their interchangeability in certain situations, CRNs and SPNs instantiate very different mathematical structures with nonequivalent representational capabilities. These differences become clear when their relationship is made explicit: CRNs are a class of the more general SPN formalism, and their usually claimed equivalence holds only under the specific conditions described in Proposition~\ref{prop:CRN-SPN}.

Their primary demarcation is the possibility to distinguish between distinct relations on the same set of objects. The CRN construction assumes that the differentiation of structurally equivalent relations is not required, whereas the more general SPN allows for distinct but structurally equivalent processes that have different properties, as in stochastic processes. SPNs expose more structure to be labelled, resulting in a more precise representation, while CRNs are more efficient at representing objects.

The broader lesson is of a fundamental trade-off between objects and relations. Objects make certain aspects of systems readily representable, while relations provide greater precision and generality. Moving away from objectification by increasing structural resolution makes representations more complex, as relations that are otherwise equivalent become distinguishable. It is then relations -- not objects -- that are the relevant unit of systems: different relational structures can change distinguishability (therefore, what appears as an object), and it follows that relational data are more appropriately represented by patterns of symmetry on relations~\cite{Crane2018Edge}. Objects are then better understood as equivalence classes of relations, and representation as the process of identifying the appropriate relations that constitute the observed structure.

\section{Systems hypergraphs generalise stochastic Petri nets and chemical reaction networks}

We now present the integration of both CRNs and SPNs within the same framework, by their joint representation as specific systems hypergraphs. The propositions are stated informally, with the formal statements and proofs provided in the Appendix (Propositions~\ref{prop:SH-rep-SPN-supp} and~\ref{prop:SH-rep-CRN-supp}). The relationships between all formalisms are summarised in Figure~\ref{fig:overview}~(\emph{A}).

\begin{definition}[SPN-type systems hypergraph]
    An \emph{SPN-type systems hypergraph} is an AMT-type systems hypergraph where $\AA \in \{\NN,\NN_+\}$ and $\VV \defeq \RR_{\ge 0}$. When considering a particular $\AA$ we refer to an SPN($\AA$)\emph{-type systems hypergraph}. 
\end{definition}

\begin{proposition}[(Informal version) Representation of stochastic Petri nets as SPN-type systems hypergraphs]\label{prop:SH-rep-SPN}
Let $\AA \in \{\NN,\NN_+\}$, let $\Y$ be the set of all stochastic Petri nets $Y \defeq (P,T,F,(\omega,\AA),(r,\RR_{\ge 0}))$ with type-1 isomorphism classes $\widetilde{\Y}$, and let $\Hh_{\text{SPN}}$ be the set of all SPN($\AA$)-type systems hypergraphs $H \defeq (V,E,\mu_E,Q,\chi)$. There exists a bijection $\Gamma \colon \widetilde{\Y} \to \Hh_{\text{SPN}}$ such that:
    \begin{enumerate}[label=(\arabic*),topsep=3pt,itemsep=1pt,leftmargin=16pt]
        \item $\Gamma\big([(P,T,F,\omega,r)]_1\big) = (P,E,\mu_E,Q,\chi)$ where $E = E(P,T,F)$, $\mu_E = \mu_E(P,T,F)$, $\widehat{\sigma} = \widehat{\sigma}(P,T,F)$, $\widehat{\mu} = \widehat{\mu}(P,T,F,\omega)$, $\widehat{\rho} = \widehat{\rho}(P,T,F,r)$, $\widehat{\mu}^{\,\Minus / \Plus} = \widehat{\mu}^{\,\Minus / \Plus}(P,T,F,\omega)$, and $\chi = \chi(P,T,F,\omega,r)$.
        \item $\Gamma^{-1}\big((V,E,\mu_E,Q,\chi)\big) = [(V,T,F,\omega,r)]_1$ where $T = T(E,\mu_E,\widehat{\sigma},\widehat{\mu},\widehat{\rho},\widehat{\mu}^{\,\Minus / \Plus},\chi)$, $F = F(V,E,\mu_E,\widehat{\sigma},\widehat{\mu},\widehat{\rho},\widehat{\mu}^{\,\Minus / \Plus},\chi)$, $\omega = \omega(V,E,\mu_E,\widehat{\sigma},\widehat{\mu},\widehat{\rho},\widehat{\mu}^{\,\Minus / \Plus},\chi)$, and $r = r(E,\mu_E,\widehat{\sigma},\widehat{\mu},\widehat{\rho},\widehat{\mu}^{\,\Minus / \Plus},\chi)$.
    \end{enumerate}
\end{proposition}

Under the forward mapping, some data are sent to the multihypergraph while other data are sent to the attribute structure: places are fixed as vertices, and transitions encode the relational structure of hyperedges through the flow relation, which also contains additional relational data -- hyperedge multiplicity and vertex orientation, both systems hypergraph attributes. Vertex multiplicity, by contrast, explicitly draws on both structural and relational data at once, the latter given by arc weight functions.

The reverse direction draws on data from all parts of the systems hypergraph structure, evidence of the higher efficiency of the framework at partitioning information as shown in Figure~\ref{fig:overview}~(\emph{C}). Unlike Proposition~\ref{prop:CRN-SPN}, this mapping does not require asymmetry of second-part vertices: systems hypergraphs can represent multihyperedges directly, with each instance distinguished by a hyperedge label corresponding to a stochastic rate constant.

The map also exposes a limitation of the SPN framework: a transition on a given set of places must be represented as many times as it is distinctly labelled by arc and vertex labels. This redundancy in representation motivates the use of multiplicity and the attribute structure in systems hypergraphs for maximum efficiency.

\begin{definition}[CRN-type systems hypergraph]
    A \emph{CRN-type systems hypergraph} is an AMT-type systems hypergraph where: $\AA \defeq \NN_+$ and $\VV \defeq \RR_{\ge 0}$; $\chi$ satisfies $\sum_{\{\, (S_e,<_e,\rnk_e) \in Q_e \mid \text{$\sigma$, $\mu$, $(\mu^\Minus,\mu^\Plus) \in S_e$} \,\}} \chi_e((S_e,<_e,\rnk_e)) \in \{0,1\}$ for all $e \in E$, $\sigma \in \{e \rightarrow \I\}$, $\mu \in \{e \rightarrow \NN_+\}$, and $(\mu^\Minus,\mu^\Plus) \in \{ e \rightharpoonup \NN_+ \} \times \{ e \rightharpoonup \NN_+ \}$; and, if $e \in E$, $(S_e,<_e,\rnk_e) \in Q_e$, and $\chi_e((S_e,<_e,\rnk_e)) > 0$ then $(\mu^\Minus,\mu^\Plus) \in S_e$ satisfies $\mu^\Minus \ne \mu^\Plus$.
\end{definition}

\begin{proposition}[(Informal version) Representation of chemical reaction networks as CRN-type systems hypergraphs] \label{prop:SH-rep-CRN}
Let $\Z$ be the set of all connected chemical reaction networks with a kinetics $Z \defeq (\S,\C,\R,k)$, and let $\Hh_{\text{CRN}}$ be the set of all CRN-type systems hypergraphs $H \defeq (V,E,\mu_E,Q,\chi)$. There exists a bijection $\Theta \colon \Z \to \Hh_{\text{CRN}}$ such that:
    \begin{enumerate}[label=(\arabic*),topsep=3pt,itemsep=1pt,leftmargin=16pt]
        \item $\Theta\big((\S,\C,\R,k)\big) = (\S,E,\mu_E,Q,\chi)$ where $E = E(\R)$, $\mu_E = \mu_E(\R)$, $\widehat{\sigma} = \widehat{\sigma}(\R)$, $\widehat{\mu} = \widehat{\mu}(\R)$, $\widehat{\rho} = \widehat{\rho}(\R,k)$, $\widehat{\mu}^{\,\Minus / \Plus} = \widehat{\mu}^{\,\Minus / \Plus}(\R)$, and $\chi = \chi(\R,k)$.
        \item $\Theta^{-1}\big((V,E,\mu_E,Q,\chi)\big) = (V,\C,\R,k)$ where $\R = \R(V,E,\mu_E,\widehat{\sigma},\widehat{\mu},\widehat{\rho},\widehat{\mu}^{\,\Minus / \Plus},\chi)$, $\C = \C(V,E,\mu_E,\widehat{\sigma},\widehat{\mu},\widehat{\rho},\widehat{\mu}^{\,\Minus / \Plus},\chi)$, and $k = k(V,E,\mu_E,\widehat{\sigma},\widehat{\mu},\widehat{\rho},\widehat{\mu}^{\,\Minus / \Plus},\chi)$.
    \end{enumerate}
\end{proposition}

As in the SPN case, species are fixed as vertices, with orientation and multiplicity data recovered from reactions alone. Vertex multiplicities and multihyperedge labels coincide with stoichiometric coefficients and stochastic rate constants, respectively. In the reverse direction, all parts of the systems hypergraph again inform each CRN substructure, illustrating the same efficiency in information partitioning.

This mapping raises an ambiguity analogous to the one arising in the CRN-to-SPN correspondence. When a reaction includes a species with zero stoichiometry, it is unclear whether the corresponding hyperedge should be incident on that species with multiplicity zero, or whether the incidence should be absent altogether. The first interpretation produces a complete hypergraph in which most incidences have multiplicity zero, reducing representational efficiency -- both multiplicity and oriented multiplicity data would be zero, encoding the same information twice. The two interpretations map to distinct systems hypergraphs and it is therefore important to be able to distinguish them.

The hypergraph perspective suggests a natural resolution. Here, relations are characterised primarily by their context of application, and it is therefore more meaningful to define hyperedges on subsets of the ground set than on the total space of objects. This clarifies the role of structure as delineating a local space of objects on which relations are defined, and requires only that local hyperedge supports be aligned to the global ground set to relate the two interpretations.

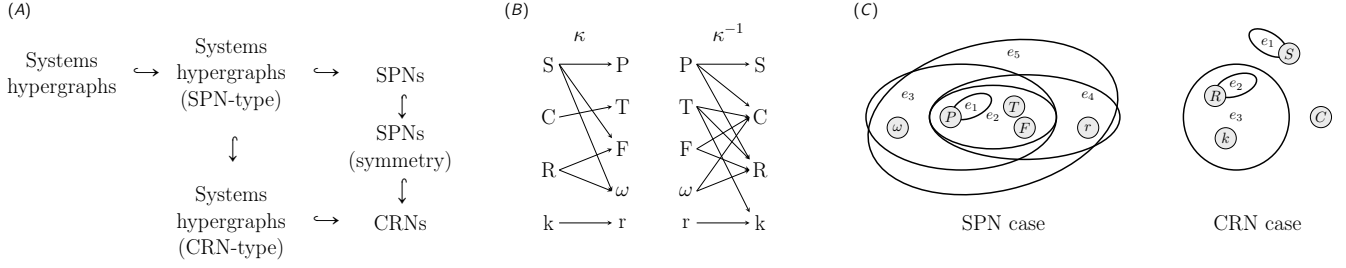
\begin{figure}[t]
    \centering
    \resizebox{\textwidth}{!}{%
\begin{tikzpicture}

\node [font=\LARGE] at (0,0) {\normalfont\sffamily\fontsize{11}{9}\selectfont(\emph{A})};
\node [font=\Large, align=center] at (1,-1.5) {Systems\\hypergraphs};
\node [font=\Large, align=center] at (5,-1.5) {Systems\\hypergraphs\\(SPN-type)};
\node [font=\Large, align=center] at (5,-5) {Systems\\hypergraphs\\(CRN-type)};
\node [font=\Large] at (9,-1.5) {SPNs};
\node [font=\Large, align=center] at (9,-3.25) {SPNs\\(symmetry)};
\node [font=\Large] at (9,-5) {CRNs};
\node [font=\LARGE, rotate around={-90:(0,0)}] at (9,-2.25) {$\hookrightarrow$};
\node [font=\LARGE] at (7.25,-1.5) {$\hookrightarrow$};
\node [font=\LARGE] at (7.25,-5) {$\hookrightarrow$};
\node [font=\LARGE, rotate around={-90:(0,0)}] at (9,-4.25) {$\hookrightarrow$};
\node [font=\LARGE, rotate around={-90:(0,0)}] at (5,-3.25) {$\hookrightarrow$};
\node [font=\LARGE] at (3,-1.5) {$\hookrightarrow$};

\node [font=\LARGE] at (11.75,0) {\normalfont\sffamily\fontsize{11}{9}\selectfont(\emph{B})};
\node [font=\Large] at (12.5,-1.25) {S};
\node [font=\Large] at (12.5,-2.5) {C};
\node [font=\Large] at (12.5,-3.75) {R};
\node [font=\Large] at (12.5,-5) {k};
\node [font=\Large] at (14.25,-1.25) {P};
\node [font=\Large] at (14.25,-2.25) {T};
\node [font=\Large] at (14.25,-3.25) {F};
\node [font=\Large] at (14.25,-4.25) {$\omega$};
\node [font=\Large] at (14.25,-5) {r};
\node [font=\Large] at (13.5,-0.5) {$\kappa\phantom{^-1}$};
\draw [->, >=stealth] (12.75,-1.25) -- (14,-1.25);
\draw [->, >=stealth] (12.75,-5) -- (14,-5);
\draw [->, >=stealth] (12.75,-1.25) -- (14,-3);
\draw [->, >=stealth] (12.75,-1.25) -- (14,-4.25);
\draw [->, >=stealth] (12.75,-3.75) -- (14,-4.25);
\draw [->, >=stealth] (12.75,-3.75) -- (14,-3.25);
\draw [->, >=stealth] (12.75,-2.5) -- (14,-2.25);
\node [font=\Large] at (15.75,-1.25) {P};
\node [font=\Large] at (15.75,-2.25) {T};
\node [font=\Large] at (15.75,-3.25) {F};
\node [font=\Large] at (15.75,-4.25) {$\omega$};
\node [font=\Large] at (15.75,-5) {r};
\node [font=\Large] at (17.5,-1.25) {S};
\node [font=\Large] at (17.5,-2.5) {C};
\node [font=\Large] at (17.5,-3.75) {R};
\node [font=\Large] at (17.5,-5) {k};
\node [font=\Large] at (16.75,-0.5) {$\kappa^{-1}$};
\draw [->, >=stealth] (16,-1.25) -- (17.25,-1.25);
\draw [->, >=stealth] (16,-5) -- (17.25,-5);
\draw [->, >=stealth] (16,-1.25) -- (17.25,-2.25);
\draw [->, >=stealth] (16,-1.25) -- (17.25,-3.5);
\draw [->, >=stealth] (16,-2.25) -- (17.25,-2.5);
\draw [->, >=stealth] (16,-2.25) -- (17.25,-3.5);
\draw [->, >=stealth] (16,-2.25) -- (17.25,-4.75);
\draw [->, >=stealth] (16,-3.25) -- (17.25,-2.5);
\draw [->, >=stealth] (16,-3.25) -- (17.25,-3.75);
\draw [->, >=stealth] (16,-4.25) -- (17.25,-2.5);
\draw [->, >=stealth] (16,-4.25) -- (17.25,-3.75);

\node [font=\LARGE] at (20,0) {\normalfont\sffamily\fontsize{11}{9}\selectfont(\emph{C})};
\draw [line width=1pt, rotate around={-31:(29.5,-0.75)}] (29.5,-0.75) ellipse (0.5cm and 0.25cm) node {\normalsize $e_1$};
\draw [line width=1pt, rotate around={13:(23,-2.5)}] (23,-2.5) ellipse (3cm and 1.75cm);
\draw [line width=1pt] (23.75,-2.5) ellipse (2.25cm and 1cm);
\draw [line width=1pt] (28.75,-2.5) ellipse (1.25cm and 1.25cm) node {\normalsize $e_3$};
\draw [line width=1pt, rotate around={19:(28.75,-1.75)}] (28.75,-1.75) ellipse (0.5cm and 0.25cm) node {\normalsize $e_2$};
\draw [line width=1pt] (22.25,-2.5) ellipse (2.25cm and 1.25cm);
\draw [line width=1pt] (23,-2.5) ellipse (1.5cm and 0.75cm) node {\normalsize $e_2$};
\draw [line width=1pt , rotate around={25:(22.5,-2.25)}] (22.5,-2.25) ellipse (0.5cm and 0.25cm) node {\normalsize $e_1$};
\draw [fill={rgb,255:red,235;green,235;blue,235}] (22,-2.5) circle (0.25cm) node {\normalsize $P$};
\draw [fill={rgb,255:red,235;green,235;blue,235}] (23.5,-2.25) ellipse (0.25cm and 0.25cm) node {\normalsize $T$};
\draw [fill={rgb,255:red,235;green,235;blue,235}] (20.75,-2.75) ellipse (0.25cm and 0.25cm) node {\normalsize $\omega$};
\draw [fill={rgb,255:red,235;green,235;blue,235}] (25.25,-2.75) ellipse (0.25cm and 0.25cm) node {\normalsize $r$};
\draw [fill={rgb,255:red,235;green,235;blue,235}] (23.75,-2.75) ellipse (0.25cm and 0.25cm) node {\normalsize $F$};
\draw [fill={rgb,255:red,235;green,235;blue,235}] (30,-1) circle (0.25cm) node {\normalsize $S$};
\draw [fill={rgb,255:red,235;green,235;blue,235}] (30.75,-2.5) ellipse (0.25cm and 0.25cm) node {\normalsize $C$};
\draw [fill={rgb,255:red,235;green,235;blue,235}] (28.5,-3) ellipse (0.25cm and 0.25cm) node {\normalsize $k$};
\draw [fill={rgb,255:red,235;green,235;blue,235}] (28.25,-2) circle (0.25cm) node {\normalsize $R$};
\node [font=\normalsize] at (21,-2) {$e_3$};
\node [font=\normalsize] at (25.25,-2) {$e_4$};
\node [font=\normalsize] at (23.5,-1) {$e_5$};
\node [font=\Large] at (23.25,-5) {SPN case};
\node [font=\Large] at (29.25,-5) {CRN case};

\end{tikzpicture}}%
    \vspace{-.6cm}
    \caption{Relational structures discussed in this work. (\emph{A}) Relations between the main formalisms discussed in this work given by inclusion. (\emph{B}) Incidence structure of Proposition~\ref{prop:CRN-SPN} mapping in the forward (left) and backwards (right) directions represented as directed bipartite graphs. (\emph{C)} Incidence structures of Propositions~\ref{prop:SH-rep-SPN} (a), left, and~\ref{prop:SH-rep-CRN} (b), right, in the forward direction, represented as unoriented hypergraphs. Hyperedges categorise information more precisely in the SPN case: $e_1$ represents objects; $e_2$, relational structure, hyperedge multiplicity, and orientation data; $e_3$, vertex multiplicity; $e_4$, hyperedge labels; and $e_5$, attribute dependency, the most complex category, incident on all parts.}
    \label{fig:overview}
\end{figure}

\paragraph{Systems hypergraphs: representing the structure of data.}

Systems hypergraphs are a natural abstraction and generalisation of both CRNs and SPNs, constructed iteratively by relating each part of the two formalisms to determine the most natural joint extension. Here hypergraphs serve as an intermediate perspective to ground both formalisms in a shared representational space, a process that clearly shows how CRNs and SPNs are two distinct frames on the same systems. The result is a framework that inherits the representational capabilities of both formalisms at once while also going further, most importantly via the hierarchical attribute system.

Our choice of abstraction level still reflects the specific need to relate CRNs and SPNs, and required striking a balance between generality and specificity. Our definition of systems hypergraphs could be generalised (and abstracted) further -- for instance, by representing all system structure and properties, including hyperedges themselves, as attributes. A fully general definition of attributed relational systems would also generate non-graph structures, beyond the requirements and scope of this work. System hypergraphs seem to represent an appropriate level of generalisation for our systems of interest: hyperedge multiplicity is built directly into the hypergraph structure (Remark~\ref{rem:Multihyperedges}), and the AMT-type specialisation (Example~\ref{ex:SysHyp}) makes explicit the specific representational needs relevant here. 

An important insight is that both CRNs and SPNs represent multidimensional structure, but do so less explicitly than hypergraphs. In CRNs, complexes already encode multidimensional relations on species, but reactions represent further binary relations on complexes and multidimensional structure ends up being distributed across two structural levels. SPNs represent the same data through binary relations between places and transitions that only encode higher-dimensional relations collectively. In both cases, the data represented directly by hypergraphs is distributed across multiple parts of the formalisms. So, graphs can (indirectly) represent multidimensional structure, however additional steps are required to build up structure from low-dimensional parts, also requiring the maintenance of logical consistency. CRNs and SPNs turn out to be two specialised frameworks that can represent multidimensional structure, along different intermediate structures -- complexes and the bipartite constraint, respectively -- but other ways are possible.

The relevance of multiplicity data at the relational level is another key result here. This only becomes obvious in the SPN case, where multihypergraphs are explicitly needed to represent bipartite graphs: a multihyperedge represents multiple transitions on the same places. The result is a more efficient representation of multidimensional data, as hyperedge multiplicity represents structural repetition in a more integrated way, and further identifies the need to distinguish instances within a multiplicity, addressed by the attribute system.

More abstractly, the attribute system clarifies the role of data in disambiguating structural symmetry. In contrast to Section~\ref{sec:spns-crns}, where graph permutation invariance led to label exchangeability only accounted for by isomorphism classes, attributes distinguish otherwise equivalent structures by augmenting them with data at multiple levels of resolution (Remark~\ref{rem:RelData}). Here multiplicity serves as the basic schema allowing distinguishability, representing distinct instances of the same structure explicitly. How precisely instances can be distinguished then depends on the complexity of the attribute dependency structure.

\section{Conclusion}

In this paper we contribute the following results:
\begin{enumerate*}[label={(\arabic*)}]
    \item \emph{Systems hypergraphs}, a general and extendable mathematical formalism for representing systems based on relational structure and a dependency structure of associated properties.
    \item A formal and constructive proof that stochastic Petri nets are strictly more general than chemical reaction networks, providing a precise correspondence between the two formalisms.
    \item Formal and constructive maps between chemical reaction networks (resp. stochastic Petri nets) and their corresponding representations as systems hypergraphs.
    \item A formal and constructive proof of the mathematical and logical equivalence between multihypergraphs and bipartite graphs, a result that is more subtle than usually acknowledged.
\end{enumerate*}

Systems hypergraphs represent the structural organisation of systems in two complementary ways: multidimensional relations on a ground set of objects, corresponding to the fundamental structure of the system; and a dependency structure of attributes associated with the multidimensional relations. An independent attribute is defined without reference to other attributes, whereas a dependent attribute depends on the existence and specification of other attributes for context. Systems hypergraphs integrate the two structures, providing an efficient and precise representation of the structure of data. In this sense, systems hypergraphs can be interpreted as both relational databases and database management systems, where a data structure corresponds to a systems hypergraph, and a map between systems hypergraphs provides a data representation.

Since systems hypergraphs provide a general description of abstract relational structures they have connections with a broad range of existing network configurations including multilayer networks~\cite{Kivela2014Multilayer}. Further, the attribute structure of systems hypergraphs enables a precise representation of more complex relational properties including intrinsic/extrinsic relations and relationships among relations (via hypostatic abstraction), essential for understanding relational systems and developing mechanistic models.

In this work we have implemented abstraction while emphasising the need to maintain logical consistency, particularly with regard to transforming a relation to an object which is formalised through the operation of hypostatic abstraction. While the symbolic transformation of a relation to an object is straightforward mathematically, any application of this process requires the preservation of information to ensure a sound logical connection within the given context. Without such logical consideration, there is no certainty that any analysis of the abstraction can be interpreted for the original system without logical error.

Our realisation of both stochastic Petri nets and chemical reaction networks as systems hypergraphs emphasises a fundamental trade-off between relations and objects when considering representational accuracy. Higher-dimensional relations can provide a more realistic representation of phenomena, however they may present technical and computational challenges. Transforming a higher-dimensional relation to a (hypostatic) object, however, results in an indirect and more abstract representation with a large increase in the number of relations. It is clear that both approaches are required: higher-dimensional relations as the accurate representation, and (hypostatic) objects to provide insight into the higher-dimensional representation.

Accuracy of relational structure, including dimensionality, is nevertheless fundamental for understanding a system. Knowledge resides in identifying patterns: in recognising shared aspects in what is otherwise different, and then forming categories of data and connections between these categories. It is relations that make this possible by describing shared properties of distinct objects. Higher-dimensional relations integrate information which makes patterns more easily identifiable, whereas lower-dimensional relations isolate incidences which can minimise their relevance and result in dissociated representations.

These considerations are relevant for our comparison of stochastic Petri nets and chemical reaction networks. Relating similar information is complicated by being distributed across different substructures in the two formalisms. More importantly, the dissociation of information obfuscated the nonequivalence of the two formalisms, which only became clear under the abstraction of systems hypergraphs. Accessibility of information is therefore key in building understanding, with explanatory power coming from integration.

This has practical consequences for systems biology. Any modelling effort must reconcile the demand for representational accuracy with mathematical or computational tractability, and requires reducing complex systems to more tractable ones. However, there is no guarantee that higher-dimensional properties are preserved by simpler, lower-dimensional representations. For instance, systems of differential equations -- the traditional representation of choice to model dynamical systems, including in biological contexts~\cite{Mesarovic1968Systems} -- often rely on simplifying assumptions of temporal or spatial symmetry rooted in physics that do not necessarily apply to biological systems, as physical complexity differs from evolved biological complexity~\cite{FoxKeller2005Century}.

The challenge of biological systems arises from their intricate and dynamic relational structure, and new mathematical perspectives are needed to better understand these systems~\cite{Vittadello2025Mathematical}. We propose systems hypergraphs as one such perspective. The formalism is sufficiently abstract and general to represent relational systems without enforcing a predefined level of detail, yet is readily applicable in concrete applications: it encompasses both chemical reaction networks and stochastic Petri nets, as well as their downstream representations including differential equations~\cite{Diaz2022HyperGraphsjl}. By integrating distinct modelling perspectives within a shared representational space, systems hypergraphs offer a vantage point from which to survey the landscape of systems biology formalisms, and identify where new abstractions are needed.

\section*{Acknowledgements}

We thank the members of the Theoretical Systems Biology group for stimulating discussions and for their continued support.

\section*{Author contributions}

L.D. contributed the main ideas and results, and drafted and completed the manuscript. S.T.V. developed the main results, established the mathematical framework, and drafted the manuscript. M.P.H.S. edited and reviewed the manuscript.

\newpage
\begin{appendices}

\section{Notation}
Here we discuss our notation additional to that in the main document.

\medskip
\noindent
We denote by $\NNex \defeq \NN \cup \{\infty\}$ the set of extended positive integers including zero. For $n \in \NNex$ we denote by $[n] $ the first $n$ positive integers, noting $[0] = \emptyset$ and $[+\infty] = \NN_+$.

For two sets $A$ and $B \subseteq \RR$ the set of all functions $f$ with domain $\Dom(f) \defeq A$, codomain $\Codom(f) \defeq B$, and finite support $\left| \supp(f) \right| < +\infty$ is denoted $\{A \circlearrow B\}$. We denote by $\indfun_C \colon A \to \{0,1\}$ the indicator function of $C \subseteq A$, and for notational simplicity $\indfun_a \defeq \indfun_{\{a\}}$ for $a \in A$.

Let $A$ and $B$ be two sets. For a pair of partial functions $(f,g) \in \{A \rightharpoonup B\} \times \{A \rightharpoonup B\}$ we define their \emph{sum} as the partial function $f \boxplus g \in \{A \rightharpoonup B\}$ such that $\Dom(f \boxplus g) \defeq \Dom(f) \cup \Dom(g)$, and for $a \in \Dom(f \boxplus g)$:
    \begin{equation*}
        f \boxplus g(a) = 
            \begin{cases}
                f(a) & \text{if $a \in \Dom(f) \setminus \Dom(g)$},\\
                g(a) & \text{if $a \in \Dom(g) \setminus \Dom(f)$},\\
                f(a) + g(a) & \text{if $a \in \Dom(f) \cap \Dom(g)$}.
        \end{cases}
    \end{equation*}

If $f = (f_i)_{i=1}^n$ and $g = (g_i)_{i=1}^n$ are two finite sequences of functions, where $n \in \NN_+$ and the range of $g_i$ is a subset of the domain of $f_i$ for $i \in [n]$, then the composition of $f$ and $g$ is denoted by $f \circledcirc g \defeq (f_i \circ g_i)_{i=1}^n$. If $h$ is a function with the range of $h$ a subset of the domain of $f_i$ for $i \in [n]$ then the composition of $f$ and $h$ is denoted by $f \circledcirc h \defeq (f_i \circ h)_{i=1}^n$, where we may consider $h$ as the constant sequence $h = (h)_{i=1}^n$.

\section{Relationships between the representational forms of complexes in chemical reaction networks}\label{app:complexes}

Chemical complexes are often written in three forms: ordered formal linear combinations, unordered formal linear combinations, or multisets of chemical species. If ordering of the species in the complexes is not required then the complexes can be described with unordered formal linear combinations or multisets, which may simplify the description. When complexes are written as formal linear combinations it is often unclear whether they should be understood as ordered or unordered despite each representation having a different implications for the resulting formalism. In this work we use multisets for complexes, and for completeness we detail here the relationships between the three representational forms of complexes.

The relationship between chemical reaction networks where the complexes are either ordered formal linear combinations or unordered formal linear combinations is described in the following definition and remark.

\begin{definition}[Ordered/unordered formal linear combinations]
    Let $A$ be a countable set. Then
    \begin{equation*}
        \F_A \defeq \big\{\, \textstyle\sum_{i=1}^m \alpha_i a_i \mid \text{$m \in \NNex$, $\alpha \in \{[m] \circlearrow \NN\}$, and $a_i \in A$ for $i \in [m]$} \,\big\}
    \end{equation*}
    is the set of all \emph{ordered formal linear combinations} of elements of $A$ with nonnegative integer coefficients that are finitely nonzero. If $m = 0$ then $\sum_{\emptyset} \defeq \sum_{i=1}^m \alpha_i a_i \in \F_A$ is the \emph{empty formal linear combination}. Further, $\L_A \defeq \{A \circlearrow \NN\}$ is the set of all \emph{unordered formal linear combinations} of elements of $A$.
\end{definition}

\begin{remark}\leavevmode
    \begin{itemize}[topsep=3pt,itemsep=1pt,leftmargin=12pt]
        \item The map $\Theta_A \colon \F_A \to \L_A$ such that $\Theta_A\big(\sum_{i=1}^m \alpha_i a_i\big) \defeq \sum_{i=1}^m \alpha_i \indfun_{a_i}$ is a well-defined surjection, noting that $\Dom(\indfun_{a_i}) \defeq A$ for each $a_i \in A$, and where we identify $\Theta_A\big(\sum_{i=1}^0 \alpha_i a_i\big) = \sum_{i=1}^0 \alpha_i \indfun_{a_i}$ with the zero function in $\L_A$. Let $R$ be the equivalence relation on $\F_A$ given by $f \mathbin{R} g$ if and only if $\Theta_A(f) = \Theta_A(g)$, for $f$, $g \in \F_A$. Then $\Theta_A$ induces a bijection between the quotient set $\Q_A \defeq \F_A/R$ and $\L_A$ given by $[f] \mapsto \Theta_A(f)$.
        
        Each equivalence class in the quotient set $\Q_A$ consists of all possible rearrangements of the terms of a given ordered formal linear combination, both termwise and also splitting/combining terms corresponding to the same elements of $A$. We may therefore consider a formal linear combination in an equivalence class in the quotient set $\Q_A$ without a specified order. In particular when $\big[\sum_{a \in A} \nu_a a\big]_R = \big[\sum_{i=1}^m \alpha_i a_i\big]_R$ we may take $\sum_{a \in A} \nu_a a$ as the representative.
        \item Let $U$ be the equivalence relation on $\F_A \times \F_A$ such that $(f,g) \mathbin{U} (f^{\prime},g^{\prime})$ if and only if $f \mathbin{R} f^{\prime}$ and $g \mathbin{R} g^{\prime}$, for all $(f,g)$, $(f^{\prime},g^{\prime}) \in \F_A \times \F_A$. Then $(\F_A \times \F_A)/U = \Q_A \times \Q_A$, and in particular $[(f,g)]_U = [f]_R \times [g]_R$.
        \item In terms of chemical reaction networks, where $A$ is the species set:
            \begin{itemize}[topsep=3pt,itemsep=1pt,leftmargin=12pt]
                \item When complexes are unordered finite linear combinations then $\L_A$ is the complex set, and $\L_A \times \L_A$ is the set of all reactions.
                \item When complexes are ordered finite linear combinations then $\F_A$ is the complex set, and $\F_A \times \F_A$ is the set of all reactions. Taking the quotient $\Q_A \defeq \F_A/R$ gives $\L_A$, up to bijection, and taking the associated quotient $(\F_A \times \F_A)/U = \Q_A \times \Q_A$ gives $\L_A \times \L_A$, up to bijection.
            \end{itemize}
    \end{itemize}
\end{remark}

The relationship between chemical reaction networks where the complexes are positive partial multisets and either ordered or unordered formal linear combination is given by the following proposition. The result stated in the proposition is standard, and the existence of the correspondence follows from the bijections between $\Q_A$ and $\L_A$ and between $\M_A$ and $\L_A$, however we specify the correspondence in detail and provide a direct proof for completeness. The proposition shows that an unordered formal linear combination, which is a function so already has a reduced form, may be identified with the corresponding set of ordered formal linear combinations that are rearrangements of each other and may not be in reduced form.

\begin{proposition}
    Suppose $A$ is a countable set, and $\M_A \defeq \{\,(A,\mu) \mid \text{$\mu \in \{A \rightharpoonup \NN_+\}$ and $\left| \Dom(\mu) \right| < +\infty$} \,\}$ is the collection of all positive partial multisets with underlying set $A$ and multiplicity functions with finite domains. Then $\psi \colon \M_A \to \Q_A$ where $(A,\mu) \mapsto \big[\sum_{a \in \Dom(\mu)} \mu(a)a\big]$ is a bijection with inverse given by $[f] \mapsto \big(A,\Theta_A(f)|_{\supp(\Theta_A(f))}\big)$.
\end{proposition}

\begin{proof}
    Denote by $\Lambda_A$ the bijection from $\Q_A$ onto $\L_A$. The map $\eta \colon \M_A \to \L_A$ such that $(A,\mu) \mapsto \sum_{a \in \Dom(\mu)} \mu(a) \indfun_a$ is a bijection. Then $\psi \defeq \Lambda_A^{-1} \circ \eta$ is a bijection. In particular, if $A = \emptyset$ then $\M_A$ contains only $(A,\mu)$ where $\mu$ is the empty function, and $\Q_A$ contains only $\big[\sum_{\emptyset}\big] = \big\{\sum_{\emptyset}\big\}$.
\end{proof}

\section{Correspondence between chemical reaction networks and stochastic Petri nets}\label{app:CRN-SPN}

We assume that CRNs are connected, as isolated complexes cannot be represented in SPNs: an isolated complex in a CRN provides the underlying subset of the species and the multiplicity of these species, however the orientation given by the position in a reaction (as reactant or product) is absent and we cannot assign corresponding arcs in SPNs.

\begin{definition}[Neighbourhood map, twins]\label{def:neighbourhood-map}
    Suppose $G = (V,E)$ is an undirected or directed graph. The \emph{open neighbourhood map} $\N \colon V \to \powerset(V)$ is such that $\N(v)$ is the subset of vertices adjacent to, and excluding, $v \in V$. The \emph{closed neighbourhood map} $\bar{\N} \colon V \to \powerset(V)$ is such that $\bar{\N}(v) = \N(v) \cup \{v\}$ for $v \in V$. Two vertices $v$, $w \in V$ are \emph{twins} if $\N(v) = \N(w)$.

    Moreover, for a directed graph, the \emph{open incoming neighbourhood map} $\N_{\text{in}} \colon V \to \powerset(V)$ is such that, for $v \in V$, $\N_{\text{in}}(v)$ is the subset of vertices, excluding $v$, that are incoming adjacent to $v$; and, the \emph{open outgoing neighbourhood map} $\N_{\text{out}} \colon V \to \powerset(V)$ is such that, for $v \in V$, $\N_{\text{out}}(v)$ is the subset of vertices, excluding $v$, that are outgoing adjacent to $v$. The \emph{closed incoming neighbourhood map} $\bar{\N}_{\text{in}} \colon V \to \powerset(V)$ is such that $\bar{\N}_{\text{in}}(v) = \N_{\text{in}}(v) \cup \{v\}$ for $v \in V$; and, the \emph{closed outgoing neighbourhood map} $\bar{\N}_{\text{out}} \colon V \to \powerset(V)$ is such that $\bar{\N}_{\text{out}}(v) = \N_{\text{out}}(v) \cup \{v\}$ for $v \in V$.
\end{definition}

\begin{remark}
    If $v$ and $w$ are twin vertices and $v \ne w$ then $w \notin \N(v)$, equivalently $v \notin \N(w)$. Our definition of twins is sometimes called \emph{false twins} in the literature: in this case, two vertices $v$, $w \in V$ are \emph{true twins} if $\bar{\N}(v) = \bar{\N}(w)$, and if $v \ne w$ then $w \in \N(v)$, equivalently $v \in \N(w)$. We have no need for this delineation here.
\end{remark}

\begin{notation}[Set of all incident edges]
    If $G = (V,E)$ is an undirected or a directed graph then for $v \in V$ the set of all edges or arcs in $E$ incident on $v$ is denoted $E(v)$.
\end{notation}

The following definition extends the notion of a type-1 isomorphism for directed bipartite graphs, as employed in the main document, to include the case of a type-0 isomorphism.

\begin{definition}[Type-t isomorphism and isomorphism classes of directed bipartite graphs]
    Let $\D$ be the set of all labelled directed bipartite graphs with arc weight functions and second-part vertex weight functions, and let $D \defeq (V,W,F,\omega,\delta)$ and $D^{\prime} \defeq (V^{\prime},W^{\prime},F^{\prime},\omega^{\prime},\delta^{\prime})$ be in $\D$. Let $t \in \{0,1\}$. A type-$t$ directed bipartite graph isomorphism $\phi \colon D \to D^{\prime}$ is a directed graph isomorphism such that: \begin{enumerate*}[label=(\arabic*)] \item $\phi(V) = V^{\prime}$ and $\phi(W) = W^{\prime}$; \item If $t = 0$ then $\phi$ preserves no vertex labels, and if $t = 1$ then $\phi$ preserves vertex labels of only the first part; \item $\Codom(\omega) = \Codom(\omega^{\prime})$ and $\omega^{\prime}\big(\phi(e)\big) = \omega(e)$ for all $e \in F$; \item $\Codom(\delta) = \Codom(\delta^{\prime})$, $\phi\big(\Dom(\delta)\big) = \Dom(\delta^{\prime})$, and $\delta^{\prime}\big(\phi(v)\big) = \delta(v)$ for all $v \in \Dom(\delta)$. \end{enumerate*} If $D$ and $D^{\prime}$ are type-$t$ isomorphic then we write $D \cong_t D^{\prime}$. Note that type-$1$ isomorphism implies type-$0$ isomorphism.

    Let $R_t$ be the equivalence relation on $\D$ such that $D \mathrel{R_t} D^{\prime}$ if and only if there exists a type-$t$ directed bipartite graph isomorphism $\phi \colon D \to D^{\prime}$, for all $D$, $D^{\prime} \in \D$. Then the quotient set $\widetilde{\D} \defeq \D/R_t = \{\, [D]_t \mid D \in \D \,\}$ consists of the type-$t$ isomorphism classes of $\D$. Each type-0 equivalence class is an \emph{unlabelled} directed bipartite graph that is arc weighted and second-part vertex weighted, and each type-1 equivalence class is a directed bipartite graph that is arc weighted, second-part vertex weighted, with labelled vertices in the first part and unlabelled vertices in the second part.
\end{definition}

\begin{definition}[Sub-bipartite graph]
    Let $G \defeq (V,W,E,\omega,\delta)$ be a vertex-labelled undirected (resp. directed) bipartite graph with edge (resp. arc) weight function $\omega$ and vertex weight function $\delta$. If $C \subseteq V \cup W$ then the \emph{induced sub-bipartite graph} $G[C]$ of $G$ with respect to $C$ is $G[C] \defeq (V \cap C,W \cap E(C),\omega|_{E(C)},\delta|_{(V \cap C) \cup (W \cap C)})$ where $E(C)$ is the set of all edges (resp. arcs) in $E$ with vertices in $(V \cap C) \cup (W \cap C)$.

    If $G$ is directed and $w \in W$ then the \emph{$w$-induced incoming sub-bipartite graph} is $G_{\text{in}}[w] \defeq (\N_{\text{in}}(w),\{w\},E_\text{in}(w),\omega|_{E_\text{in}(w)},\allowbreak\delta|_{\bar{\N}_{\text{in}}(w)})$ where $E_\text{in}(w) \defeq \{\, (v,w) \in E \mid v \in \N_{\text{in}}(w)\,\}$; and the \emph{$w$-induced outgoing sub-bipartite graph} is $G_{\text{out}}[w] \defeq (\N_{\text{out}}(w),\{w\},E_\text{out}(w),\omega|_{E_\text{out}(w)},\delta|_{\bar{\N}_{\text{out}}(w)})$ where $E_\text{out}(w) \defeq \{\, (w,v) \in E \mid v \in \N_{\text{out}}(w)\,\}$.

    We employ the same notation for sub-bipartite graphs without vertex weight functions.
\end{definition}

\begin{definition}[Twin-vertices maps]\label{def:twin-vertices-maps}
    Suppose $G = (V,E)$ is an undirected or directed graph. The \emph{twin-vertices map} $\T \colon V \to \powerset(V)$ is given by $\T(v) \defeq \{\, w \in V \mid \N(w) = \N(v) \,\}$ for $v \in V$. Note that if $G$ is undirected then $w \in \T(v)$ if and only if $v$ and $w$ are symmetric vertices. If $G$ is directed then the \emph{directed twin-vertices map} $\T_{\text{dir}} \colon V \to \powerset(V)$ is such that $\T_{\text{dir}}(v) \defeq \{\, w \in V \mid \text{$\N_{\text{in}}(w) = \N_{\text{in}}(v)$ and $\N_{\text{out}}(w) = \N_{\text{out}}(v)$} \,\}$ for $v \in V$. Note that: $w \in \T_{\text{dir}}(v)$ if and only if $v$ and $w$ are symmetric vertices; and $\T_{\text{dir}}(v) \subseteq \T(v)$.

    Suppose $D \defeq (V,W,F,\omega)$ is a directed bipartite graph with arc weight function $\omega$. The \emph{arc weight twin-vertices map} $\T_{{\text{awgt}}} \colon W \to \powerset(W)$ is given by $\T_{{\text{awgt}}}(v) \defeq \{\, w \in \T(v) \mid \text{$\sum_{e \in F(w) \cap F(x)} \omega(e) = \sum_{e \in F(v) \cap F(x)} \omega(e)$ for all $x \in \N(v)$} \,\}$ for $v \in W$. The \emph{incoming arc weight twin-vertices map} $\T_{{\text{in}}} \colon W \to \powerset(W)$ is given by $\T_{{\text{in}}}(v) \defeq \{\, w \in \T(v) \mid D_{\text{in}}[w] \cong_1 D_{\text{in}}[v] \,\}$ for $v \in W$, and the \emph{outgoing arc weight twin-vertices map} $\T_{{\text{out}}} \colon W \to \powerset(W)$ is given by $\T_{{\text{out}}}(v) \defeq \{\, w \in \T(v) \mid D_{\text{out}}[w] \cong_1 D_{\text{out}}[v] \,\}$ for $v \in W$. Note that the type-1 isomorphisms in the definitions of $\T_{{\text{in}}}$ and $\T_{{\text{out}}}$ are with respect to the graph $D$ and exclude any other associated information such as vertex weight functions.

    Suppose $D \defeq (V,W,F,\delta)$ is a directed bipartite graph with vertex weight function $\delta$. The \emph{vertex weight twin-vertices map} $\T_{{\text{vwgt}}} \colon \Dom(\delta) \to \powerset\big(\Dom(\delta)\big)$ is given by $\T_{{\text{vwgt}}}(v) \defeq \{\, w \in \T(v) \cap \Dom(\delta) \mid \delta(v) = \delta(w) \,\}$ for $v \in \Dom(\delta)$.
\end{definition}

\begin{remark}
    The sets $\{\, \T(v) \mid v \in V \,\}$ and $\{\, \T_{\text{dir}}(v) \mid v \in V \,\}$ each partition $V$. Further, $\{\, \T_{\text{dir}}(w) \mid w \in \T(v) \,\}$ partitions $\T(v)$ for $v \in V$, the sets $\{\, \T_{{\text{awgt}}}(v) \mid v \in W \,\}$, $\{\, \T_{{\text{vwgt}}}(v) \mid v \in W \,\}$, $\{\, \T_{{\text{in}}}(v) \mid v \in W \,\}$, and $\{\, \T_{{\text{out}}}(v) \mid v \in W \,\}$ each partition $W$.
\end{remark}

\begin{definition}[Converse of a directed graph]
If $D = (V,F)$ is a directed graph then the \emph{converse} of $D$ is the directed graph $D^{\ast} \defeq (V,F^{\ast})$ with the arcs from $D$ reversed, so $(u,v) \in F^{\ast}$ if and only if $(v,u) \in F$. We may regard the converse graph as the result of applying a bijective transformation from $F$ to $F^{\ast}$, so that any functions on $F$ are preserved: for example, if $\omega \colon F \to \AA$ is an arc weight function for $D$ then the corresponding arc weight function $\omega^{\ast} \colon F^{\ast} \to \AA$ for $D^{\ast}$ satisfies $\omega^{\ast}(u,v) \defeq \omega(v,u)$.
\end{definition}

\begin{proposition}[Correspondence between chemical reaction networks and stochastic Petri nets] \label{prop:CRN-SPN-supp}
    Let $\Z$ be the set of all connected chemical reaction networks with a kinetics $Z \defeq (\S,\C,\R,k)$. Let $\U$ be the set of all stochastic Petri nets $U \defeq (P,T,F,(\omega,\NN_+),(r,\RR_{\ge 0}))$ such that $\T_{\text{in}}(z) \cap \T_{\text{out}}(z) = \{z\}$ and $\dot{U}_{\text{in}}[z] \ne \dot{U}_{\text{out}}^{\ast}[z]$ for all $z \in T$ with respect to $\dot{U} \defeq (P,T,F,(\omega,\NN_+))$, and let $\widetilde{\U}$ be the set of all type-1 isomorphism classes of $\U$. Then there exists a bijection $\kappa \colon \Z \to \widetilde{\U}$ such that:
        \begin{enumerate}[label=(\arabic*)]
            \item $\kappa\big((\S,\C,\R,k)\big) = \big[(\S,\R,F,\omega,k)\big]_1$ where:
                \begin{enumerate}[label=(\arabic{enumi}.\arabic*),topsep=0pt]
                    \item $F \defeq \big\{\, \big(s,(f,g)\big) \in \S \times \R \mid s \in \Dom(f) \,\big\} \cup \big\{\, \big((f,g),s\big) \in \R \times \S \mid s \in \Dom(g) \,\big\}$.
                    \item $\omega \colon F \to \NN_+$ is given by $\omega\big(s,(f,g)\big) \defeq f(s)$ when $\big(s,(f,g)\big) \in F$ and $\omega\big((f,g),s\big) \defeq g(s)$ when $\big((f,g),s\big) \in F$.
                \end{enumerate}
            \item $\kappa^{-1}\big(\big[(P,T,F,\omega,r)\big]_1\big) = (P,\C,\R,k)$ where: \label{prop:kappa-inv}
                \begin{enumerate}[label=(\arabic{enumi}.\arabic*),topsep=0pt]
                    \item $\R \defeq \{\, (f_x,g_x) \mid x \in T \,\}$ such that, for each $x \in T$, we have:
                    \begin{enumerate}[label=(\arabic{enumi}.\arabic{enumii}.\arabic*)]
                        \item $f_x$, $g_x \in \{P \rightharpoonup \NN_+\}$.
                        \item $f_x(y) \defeq \omega(y,x)$ for $y \in \Dom(f_x) \defeq \N_{\text{in}}(x)$.
                        \item $g_x(y) \defeq \omega(x,y)$ for $y \in \Dom(g_x) \defeq \N_{\text{out}}(x)$.
                    \end{enumerate}
                    \item $\C \defeq \bigcup_{(f,g) \in \R} \{f,g\}$.
                    \item $k \colon \R \to \RR_{\ge 0}$ is given by $k(f_x,g_x) \defeq r(x)$ for $x \in T$.
                \end{enumerate}
        \end{enumerate}
\end{proposition}

\begin{proof}
    The proof consists of four steps, where we show:
        \begin{enumerate*}[label={(\arabic*)}]
            \item The map $\kappa$ is well defined. \label{kappa}
            \item The map $\eta \colon \widetilde{\U} \to \Z$ in Part~\ref{prop:kappa-inv} of this proposition is well defined. \label{eta}
            \item $\eta \circ \kappa$ is the identity map on $\Z$. \label{eta-kappa}
            \item $\kappa \circ \eta$ is the identity map on $\widetilde{\U}$. \label{kappa-eta}
        \end{enumerate*}
    It then follows, in particular, that $\eta = \kappa^{-1}$.

    Step~\ref{kappa}: To show that $\kappa$ is well defined, let $Z \defeq (\S,\C,\R,k) \in \Z$ and we show that $U \defeq (\S,\R,F,\omega,k) \in \U$. Note that $U$ is a stochastic Petri net with $\omega$ positive.

    For $z \defeq (f,g) \in \R$ we show that $\T_{\text{in}}(z) \cap \T_{\text{out}}(z) = \{z\}$. If $x \defeq (f^{\prime},g^{\prime}) \in \R$ with $x \in \T_{\text{in}}(z) \cap \T_{\text{out}}(z)$ then $\dot{U}_{\text{in}}[x] \cong_1 \dot{U}_{\text{in}}[z]$ and $\dot{U}_{\text{out}}[x] \cong_1 \dot{U}_{\text{out}}[z]$, so: $\N_{\text{in}}(x) = \N_{\text{in}}(z)$ and $\N_{\text{out}}(x) = \N_{\text{out}}(z)$, hence $\Dom(f^{\prime}) = \N_{\text{in}}(x) = \N_{\text{in}}(z) = \Dom(f)$ and $\Dom(g^{\prime}) = \N_{\text{out}}(x) = \N_{\text{out}}(z) = \Dom(g)$; $f^{\prime}(s) = \omega\big(s,(f^{\prime},g^{\prime})\big) = \omega\big(s,(f,g)\big) = f(s)$ for $s \in \Dom(f^{\prime}) = \Dom(f)$; and $g^{\prime}(s) = \omega\big((f^{\prime},g^{\prime}),s\big) = \omega\big((f,g),s\big) = g(s)$ for $s \in \Dom(g^{\prime}) = \Dom(g)$. Therefore $x = (f^{\prime},g^{\prime}) = (f,g) = z$.

    For $z \defeq (f,g) \in \R$ we show $\dot{U}_{\text{in}}[z] \ne \dot{U}_{\text{out}}^{\ast}[z]$. Now, $\dot{U}_{\text{in}}[z] = \dot{U}_{\text{out}}^{\ast}[z]$ if and only if $\N_{\text{in}}(z) = \N_{\text{out}}(z)$ and $\omega\big((s,z)\big) = \omega\big((z,s)\big)$ for all $s \in \N_{\text{in}}(z) = \N_{\text{out}}(z)$, if and only if $\Dom(f) = \Dom(g)$ and $f(s) = g(s)$ for all $s \in \Dom(f) = \Dom(g)$, if and only if $f = g$. Since the binary relation $\R$ is irreflexive the result follows.

    Step~\ref{eta}: To show that $\eta$ is well defined, let $U \defeq (P,T,F,\omega,r)$ and $U^{\prime} \defeq (P,T^{\prime},F^{\prime},\omega^{\prime},r^{\prime})$ be in $\U$ with $U^{\prime} \in [U]_1$, and we show that $Z \defeq \eta\big([U]_1\big) = (P,\C,\R,k)$ is equal to $Z^{\prime} \defeq \eta\big([U^{\prime}]_1\big) = (P,\C^{\prime},\R^{\prime},k^{\prime})$, and $Z \in \Z$. Denote by $\psi \colon U \to U^{\prime}$ a type-1 isomorphism.

    We first show that $\R = \R^{\prime}$. If $(f_x,g_x) \in \R$ for $x \in T$ then $(f_{\psi(x)},g_{\psi(x)}) \in \R^{\prime}$ where $\psi(x) \in T^{\prime}$. So $\Dom(f_{\psi(x)}) = \N_{\text{in}}\big(\psi(x)\big) = \N_{\text{in}}(x) = \Dom(f_x)$, and $f_{\psi(x)}(y) = \omega^{\prime}\big(y,\psi(x)\big) = \omega(y,x) = f_x(y)$ for $y \in \Dom(f_x) = \Dom(f_{\psi(x)})$, imply $f_{\psi(x)} = f_x$. Similarly $g_{\psi(x)} = g_x$, so $(f_x,g_x) = (f_{\psi(x)},g_{\psi(x)}) \in \R^{\prime}$. Therefore $\R \subseteq \R^{\prime}$, and a similar argument gives $\R^{\prime} \subseteq \R$. Since $\R = \R^{\prime}$ it follows that $\C = \C^{\prime}$.

    We show $k = k^{\prime}$. Note, in particular, that $k$ is well defined since it follows from the condition $\T_{\text{in}}(z) \cap \T_{\text{out}}(z) = \{z\}$ for $z \in T$ with respect to $\dot{U}$ that $(f_x,g_x) = (f_y,g_y)$ implies $x = y$. Now, if $(f_x,g_x) \in \R$ for some $x \in T$ then, since $(f_x,g_x) = (f_{\psi(x)},g_{\psi(x)})$, we have $k^{\prime}(f_x,g_x) = k^{\prime}(f_{\psi(x)},g_{\psi(x)}) = r^{\prime}(\psi(x)) = r(x) = k(f_x,g_x)$. We conclude that $k = k^{\prime}$. Therefore $Z = Z^{\prime}$.

    Finally we show $Z \in \Z$. By the definition of $\R$, and since $\omega$ is positive, each $f \in \C$ is a positive partial multiset with underlying set $P$. Moreover, $Z$ is connected by the definition of $\C$. The binary relation $\R$ is irreflexive, since $(f_x,g_x) \in \R$ for some $x \in T$ with $f_x = g_x$ implies $\N_{\text{in}}(x) = \N_{\text{out}}(x)$ and $\omega(y,x) = f_x(y) = g_x(y) = \omega(x,y)$ for $y \in \N_{\text{in}}(x) = \N_{\text{out}}(x)$, hence $\dot{U}_{\text{in}}[x] = \dot{U}_{\text{out}}^{\ast}[x]$.

    Step~\ref{eta-kappa}: For $(\S,\C,\R,k) \in \Z$ we have $\eta \circ \kappa\big((\S,\C,\R,k)\big) = \eta\big(\big[(\S,\R,F,\omega,k)\big]_1\big) = (\S,\C^{\prime},\R^{\prime},k^{\prime})$. We show that $\R = \R^{\prime}$. Let $x \defeq (f,g) \in \R$. Then $f_x(y) = \omega(y,x) = \omega\big(y,(f,g)\big) = f(y)$ for $y \in \Dom(f_x) = \N_{\text{in}}(x) = \Dom(f)$, hence $f_x = f$. Similarly, $g_x = g$. So $x = (f,g) = (f_x,g_x) \in \R^{\prime}$. Therefore $\R \subseteq \R^{\prime}$. For the reverse inclusion, let $(f_x,g_x) \in \R^{\prime}$ for some $x \defeq (f,g) \in \R$. Then $f_x(y) = \omega(y,x) = \omega\big(y,(f,g)\big) = f(y)$ for $y \in \Dom(f_x) = \N_{\text{in}}(x) = \Dom(f)$, hence $f_x = f$. Similarly, $g_x = g$, so $(f_x,g_x) = (f,g) \in \R$. Therefore $\R^{\prime} \subseteq \R$. We conclude that $\R = \R^{\prime}$, and it follows that $\C^{\prime} = \bigcup_{(f,g) \in \R^{\prime}} \{f,g\} = \bigcup_{(f,g) \in \R} \{f,g\} = \C$. Further, for $x \defeq (f,g) \in \R = \R^{\prime}$ we have $(f_x,g_x) = (f,g)$ so $k^{\prime}(f,g) = k^{\prime}(f_x,g_x) = k(x) = k(f,g)$, hence $k = k^{\prime}$. We conclude that $\eta \circ \kappa$ is the identity map on $\Z$.

    Step~\ref{kappa-eta}: For $\big[(P,T,F,\omega,r)\big]_1 \in \widetilde{\U}$ we have $\kappa \circ \eta\big(\big[(P,T,F,\omega,r)\big]_1\big) = \kappa\big((P,\C,\R,k)\big) = \big[(P,\R,F^{\prime},\omega^{\prime},k)\big]_1$. It suffices to show that the vertex map $\psi \colon P \cup T \to P \cup \R$ such that $\psi|_P$ is the identity map and $\psi|_T (x) = (f_x,g_x)$ for $x \in T$ is a type-1 isomorphism from $B \defeq (P,T,F,\omega,r)$ to $B^{\prime} \defeq (P,\R,F^{\prime},\omega^{\prime},k)$.

    $\psi|_P$ is a bijection onto $P$, so we show that $\psi|_T$ is a bijection onto $\R$. Surjectivity of $\psi|_T$ follows from the definition of $\R$. For injectivity, let $x$, $z \in T$ with $\psi|_T(x) = \psi|_T(z)$, hence $(f_x,g_x) = (f_z,g_z)$. So $\N_{\text{in}}(x) = \Dom(f_x) = \Dom(f_z) = \N_{\text{in}}(z)$ and $\N_{\text{out}}(x) = \Dom(g_x) = \Dom(g_z) = \N_{\text{out}}(z)$, and then $\omega(y,x) = f_x(y) = f_z(y) = \omega(y,z)$ for all $y \in \N_{\text{in}}(x) = \N_{\text{in}}(z)$, and $\omega(x,y) = g_x(y) = g_z(y) = \omega(z,y)$ for all $y \in \N_{\text{out}}(x) = \N_{\text{out}}(z)$. So $x \in \T(z)$, $\dot{B}_{\text{in}}[x] \cong_1 \dot{B}_{\text{in}}[z]$, and $\dot{B}_{\text{out}}[x] \cong_1\dot{B}_{\text{out}}[z]$, hence $x \in \T_{\text{in}}(z) \cap \T_{\text{out}}(z) = \{z\}$, thus $x = z$. Therefore $\psi|_T$ is an injection.

    We show that $\psi$ preserves arcs. If $(p,t) \in F \cap (P \times T)$ then $\psi(p,t) = \big(\psi(p),\psi(t)\big) = \big(p,(f_t,g_t)\big) \in F^{\prime} \cap (P \times \R)$, since $p \in \N_{\text{in}}(t) = \Dom(f_t)$. If $(t,p) \in F \cap (T \times P)$ then a similar argument gives $\psi(t,p) \in F^{\prime} \cap (\R \times P)$. Conversely, let $\big(p,(f_t,g_t)\big) \in F^{\prime} \cap (P \times \R)$ where $p \in \Dom(f_t)$ and $t \in T$. Then $\psi^{-1}\big(\big(p,(f_t,g_t)\big)\big) = \big(\big(\psi^{-1}(p),\psi^{-1}(f_t,g_t)\big)\big) = (p,t) \in F \cap (P \times T)$, since $p \in \Dom(f_t) = \N_{\text{in}}(t)$. A similar argument gives $(t,p) \in F \cap (T \times P)$ for $\big((f_t,g_t),p\big) \in F^{\prime} \cap (\R \times P)$.

    For preservation of arc weights, if $(p,t) \in F \cap (P \times T)$ then $\omega^{\prime}\big(\psi(p),\psi(t)\big) = \omega^{\prime}\big(p,(f_t,g_t)\big) = f_t(p) = \omega(p,t)$, and similarly for $(t,p) \in F \cap (T \times P)$.

    For preservation of stochastic rate constants, if $x \in T$ then $k\big(\psi(x)\big) = k\big((f_x,g_x)\big) = r(x)$.

    We show that $\psi$ preserves arcs. If $(p,t) \in F \cap (P \times T)$ then $\psi(p,t) = \big(\psi(p),\psi(t)\big) = \big(p,(f_t,g_t)\big) \in F^{\prime} \cap (P \times \R)$, since $p \in \N_{\text{in}}(t) = \Dom(f_t)$. If $(t,p) \in F \cap (T \times P)$ then a similar argument gives $\psi(t,p) \in F^{\prime} \cap (\R \times P)$. Conversely, let $\big(p,(f_t,g_t)\big) \in F^{\prime} \cap (P \times \R)$ where $p \in \Dom(f_t)$ and $t \in T$. Then $\psi^{-1}\big(\big(p,(f_t,g_t)\big)\big) = \big(\big(\psi^{-1}(p),\psi^{-1}(f_t,g_t)\big)\big) = (p,t) \in F \cap (P \times T)$, since $p \in \Dom(f_t) = \N_{\text{in}}(t)$. A similar argument gives $(t,p) \in F \cap (T \times P)$ for $\big((f_t,g_t),p\big) \in F^{\prime} \cap (\R \times P)$.

    For preservation of arc weights, if $(p,t) \in F \cap (P \times T)$ then $\omega^{\prime}\big(\psi(p),\psi(t)\big) = \omega^{\prime}\big(p,(f_t,g_t)\big) = f_t(p) = \omega(p,t)$, and similarly for $(t,p) \in F \cap (T \times P)$.

    For preservation of stochastic rate constants, if $x \in T$ then $k\big(\psi(x)\big) = k\big((f_x,g_x)\big) = r(x)$.
\end{proof}

\section{Existence of a well-ordered distribution hierarchy for systems hypergraphs}\label{app:dist-hier}

\begin{definition}[Well-ordered distribution hierarchy]
Let $X \defeq (V,E,\mu_E)$ be a multihypergraph, and let $Q \defeq (S,<,\rnk)$ be a distribution hierarchy for $X$. Suppose $\prec$ is a well ordering on $S$ such that: $j$, $k \in \NN$ with $j < k$ implies $\widehat{\pi} \prec \widehat{\pi}^{\prime}$ for all $\widehat{\pi} \in \rnk^{-1}(j)$ and $\widehat{\pi}^{\prime} \in \rnk^{-1}(k)$, and $\rnk^{-1}(j)$ is an interval in $(S,\prec)$ for $j \in \NN$. Then $(S,\prec)$ is a \emph{well-ordered distribution hierarchy} for $X$.
\end{definition}

The existence of a well ordering on $S$ follows from the existence of a linear order on $S$, which is shown in the following Proposition 2, and from $S$ being finite. Note that the well ordering of the distribution hierarchy is not necessarily unique, since the distribution functions within an interval are mutually independent so can be permuted within the interval.

\begin{proposition}\label{prop:dist-hier}
Let $X \defeq (V,E,\mu_E)$ be a multihypergraph, and let $Q \defeq (S,<,\rnk)$ be a distribution hierarchy for $X$. Then there exists a well-ordered strict linear extension $\prec$ of $<$ on $S$ such that $\rnk^{-1}(i)$ is an interval in $(S,\prec)$ for $i \in \NN$.
\end{proposition}

\begin{proof}
    Define the binary relation $\lhd$ on $S$ such that $\widehat{\pi} \lhd \widehat{\pi}^{\prime}$ if and only if $\rnk(\widehat{\pi}) < \rnk(\widehat{\pi}^{\prime})$. Then $(S,\lhd)$ is a strict poset: asymmetry and transitivity follow from the respective asymmetry and transitivity of $(<,\NN)$. By the order-extension principle the strict partial order $\lhd$ on $S$ can be extended to a strict linear order $\lhd^{\ast}$ of $S$. Since $(S,<)$ is a subposet of $(S,\lhd)$ the strict linear order $\lhd^{\ast}$ is also a strict linear extension of $<$, so let $\prec \defeq \lhd^{\ast}$.

    Let $i \in \NN$ with $\widehat{\pi}$, $\widehat{\pi}^{\prime} \in \rnk^{-1}(i)$ and $\widehat{\pi}^{\dprime} \in S$ such that $\widehat{\pi} <^{\ast} \widehat{\pi}^{\dprime} <^{\ast} \widehat{\pi}^{\prime}$. Note that $\rnk(\widehat{\pi}) = \rnk(\widehat{\pi}^{\prime}) = i$. If $\rnk(\widehat{\pi}^{\dprime}) > i$ then, in particular, $\widehat{\pi}^{\prime} \lhd \widehat{\pi}^{\dprime}$ and hence $\widehat{\pi}^{\prime} \prec \widehat{\pi}^{\dprime}$, so we must have $\rnk(\widehat{\pi}^{\dprime}) \le i$. Similarly, if $\rnk(\widehat{\pi}^{\dprime}) < i$ then, in particular, $\widehat{\pi}^{\dprime} \lhd \widehat{\pi}$ and hence $\widehat{\pi}^{\dprime} \prec \widehat{\pi}$, so we must have $\rnk(\widehat{\pi}^{\dprime}) \ge i$. Therefore $\rnk(\widehat{\pi}^{\dprime}) = i$, so $\widehat{\pi}^{\dprime} \in \rnk^{-1}(i)$, and it follows that $\rnk^{-1}(i)$ is an interval in $(S,\prec)$.
\end{proof}

\begin{example}[Well-ordered distribution hierarchy]
Let $X \defeq (V,E,\mu_E)$ be a multihypergraph. Continuing from Example 16 in the main document, there are six well-orderings on $S$ determined by the six permutations of the set $\{\widehat{\sigma}, \widehat{\mu}, \widehat{\rho}\}$. For example, one well ordering $\prec^{\ast}$ on $S$ is $\widehat{\sigma} \prec^{\ast} \widehat{\mu} \prec^{\ast} \widehat{\rho} \prec^{\ast} \widehat{\mu}^{\,\Minus / \Plus}$.
\end{example}

\section{Correspondence between systems hypergraphs and bipartite graphs}\label{app:hypergraphs-bipartite-graphs}

In this section we establish a technical result which formalises the relationship between AMT-type systems hypergraphs and directed bipartite graphs with arc weight functions and second-part vertex weight functions.

We first recall the standard notions of homomorphism for both hypergraphs and graded posets. Given two hypergraphs $(V,E)$ and $(V^{\prime},E^{\prime})$ a hypergraph homomorphism is a vertex map $\phi \colon V \to V^{\prime}$ such that $e \in E$ implies $\phi(e) \in E^{\prime}$; the homomorphism $\phi$ is an isomorphism when $\phi$ is bijective and, for every $e \subseteq V$, $e \in E$ if and only if $\phi(e) \in E^{\prime}$. Given two graded posets $(P,<,\rnk)$ and $(P^{\prime},<^{\prime},\rnk^{\prime})$ a graded poset homomorphism is a map $f \colon P \to P^{\prime}$ that is order preserving, so $p < q$ implies $f(p) <^{\prime} f(q)$ for all $p$, $q \in P$, and rank preserving, so $\rnk^{\prime}(f(p)) = \rnk(p)$ for all $p \in P$.

\begin{definition}[Systems hypergraph isomorphism]
    Let $H \defeq (V,E,\mu_E,Q,\chi)$ where $Q \defeq (S,<,\rnk)$ and $\chi \colon E \to \bigcup_{e \in E} \{ Q_e \rightarrow \NN \}$, and $H^{\prime} \defeq (V^{\prime},E^{\prime},\mu^{\prime}_{E^{\prime}},Q^{\prime},\chi^{\prime})$ where $Q^{\prime} \defeq (S^{\prime},<^{\prime},\rnk^{\prime})$ and $\chi^{\prime} \colon E^{\prime} \to \bigcup_{e^{\prime} \in E^{\prime}} \{ Q^{\prime}_{e^{\prime}} \rightarrow \NN \}$ be two systems hypergraphs. A \emph{systems hypergraph isomorphism} $(\phi,\psi) \colon H \to H^{\prime}$ consists of a multihypergraph isomorphism $\phi \colon (V,E,\mu_E) \to (V^{\prime},E^{\prime},\mu^{\prime}_{E^{\prime}})$, which preserves multihyperedges in the sense that $\mu^{\prime}_{E^{\prime}}\big(\phi(e)\big) = \mu_E(e)$ for all $e \in E$, and an order isomorphism $\psi \colon Q \to Q^{\prime}$ of graded posets, such that:
        \begin{itemize}[topsep=3pt,itemsep=1pt,leftmargin=12pt]
            \item (\textbf{Preservation of the distribution functions}) If $\widehat{\pi} \in S$ is an $n$-distribution function for $X \defeq (V,E,\mu_E)$ with respect to the sequence of domain functions $(\lambda_i)_{i=1}^n$ for $X$, and the sequence of attribute sets $(\AA_i)_{i=1}^n$, for some $n \in \NN_+$, then $\widehat{\pi}^{\prime} \defeq \psi(\widehat{\pi}) \in S^{\prime}$ is an $n$-distribution function with respect to the domain functions $(\lambda^{\prime}_i)_{i=1}^n$ and attribute sets $(\AA^{\prime}_i)_{i=1}^n$ where:
                \begin{itemize}[topsep=3pt,itemsep=1pt,leftmargin=12pt]
                    \item $\lambda^{\prime}_i = \phi \circ \lambda_i \circ \phi^{-1}$ and $\AA^{\prime}_i = \AA_i$ for $i \in [n]$.
                    \item For $e \in E$, noting that $\Dom(\widehat{\pi}_e) = \bigtimes_{i=1}^n \{ \lambda_i(e) \rightrightarrows \AA_i \}$, $\Dom(\widehat{\pi}^{\prime}_{\phi(e)}) = \bigtimes_{i=1}^n \{ \phi \circ \lambda_i(e) \rightrightarrows \AA_i \}$, and the map  $\Phi_{\hat{\pi},e} \colon \Dom(\widehat{\pi}_e) \to \Dom(\widehat{\pi}^{\prime}_{\phi(e)})$ where $\pi \mapsto \pi \circledcirc \phi^{-1}$ for $\pi \in \Dom(\widehat{\pi}_e)$ is a bijection with inverse map $\pi^{\prime} \mapsto \pi^{\prime} \circledcirc \phi$ for $\pi^{\prime} \in \Dom(\widehat{\pi}^{\prime}_{\phi(e)})$, we have $\widehat{\pi}^{\prime}_{\phi(e)}\big(\Phi_{\hat{\pi},e}(\pi)\big) = \widehat{\pi}_e(\pi)$ for $\pi \in \Dom(\widehat{\pi}_e)$.
                \end{itemize}
            \item (\textbf{Preservation of the association function}) Denoting for each $e \in E$ the bijection $\Psi_e \colon Q_e \to Q^{\prime}_{\phi(e)}$ where $\Psi_e \big((S_e,<_e,\rnk_e)\big) = (S^{\prime}_{\phi(e)},<^{\prime}_{\phi(e)},\rnk^{\prime}_{\phi(e)})$ with $S^{\prime}_{\phi(e)} \defeq \{\, \Phi_{\hat{\pi},e}(\pi) \mid \pi \in S_e \cap \Dom(\widehat{\pi}_e) \,\}$, $\pi <_e \eta$ if and only if $\Phi_{\hat{\pi},e}(\pi) <_{\phi(e)} \Phi_{\hat{\eta},e}(\eta)$ for all $\pi \in S_e \cap \Dom(\widehat{\pi}_e)$ and $\eta \in S_e \cap \Dom(\widehat{\eta}_e)$, and $\rnk_e(\pi) = \rnk_{\phi(e)}(\Phi_{\hat{\pi},e}(\pi))$ for all $\pi \in S_e \cap \Dom(\widehat{\pi}_e)$, we have $\chi^{\prime}_{\phi(e)}\big(\Psi_e \big((S_e,<_e,\rnk_e)\big)\big) = \chi_e((S_e,<_e,\rnk_e))$ for $e \in E$ and $(S_e,<_e,\rnk_e) \in Q_e$.
            \end{itemize}
    For convenience, and with regard to our terminology for bipartite graph isomorphisms, we refer to systems hypergraph isomorphisms without preservation of vertex labels as \emph{type-0}, and systems hypergraph isomorphisms that preserve vertex labels as \emph{type-1}.
\end{definition}

\begin{notation}[Isomorphism classes of systems hypergraphs]
    Let $\Hh$ be the set of all vertex-labelled systems hypergraphs. Let $R_t$ be the equivalence relation on $\Hh$ such that $H \mathrel{R_t} H^{\prime}$ if and only if there exists a type-$t$ systems hypergraph isomorphism $(\phi,\psi) \colon H \to H^{\prime}$, for all $H$, $H^{\prime} \in \Hh$. Then the quotient set $\widetilde{\Hh} \defeq \Hh/R_t = \{\, [H]_t \mid H \in \Hh \,\}$ consists of the type-$t$ isomorphism classes of $\Hh$. Each type-$0$ equivalence class is an \emph{unlabelled} systems hypergraph. Note that each type-1 equivalence class $[H]_1 \in \widetilde{\Hh}$ is a singleton, that is $[H]_1 = \{H\}$. 
\end{notation}

\begin{definition}[Orientation maps induced from directed bipartite graphs]
    Let $D \defeq (V,W,F)$ be a directed bipartite graph. For each $x \in V \cup W$ let $\sigma_x \colon \N(x) \to \I$ be the \emph{orientation map} such that, for $y \in \N(x)$:
    \begin{equation*}
        \sigma_x(y) \defeq
            \begin{cases} \{-1\} & \text{if $y \in \N_{\text{in}}(x) \setminus \N_{\text{out}}(x)$},\\ \{+1\} & \text{if $y \in \N_{\text{out}}(x) \setminus \N_{\text{in}}(x)$},\\ \{-1,+1\} & \text{if $y \in \N_{\text{in}}(x) \cap \N_{\text{out}}(x)$}.
            \end{cases}
    \end{equation*}
    Note that if $\N(x) = \emptyset$ then $\sigma_x$ is the empty function to $\I$.
\end{definition}

\begin{definition}[Multiset maps induced from directed bipartite graphs with arc weights]
    Let $D \defeq (V,W,F,\omega)$ be a directed bipartite graph with arc weight function $\omega \colon F \to \AA$ where $\AA \subseteq \RR$ is nonempty and closed under addition. For each $x \in V \cup W$ let:
        \begin{itemize}[topsep=3pt,itemsep=1pt,leftmargin=12pt]
            \item $\mu_x \colon \N(x) \to \AA$ be the \emph{multiset map} such that, for $y \in \N(x)$, $\mu_x(y) \defeq \sum_{e \in F(x) \cap F(y)} \omega(e)$.
            \item $\mu^\Minus_x \colon \N(x) \to \AA$ be the \emph{negative oriented-multiset map} such that $\Dom(\mu^\Minus_x) \defeq \N_{\text{in}}(x)$, hence $\mu^\Minus_x$ may be a partial map, and $\mu^\Minus_x(y) \defeq \omega(y,x)$ for $y \in \N_{\text{in}}(x)$.
            \item $\mu^\Plus_x \colon \N(x) \to \AA$ be the \emph{positive oriented-multiset map} such that $\Dom(\mu^\Plus_x) \defeq \N_{\text{out}}(x)$, hence $\mu^\Plus_x$ may be a partial map, and $\mu^\Plus_x(y) \defeq \omega(x,y)$ for $y \in \N_{\text{out}}(x)$.
        \end{itemize}
    Note that if $\N(x) = \emptyset$ then $\mu_x$, $\mu^\Minus_x$, and $\mu^\Plus_x$ are empty functions, and if $\N(x) \ne \emptyset$ with $\N_{\text{in}}(x) = \emptyset$ (resp. $\N_{\text{out}}(x) = \emptyset$) then $\mu^\Minus_x$ (resp. $\mu^\Plus_x$) is the empty partial function.
\end{definition}

\begin{definition}[Hyperedge label maps induced from directed bipartite graphs with vertex weights]
    Let $D \defeq (V,W,F,\delta)$ be a directed bipartite graph with vertex weight function $\delta \colon V \cup W \to \VV$ where $\VV \subseteq \RR$ is nonempty. For each $x \in V \cup W$ let $\rho_x \colon \big\{\N(x)\big\} \to \VV$ be the \emph{hyperedge label map} such that $\rho_x\big(\N(x)\big) \defeq \delta(x)$.
\end{definition}

\begin{definition}[Attribute hierarchy induced from directed bipartite graphs with arc and vertex weights]
    Let $D \defeq (V,W,F,\omega,\delta)$ be a directed bipartite graph with arc weight function $\omega \colon F \to \AA$ where $\AA \subseteq \RR$ is nonempty and closed under addition, and with vertex weight function $\delta \colon V \cup W \to \VV$ where $\VV \subseteq \RR$ is nonempty. For each $x \in V \cup W$ the \emph{induced attribute hierarchy} $(S_x,<_x,\rnk_x)$ satisfies $S_x \defeq \{ \sigma_x,\mu_x,\rho_x,(\mu_x^\Minus,\mu_x^\Plus) \}$, $\rnk_x \colon S_x \to \NN$ with $\rnk_x(\sigma_x) = \rnk_x(\mu_x) = \rnk_x(\rho_x) = 0$ and $\rnk_x((\mu_x^\Minus,\mu_x^\Plus)) = 1$, and the comparable attribute sequences are $\sigma_x <_x (\mu_x^\Minus,\mu_x^\Plus)$ and $\mu_x <_x (\mu_x^\Minus,\mu_x^\Plus)$.
\end{definition}

\begin{theorem}[Correspondence between AMT-type systems hypergraphs and directed bipartite graphs with arc weights and second-part weights] \label{theorem:Corr-SH-DBGW}
    Let $t \in \{0,1\}$, let $\B$ be the set of all vertex-labelled directed bipartite graphs with arc weight functions and second-part weight functions, and type-$t$ isomorphism classes $\widetilde{\B}$, and let $\Hh$ be the set of all AMT-type systems hypergraphs with type-$t$ isomorphism classes $\widetilde{\Hh}$. There exists a surjection $K \colon \B \to \Hh$ such that:
    \begin{enumerate}[label=(\arabic*)]
        \item For $(V,W,F,(\omega,\AA),(\delta,\VV)) \in \B$ we have $K\big((V,W,F,\omega,\delta)\big) = (V,E,\mu_E,Q,\chi) \in \Hh$, where:
            \begin{enumerate}[label=(\arabic{enumi}.\arabic*),topsep=0pt]
                \item $E \defeq \{\, \mathcal{N}(w) \mid w \in W \,\}$.
                \item $\mu_E \colon E \to \NN_+$ satisfies $\mu_E\big(\mathcal{N}(w)\big) \defeq \left| \T(w) \right|$ for $w \in W$.
                \item $\widehat{\sigma} \colon E \to \bigcup_{e \in E} \{ \{ e \rightarrow \I \} \rightarrow \NN \}$ satisfies $\widehat{\sigma}_{\N(w)}^{-1} (\NN_+) = \{\, \sigma_z \mid z \in \T(w) \,\}$ and $\widehat{\sigma}_{\N(w)} (\sigma_z) = \left| \T_{\text{dir}}(z) \right|$, for $w \in W$ and $\sigma_z \in {\widehat{\sigma}_{\N(w)}}^{-1} (\NN_+)$.
                \item $\widehat{\mu} \colon E \to \bigcup_{e \in E} \{ \{e \rightarrow \AA\} \rightarrow \NN \}$ satisfies $\widehat{\mu}_{\N(w)}^{-1} (\NN_+) = \{\, \mu_z \mid z \in \T(w) \,\}$ and $\widehat{\mu}_{\N(w)} (\mu_z) = \left| \T_{\text{awgt}}(z) \right|$, for $w \in W$ and $\mu_z \in {\widehat{\mu}_{\N(w)}}^{-1} (\NN_+)$.
                \item $\widehat{\rho} \colon E \to \bigcup_{e \in E} \{ \{ \{e\} \rightarrow \VV \} \rightarrow \NN \}$ satisfies $\widehat{\rho}_{\N(w)}^{-1} (\NN_+) = \{\, \rho_z \mid z \in \T(w) \,\}$ and $\widehat{\rho}_{\N(w)} (\rho_z) = \left| \T_{\text{vwgt}}(z) \right|$, for $w \in W$ and $\rho_z \in {\widehat{\rho}_{\N(w)}}^{-1} (\NN_+)$.
                \item $\widehat{\mu}^{\,\Minus / \Plus} \colon E \to \bigcup_{e \in E} \{ \{ e \rightharpoonup \AA \} \times \{ e \rightharpoonup \AA \} \rightarrow \NN \}$ satisfies $\widehat{\mu}^{\,\Minus / \Plus \; -1}_{\N(w)} (\NN_+) = \{\, (\mu^\Minus_z,\mu^\Plus_z) \mid z \in \T(w) \,\}$ and $\widehat{\mu}^{\,\Minus / \Plus}_{\N(w)} ((\mu^\Minus_z,\mu^\Plus_z)) = \left| \T_{\text{in}}(z) \cap \T_{\text{out}}(z) \right|$, for $w \in W$ and $(\mu^\Minus_z,\mu^\Plus_z) \in \widehat{\mu}^{\,\Minus / \Plus \; -1}_{\N(w)} (\NN_+)$.
                \item $\chi \colon E \to \bigcup_{e \in E} \{ Q_e \rightarrow \NN \}$ satisfies $\chi_{\N(w)}^{-1} (\NN_+) = \{\, (S_z,<_z,\rnk_z) \mid z \in \T(w) \,\}$ and $\chi_{\N(w)}\big((S_z,<_z,\rnk_z)\big) =\\ \left| \T_{\text{dir}}(z) \cap \T_{\text{awgt}}(z) \cap \T_{\text{vwgt}}(z) \cap \T_{\text{in}}(z) \cap \T_{\text{out}}(z) \right|$, for $w \in W$.
            \end{enumerate}
        \item A right inverse $K_r \colon \Hh \to \B$ of $K$ is given by $K_r\big((V,E,\mu_E,Q,\chi)\big) = (V,W,F,(\omega,\AA),(\delta,\VV)) \in \B$, for $(V,E,\mu_E,Q,\chi) \in \Hh$ with $S = \{(\widehat{\sigma},\I),(\widehat{\mu},\AA),(\widehat{\rho},\VV),(\widehat{\mu}^{\,\Minus / \Plus},\AA)\}$, where:
            \begin{enumerate}[label=(\arabic{enumi}.\arabic*),topsep=0pt]
                \item $W \defeq \big\{\, \big(e,(S_e,<_e,\rnk_e),i\big) \mid \text{$e \in E$, $(S_e,<_e,\rnk_e) \in \chi_e^{-1}(\NN_+)$, and $i \in [\chi_e\big((S_e,<_e,\rnk_e)\big)]$} \,\big\}$.
                \item $F \subseteq (V \times W) \sqcup (W \times V)$ such that $(v,(e,(S_e,<_e,\rnk_e),i)) \in F$ if and only if $v \in e$ and $-1 \in \sigma(v)$ where $\sigma \in S_e$, and $((e,(S_e,<_e,\rnk_e),i),v) \in F$ if and only if $v \in e$ and $+1 \in \sigma(v)$ where $\sigma \in S_e$.
                \item $\omega \colon F \to \AA$ is given by $\omega\big((v,(e,(S_e,<_e,\rnk_e),i))\big) = \mu^\Minus(v)$ if $(v,(e,(S_e,<_e,\rnk_e),i)) \in F$ and $\omega\big(((e,(S_e,<_e,\rnk_e),i),v)\big) = \mu^\Plus(v)$ if $((e,(S_e,<_e,\rnk_e),i),v) \in F$, where $(\mu^\Minus,\mu^\Plus) \in S_e$.
                \item $\delta \colon V \cup W \to \VV$ with $\Dom(\delta) = W$ is given by $\delta\big((e,(S_e,<_e,\rnk_e),i)\big) = \rho(e)$ for all $(e,(S_e,<_e,\rnk_e),i\big) \in W$ where $\rho \in S_e$.
            \end{enumerate} 
        \item $K_r \circ K \colon \B \to \B$ satisfies $K_r \circ K(B) \cong_t B$ for $B \in \B$. \label{thm:KrK}
        \item $K$ induces a bijection $\widetilde{K} \colon \widetilde{\B} \to \widetilde{\Hh}$ such that:
            \begin{enumerate}[label=(\arabic{enumi}.\arabic*),topsep=0pt]
                \item $\widetilde{K}\big([B]_t\big) = [K(B)]_t$ for $[B]_t \in \widetilde{\B}$.
                \item $\widetilde{K}^{-1}\big([H]_t\big) = [K_r(H)]_t$ for $[H]_t \in \widetilde{\Hh}$. \label{thm:eqKinv}
            \end{enumerate}
    \end{enumerate}
\end{theorem}

\begin{proof}
The proof consists of six steps, where we show:
    \begin{enumerate*}[label={(\arabic*)}]
        \item The map $K$ is well defined. \label{K}
        \item The map $K_r$ is a well-defined right inverse for $K$. \label{Kr}
        \item $K_r \circ K(B)$ and $B$ are type-$t$ isomorphic for $B \in \B$. \label{KrK}
        \item The map $\widetilde{K} \colon \widetilde{\B} \to \widetilde{\Hh}$ is well defined. \label{eqK}
        \item The map $\widetilde{L} \colon \widetilde{\Hh} \to \widetilde{\B}$ in Part~\ref{thm:eqKinv} of this theorem is well defined. \label{eqKinv}
        \item $\widetilde{L} = \widetilde{K}^{-1}$. \label{eqKbij}
    \end{enumerate*}

    \vspace{0.2cm}

    Step~\ref{K}: To show $K$ is well defined we need to establish that $K\big((V,W,F,\omega,\delta)\big) = (V,E,\mu_E,Q,\chi)$ is an AMT-type systems hypergraph for $(V,W,F,(\omega,\AA),(\delta,\VV)) \in \B$.

    Note that $\mu_E$ is a multiplicity function, $\widehat{\sigma}$ is a multihypergraph orientation function, $\widehat{\mu}$ is a multihypergraph multiset function, $\widehat{\rho}$ is a multihypergraph hyperedge label function, and $\widehat{\mu}^{\,\Minus / \Plus}$ is a multihypergraph negative/positive-oriented multiset function.

    We show the map $\chi$ is an AMT-type association function for $(V,E,\mu_E)$ with respect to $Q$. Let $w \in W$ and $z \in \T(w)$. Consider $\widehat{\sigma}$ and note that $\sigma_z \in \Dom(\widehat{\sigma}_{\N(w)})$. Then $\sum_{\{\, (S_x,<_x,\rnk_x) \in Q_{\N(w)} \mid \sigma_z \in S_x \,\}} \chi_{\N(w)}((S_x,<_x,\rnk_x)) = \sum_{x \in \T_{\text{dir}}(z)} \chi_{\N(w)}((S_x,<_x,\rnk_x)) = \sum_{x \in \T_{\text{dir}}(z)} \left| \T_{\text{dir}}(x) \cap \T_{\text{awgt}}(x) \cap \T_{\text{vwgt}}(x) \cap \T_{\text{in}}(x) \cap \T_{\text{out}}(x) \right| = \left| \T_{\text{dir}}(z) \right| = \widehat{\sigma}_e (\sigma_z)$, since the family of sets $\T_{\text{dir}}(x) \cap \T_{\text{awgt}}(x) \cap \T_{\text{vwgt}}(x) \cap \T_{\text{in}}(x) \cap \T_{\text{out}}(x)$ for $x \in \T_{\text{dir}}(z)$ is a partition of $\T_{\text{dir}}(z)$. Similar arguments hold for $ \widehat{\mu}$, $ \widehat{\rho}$, and $\widehat{\mu}^{\,\Minus / \Plus}$ with corresponding partitions of $\T_{\text{awgt}}(z)$, $\T_{\text{vwgt}}(z)$, and $\T_{\text{in}}(z) \cap \T_{\text{out}}(z)$, respectively. Further, $\Dom(\mu^\Minus_z) = \N_{\text{in}}(z) = \{\, v \in \N(z) \mid -1 \in \sigma_z(v) \,\}$, $\Dom(\mu^\Plus_z) = \N_{\text{out}}(z) = \{\, v \in \N(z) \mid +1 \in \sigma_z(v) \,\}$, $\mu^\Minus_z(v) = \omega(v,z) = \sum_{e \in F(v) \cap F(z)} \omega(e) = \mu_z(v)$ when $v \in \N(z)$ and $\sigma_z(v) = \{-1\}$, $\mu^\Plus_z(v) = \omega(z,v) = \sum_{e \in F(z) \cap F(v)} \omega(e) = \mu_z(v)$ when $v \in \N(z)$ and $\sigma_z(v) = \{+1\}$, and $\big(\mu^\Minus_z(v),\mu^\Plus_z(v)\big) \in C(\mu_z(v),2)$ since $\mu^\Minus_z(v) + \mu^\Plus_z(v) = \omega(v,z) + \omega(z,v) = \sum_{e \in F(v) \cap F(z)} \omega(e) = \mu_z(v)$ when $v \in \N(z)$ and $\sigma_z(v) = \{-1,+1\}$.

    \vspace{0.2cm}

    Step~\ref{Kr}: The map $K_r$ is well defined, so we show that $K \circ K_r$ is the identity map on $\Hh$. Consider $(K \circ K_r) \big((V,E,\mu_E,Q,\chi)\big) = K\big((V,W,F,\omega,\delta)\big) = (V,E^{\prime},\mu^{\prime}_{E^{\prime}},Q^{\prime},\chi^{\prime})$.

    For $E^{\prime} = E$, note that $\N\Big(\big(e,(S_e,<_e,\rnk_e),i\big)\Big) = e$ for $\big(e,(S_e,<_e,\rnk_e),i\big) \in W$.

    For $\mu^{\prime}_{E^{\prime}} = \mu_E$, let $e \in E^{\prime} = E$ and $w \in W$ with $e = \mathcal{N}(w)$, and it suffices to show $\mu^{\prime}_{E^{\prime}}(e) = \mu_E$(e). So $\mu^{\prime}_{E^{\prime}}(e) = \mu^{\prime}_{E^{\prime}}\big(\N(w)\big) = \left| \T(w) \right| = \left| \{\, x \in W \mid \N(x) = \N(w) \,\} \right| = \left| \{\, \big(f,(S_f,<_f,\rnk_f),i\big) \in W \mid f = e \,\} \right| = \bigl| \{\, \big(f,(S_f,<_f,\rnk_f),i\big) \in W \mid \text{$(S_f,<_f,\rnk_f) \in \chi_e^{-1}(\NN_+)$ and $i \in [\chi_e((S_f,<_f,\rnk_f))]$} \,\} \bigr| = \sum_{(S_f,<_f,\rnk_f) \in Q_e} \chi_e((S_f,<_f,\rnk_f)) = \mu_E(e)$.

    For $Q^{\prime} = Q$, since both hypergraphs are AMT type, it suffices to show $ \widehat{\sigma}^{\prime} = \widehat{\sigma}$, $\widehat{\mu}^{\prime} = \widehat{\mu}$, $\widehat{\rho}^{\prime} = \widehat{\rho}$, and $\widehat{\mu}^{\,\Minus / \Plus \; \prime} = \widehat{\mu}^{\,\Minus / \Plus}$.

    For $\widehat{\sigma}^{\prime} = \widehat{\sigma}$, let $e \in E^{\prime} = E$, and it suffices to show $\widehat{\sigma}^{\prime}_e = \widehat{\sigma}_e$. If $\sigma \in \widehat{\sigma}_e^{-1}(\NN_+)$ then, choosing some $(S_f,<_f,\rnk_f) \in \chi_e^{-1}(\NN_+)$ and $i \in [\chi_e((S_f,<_f,\rnk_f))]$ with $\sigma \in S_f$, let $w \defeq (e,(S_f,<_f,\rnk_f),i) \in W$. Note that $\N(w) = e$, $\N_{\text{in}}(w) = \{\, x \in e \mid -1 \in \sigma(x) \,\}$, and $\N_{\text{out}}(w) = \{\, x \in e \mid +1 \in \sigma(x) \,\}$, so $\sigma = \sigma_w$, hence $\sigma = \sigma_w \in \widehat{\sigma}^{\prime \, -1}_e(\NN_+)$. Conversely, suppose $\sigma_z \in \widehat{\sigma}^{\prime \, -1}_e (\NN_+)$ for some $z \in \T(x)$ where $x \in W$ with $\N(x) = e$. Then $z = (e,(S_g,<_g,\rnk_g),j) \in W$ for some $(S_g,<_g,\rnk_g) \in \chi_e^{-1}(\NN_+)$ and $j \in [\chi_e((S_g,<_g,\rnk_g))]$. As for the previous direction, noting $\N(z) = e$, we then have $\sigma_z = \sigma$ where $\sigma \in S_g$. Since $\widehat{\sigma}_e(\sigma) \ge \chi_e((S_g,<_g,\allowbreak\rnk_g)) > 0$ we have $\sigma_z = \sigma \in \widehat{\sigma}_e^{-1}(\NN_+)$. By these two arguments we therefore have $\widehat{\sigma}_e^{-1}(\NN_+) = \widehat{\sigma}^{\prime \, -1}_e(\NN_+)$. Now, let $\sigma_z \in \widehat{\sigma}^{\prime \, -1}_e(\NN_+) = \widehat{\sigma}_e^{-1}(\NN_+)$ for some $z \in W$ with $\N(z) = e$. Then $z = (e,(S_f,<_f,\rnk_f),i)$ for some $(S_f,<_f,\rnk_f) \in \chi_e^{-1}(\NN_+)$, and $i \in [\chi_e((S_f,<_f,\rnk_f))]$. Note that $\sigma_z = \sigma$ where $\sigma \in S_f$, so $\widehat{\sigma}^{\prime}_e(\sigma_z) = \left| \T_{\text{dir}}(z) \right| = \left| \{\, (e,(S_g,<_g,\rnk_g),j) \in W \mid \text{ $(S_g,<_g,\rnk_g) \in \chi_e^{-1}(\NN_+)$ with $\sigma_z \in S_g$ and $j \in [\chi_e((S_g,<_g,\rnk_g))]$} \,\} \right| = \sum_{\{\, (S_g,<_g,\rnk_g) \in Q_e \mid \sigma_z \in S_g \,\}}  \allowbreak \chi_e((S_g,<_g,\rnk_g)) = \widehat{\sigma}_e(\sigma_z)$. We conclude that $\widehat{\sigma}^{\prime}_e = \widehat{\sigma}_e$.

    For $\widehat{\mu}^{\prime} = \widehat{\mu}$, let $e \in E^{\prime} = E$, and it suffices to show $\widehat{\mu}^{\prime}_e = \widehat{\mu}_e$. If $\mu \in \widehat{\mu}_e^{-1}(\NN_+)$ then, choosing some $(S_f,<_f,\rnk_f) \in \chi_e^{-1}(\NN_+)$ and $i \in [\chi_e((S_f,<_f,\rnk_f))]$ with $\mu \in S_f$, let $w \defeq (e,(S_f,<_f,\rnk_f),i) \in W$. Let $S_f = \{\sigma,\mu,\rho,(\mu^{\Minus},\mu^{\Plus})\}$, and note that $\N(w) = e$. For $y \in \N(w)$ we have
    \begin{equation*}
            \mu_w(y) = \sum_{g \in F(w) \cap F(y)} \omega(g)
            = \begin{cases}
                \mu^\Minus(y) & \text{if $(y,w) \in F$ and $(w,y) \notin F$},\\
                \mu^\Plus(y) & \text{if $(w,y) \in F$ and $(y,w) \notin F$},\\
                \mu^\Minus(y) + \mu^\Plus(y) & \text{if $(y,w)$, $(w,y) \in F$}
            \end{cases}
            = \begin{cases}
                \mu^\Minus(y) & \text{if $\sigma(y) = \{-1\}$},\\
                \mu^\Plus(y) & \text{if $\sigma(y) = \{+1\}$},\\
                \mu^\Minus(y) + \mu^\Plus(y) & \text{if $\sigma(y) = \{-1,+1\}$}
            \end{cases}
            = \mu(y),
        \end{equation*}
    hence $\mu = \mu_w \in \widehat{\mu}^{\prime \, -1}_e(\NN_+)$. Conversely, suppose $\mu_z \in \widehat{\mu}^{\prime \, -1}_e (\NN_+)$ for some $z \in \T(x)$ where $x \in W$ with $\N(x) = e$. Then $z = (e,(S_g,<_g,\rnk_g),j) \in W$ for some $(S_g,<_g,\rnk_g) \in \chi_e^{-1}(\NN_+)$ and $j \in [\chi_e((S_g,<_g,\rnk_g))]$. As for the previous direction, noting $\N(z) = e$, we then have $\mu_z = \mu$ where $\mu \in S_g$. Since $\widehat{\mu}_e(\mu) \ge \chi_e((S_g,<_g,\allowbreak\rnk_g)) > 0$ we have $\mu_z = \mu \in \widehat{\mu}_e^{-1}(\NN_+)$. By these two arguments we therefore have $\widehat{\mu}_e^{-1}(\NN_+) = \widehat{\mu}^{\prime \, -1}_e(\NN_+)$. Now, let $\mu_z \in \widehat{\mu}^{\prime \, -1}_e(\NN_+) = \widehat{\mu}_e^{-1}(\NN_+)$ for some $z \in W$ with $\N(z) = e$. Then $z = (e,(S_f,<_f,\rnk_f),i)$ for some $(S_f,<_f,\rnk_f) \in \chi_e^{-1}(\NN_+)$, and $i \in [\chi_e((S_f,<_f,\rnk_f))]$. Note that $\mu_z = \mu$ where $\mu \in S_f$, so $\widehat{\mu}^{\prime}_e(\mu_z) = \left| \T_{\text{awgt}}(z) \right| = \left| \{\, (e,(S_g,<_g,\rnk_g),j) \in W \mid \text{ $(S_g,<_g,\rnk_g) \in \chi_e^{-1}(\NN_+)$ with $\mu_z \in S_g$ and $j \in [\chi_e((S_g,<_g,\rnk_g))]$} \,\} \right| = \sum_{\{\, (S_g,<_g,\rnk_g) \in Q_e \mid \mu_z \in S_g \,\}} \chi_e((S_g,<_g,\rnk_g)) = \widehat{\mu}_e(\mu_z)$. We conclude that $\widehat{\mu}^{\prime}_e = \widehat{\mu}_e$.

    For $\widehat{\rho}^{\prime} = \widehat{\rho}$, let $e \in E^{\prime} = E$, and it suffices to show $\widehat{\rho}^{\prime}_e = \widehat{\rho}_e$. If $\rho \in \widehat{\rho}_e^{-1}(\NN_+)$ then, choosing some $(S_f,<_f,\rnk_f) \in \chi_e^{-1}(\NN_+)$ and $i \in [\chi_e((S_f,<_f,\rnk_f))]$ with $\rho \in S_f$, let $w \defeq (e,(S_f,<_f,\rnk_f),i) \in W$. Note that $\N(w) = e$. Then $\rho_w(e) = \delta(w) = \rho(e)$, hence $\rho = \rho_w \in \widehat{\rho}^{\prime \, -1}_e(\NN_+)$. Conversely, suppose $\rho_z \in \widehat{\rho}^{\prime \, -1}_e (\NN_+)$ for some $z \in \T(x)$ where $x \in W$ with $\N(x) = e$. Then $z = (e,(S_g,<_g,\rnk_g),j) \in W$ for some $(S_g,<_g,\rnk_g) \in \chi_e^{-1}(\NN_+)$ and $j \in [\chi_e((S_g,<_g,\rnk_g))]$. As for the previous direction, noting $\N(z) = e$, we then have $\rho_z = \rho$ where $\rho \in S_g$. Since $\widehat{\rho}_e(\rho) \ge \chi_e((S_g,<_g,\rnk_g)) > 0$ we have $\rho_z = \rho \in \widehat{\rho}_e^{-1}(\NN_+)$. By these two arguments we therefore have $\widehat{\rho}_e^{-1}(\NN_+) = \widehat{\rho}^{\prime \, -1}_e(\NN_+)$. Now, let $\rho_z \in \widehat{\rho}^{\prime \, -1}_e(\NN_+) = \widehat{\rho}_e^{-1}(\NN_+)$ for some $z \in W$ with $\N(z) = e$. Then $z = (e,(S_f,<_f,\rnk_f),i)$ for some $(S_f,<_f,\rnk_f) \in \chi_e^{-1}(\NN_+)$, and $i \in [\chi_e((S_f,<_f,\rnk_f))]$. Note that $\rho_z = \rho$ where $\rho \in S_f$, so $\widehat{\rho}^{\prime}_e(\rho_z) = \left| \T_{\text{vwgt}}(z) \right| = \left| \{\, (e,(S_g,<_g,\rnk_g),j) \in W \mid \text{ $(S_g,<_g,\rnk_g) \in \chi_e^{-1}(\NN_+)$ with $\rho_z \in S_g$ and $j \in [\chi_e((S_g,<_g,\rnk_g))]$} \,\} \right| = \sum_{\{\, (S_g,<_g,\rnk_g) \in Q_e \mid \rho_z \in S_g \,\}} \chi_e((S_g,<_g,\rnk_g)) = \widehat{\rho}_e(\rho_z)$. We conclude that $\widehat{\rho}^{\prime}_e = \widehat{\rho}_e$.

    For $\widehat{\mu}^{\,\Minus / \Plus \; \prime} = \widehat{\mu}^{\,\Minus / \Plus}$, let $e \in E^{\prime} = E$, and it suffices to show $\widehat{\mu}^{\,\Minus / \Plus \; \prime}_e = \widehat{\mu}^{\,\Minus / \Plus}_e$. If $(\mu^\Minus,\mu^\Plus) \in \widehat{\mu}^{\,\Minus / \Plus \, -1}_e(\NN_+)$ then, choosing some $(S_f,<_f,\rnk_f) \in \chi_e^{-1}(\NN_+)$ and $i \in [\chi_e((S_f,<_f,\rnk_f))]$ with $(\mu^\Minus,\mu^\Plus) \in S_f$, let $w \defeq (e,(S_f,<_f,\rnk_f),i) \in W$. Note that $\N(w) = e$ and $\Dom(\mu^\Minus_w) = \N_{\text{in}}(w) = \{\, v \in e \mid -1 \in \sigma(v) \,\} = \Dom(\mu^\Minus)$. For $y \in \N_{\text{in}}(w)$ we have $\mu^\Minus_w(y) = \omega(y,w) = \mu^\Minus(y)$, hence $\mu^\Minus_w = \mu^\Minus$. Similarly, $\mu^\Plus_w = \mu^\Plus$. Hence $(\mu^\Minus,\mu^\Plus) = (\mu^\Minus_w,\mu^\Plus_w) \in \widehat{\mu}^{\,\Minus / \Plus \; \prime \; -1}_e(\NN_+)$. Conversely, suppose $(\mu^\Minus_z,\mu^\Plus_z) \in \widehat{\mu}^{\,\Minus / \Plus \; \prime \; -1}_e(\NN_+)$ for some $z \in \T(x)$ where $x \in W$ with $\N(x) = e$. Then $z = (e,(S_g,<_g,\rnk_g),j) \in W$ for some $(S_g,<_g,\rnk_g) \in \chi_e^{-1}(\NN_+)$ and $j \in [\chi_e((S_g,<_g,\rnk_g))]$. As for the previous direction, noting $\N(z) = e$, we then have $(\mu^\Minus_z,\mu^\Plus_z) = (\mu^\Minus,\mu^\Plus)$ where $(\mu^\Minus,\mu^\Plus) \in S_g$. Since $\widehat{\mu}^{\,\Minus / \Plus}_e \big((\mu^\Minus,\mu^\Plus)\big) \ge \chi_e((S_g,<_g,\rnk_g)) > 0$ we have $(\mu^\Minus_z,\mu^\Plus_z) = (\mu^\Minus,\mu^\Plus) \in \widehat{\mu}^{\,\Minus / \Plus \; -1}_e(\NN_+)$. By these two arguments we therefore have $\widehat{\mu}^{\,\Minus / \Plus \; -1}_e(\NN_+) = \widehat{\mu}^{\,\Minus / \Plus \; \prime \; -1}_e(\NN_+)$. Now, let $(\mu^\Minus_z,\mu^\Plus_z) \in \widehat{\mu}^{\,\Minus / \Plus \; \prime \; -1}_e(\NN_+) = \widehat{\mu}^{\,\Minus / \Plus \; -1}_e(\NN_+)$ for some $z \in W$ with $\N(z) = e$. Then $z = (e,(S_f,<_f,\rnk_f),i)$ for some $(S_f,<_f,\rnk_f) \in \chi_e^{-1}(\NN_+)$, and $i \in [\chi_e((S_f,<_f,\rnk_f))]$. Note that $(\mu^\Minus_z,\mu^\Plus_z) = (\mu^\Minus,\mu^\Plus)$ where $(\mu^\Minus,\mu^\Plus) \in S_f$, so $\widehat{\mu}^{\,\Minus / \Plus \; \prime}_e \big((\mu^\Minus_z,\mu^\Plus_z)\big) = \left| \T_{\text{in}}(z) \cap \T_{\text{out}}(z) \right| = \bigl| \{\, (e,(S_g,<_g,\rnk_g),j) \in W \mid \text{ $(S_g,<_g,\rnk_g) \in \chi_e^{-1}(\NN_+)$ with $(\mu^\Minus_z,\mu^\Plus_z) \in S_g$ and $j \in [\chi_e((S_g,<_g,\rnk_g))]$} \,\} \bigr|  =  \sum_{\{\, (S_g,<_g,\rnk_g) \in  Q_e \mid (\mu^\Minus_z,\mu^\Plus_z) \in S_g \,\}}  \chi_e((S_g,\allowbreak<_g,\rnk_g)) = \widehat{\mu}^{\,\Minus / \Plus}_e\big((\mu^\Minus_z,\mu^\Plus_z)\big)$. We conclude that $\widehat{\mu}^{\,\Minus / \Plus \; \prime}_e = \widehat{\mu}^{\,\Minus / \Plus}_e$. It follows that $Q^{\prime} = Q$.

    For $\chi^{\prime} = \chi$, let $e \in E^{\prime} = E$, and it suffices to show $\chi^{\prime}_e = \chi_e$. If $(S_e,<_e,\rnk_e) \in \chi_e^{-1}(\NN_+)$ then, choosing some $i \in [\chi_e((S_e,<_e,\rnk_e))]$, let $w \defeq (e,(S_e,<_e,\rnk_e),i) \in W$. Note that $\N(w) = e$, and denote $S_e = \{\sigma,\mu,\rho,(\mu^{\Minus},\mu^{\Plus})\}$. Since $(S_e,<_e,\rnk_e) \in \chi_e^{-1}(\NN_+)$, the definition of an association function implies $\sigma \in \widehat{\sigma}_e^{-1}(\NN_+)$, $\mu \in \widehat{\mu}_e^{-1}(\NN_+)$, $\rho \in \widehat{\rho}_e^{-1}(\NN_+)$, and $(\mu^\Minus,\mu^\Plus) \in \widehat{\mu}^{\,\Minus / \Plus \, -1}_e(\NN_+)$. It follows from the previous arguments that $\sigma = \sigma_w \in \widehat{\sigma}^{\prime \, -1}_e(\NN_+)$, $\mu = \mu_w \in \widehat{\mu}^{\prime \, -1}_e(\NN_+)$, $\rho = \rho_w \in \widehat{\rho}^{\prime \, -1}_e(\NN_+)$, and $(\mu^\Minus,\mu^\Plus) = (\mu^\Minus_w,\mu^\Plus_w) \in \widehat{\mu}^{\,\Minus / \Plus \; \prime \; -1}_e(\NN_+)$, therefore $(S_e,<_e,\rnk_e) = (S_w,<_w,\rnk_w) \in \chi^{\prime \, -1}_e (\NN_+)$. Conversely, suppose $(S_z,<_z,\rnk_z) \in \chi^{\prime \, -1}_e (\NN_+)$ for some $z \in \T(x)$ where $x \in W$ with $\N(x) = e$. Then $z = (e,(S_g,<_g,\rnk_g),j) \in W$ for some $(S_g,<_g,\rnk_g) \in \chi_e^{-1}(\NN_+)$ and $j \in [\chi_e((S_g,<_g,\rnk_g))]$. As for the previous direction, noting $\N(z) = e$, we then have $(S_z,<_z,\rnk_z) = (S_g,<_g,\rnk_g) \in \chi^{-1}_e (\NN_+)$. By these two arguments we therefore have $\chi_e^{-1}(\NN_+) = \chi^{\prime \, -1}_e (\NN_+)$. Now, let $(S_z,<_z,\rnk_z) \in \chi^{\prime \, -1}_e (\NN_+) = \chi_e^{-1} (\NN_+)$ for some $z \in W$ with $\N(z) = e$. Then $z = (e,(S_f,<_f,\rnk_f),i)$ for some $(S_f,<_f,\rnk_f) \in \chi_e^{-1}(\NN_+)$, and $i \in [\chi_e((S_f,<_f,\rnk_f))]$. Note that $(S_z,<_z,\rnk_z) = (S_f,<_f,\rnk_f)$, so $\chi^{\prime}_e ((S_z,<_z,\rnk_z)) = \left| \T_{\text{dir}}(z) \cap \T_{\text{awgt}}(z) \cap \T_{\text{vwgt}}(z) \cap \T_{\text{in}}(z) \cap \T_{\text{out}}(z) \right| = \left| \{\, (e,(S_z,<_z,\rnk_z),j) \in W \mid j \in [\chi_e((S_z,<_z,\rnk_z))] \,\} \right| = \chi_e((S_z,<_z,\rnk_z))$. We conclude that $\chi^{\prime}_e = \chi_e $.

    Therefore, $(V,E,\mu_E,Q,\chi) = (V,E^{\prime},\mu^{\prime}_{E^{\prime}},Q^{\prime},\chi^{\prime})$, and it follows that $K \circ K_r$ is the identity map on $\Hh$.

    \vspace{0.2cm}

    Step~\ref{KrK}: Let $B \defeq (V,W,F,(\omega,\AA),(\delta,\VV)) \in \B$, and define $B^{\prime} \defeq K_r \circ K(B) = K_r\big((V,E,\mu_E,Q,\chi)\big) = (V,W^{\prime},F^{\prime},(\omega^{\prime},\AA),(\delta^{\prime},\VV)) \in \B$. We show that $B \cong_1 B^{\prime}$, which implies $B \cong_0 B^{\prime}$. First note that the family of sets consisting of each set $S(x) \defeq \T_{\text{dir}}(x) \cap \T_{\text{awgt}}(x) \cap \T_{\text{vwgt}}(x) \cap \T_{\text{in}}(x) \cap \T_{\text{out}}(x)$ for $x \in W$ is such that any two sets are either the same or disjoint. We can therefore choose some $I \subseteq W$ such that $S(x) \ne S(y)$ when $x$, $y \in I$ with $x \ne y$, and $\big\{S(x)\big\}_{x \in I}$ is a partition of $W$. Further, for each $z \in I$, note that $\left| S(z) \right| = \chi_{\N(z)}\big((S_z,<_z,\rnk_z)\big)$ so there exists a bijection $\beta_z \colon S(z) \to \Big\{\, \big(\N(z),(S_z,<_z,\rnk_z),i\big) \mid i \in [\chi_{\N(z)}\big((S_z,<_z,\rnk_z)\big)] \,\Big\} \subseteq W^{\prime}$. So let $\beta \colon W \to W^{\prime}$ be the bijection such that $\beta|_{S(z)} = \beta_z$ for $z \in I$, noting that it follows from the definition of $W^{\prime}$ that $\beta$ is surjective. We extend $\beta$ to a bijection $\alpha \colon V \cup W \to V \cup W^{\prime}$ such that $\alpha|_V \colon V \to V$ is the identity map, and show that the vertex map $\alpha$ is a type-1 directed bipartite graph isomorphism from $B$ to $B^{\prime}$.

    For arc preservation, let $(v,w) \in F \cap (V \times W)$. There exists $x \in I$ such that $w \in S(x)$, noting $v \in \N(w) = \N(x)$, so $\beta|_{S(x)}(w) = (\N(x),(S_x,<_x,\rnk_x),i)$ for some $i \in [\chi_{\N(x)}((S_x,<_x,\rnk_x))]$. Since $v \in \N_{\text{in}}(w)$, and since $w \in S(x)$ implies $\N_{\text{in}}(x) = \N_{\text{in}}(w)$, we have $v \in \N_{\text{in}}(x)$, so $-1 \in \sigma_x(v)$. It follows that $\alpha\big((v,w)\big) = \big(\alpha(v),\alpha(w)\big) = \big(v,\beta|_{S(x)}(w)\big) = \big(v,(\N(x),(S_x,<_x,\rnk_x),i)\big) \in F^{\prime}$. A similar argument holds when $(w,v) \in F \cap (W\times V)$. Conversely, let $(v,(\N(w),(S_z,<_z,\rnk_z),i)) \in F^{\prime} \cap (V \times W^{\prime})$ where $w \in W$, $\N(w) \in E$, $v \in \N(w)$, $z \in \T(w)$, $(S_z,<_z,\rnk_z) \in \chi_{\N(w)}^{-1}(\NN_+)$, $i \in [\chi_{\N(w)}((S_z,<_z,\rnk_z))]$, and $-1 \in \sigma_z(v)$. Choose $y \in I$ such that $z \in S(y)$, noting $(S_z,<_z,\rnk_z) = (S_y,<_y,\rnk_y)$. Choose $s \in S(y)$ such that $\beta|_{S(y)}(s) = (\N(y),(S_y,<_y,\rnk_y),i) = (\N(w),(S_z,<_z,\rnk_z),i)$, where the last equality uses $\N(y) = \N(z) = \N(w)$. Note that $z$, $s \in S(y)$ implies $\sigma_s = \sigma_z$. Then $\alpha^{-1}\Big((v,(\N(w),(S_z,<_z,\allowbreak\rnk_z),i))\Big) = \Big((\alpha^{-1}(v),\alpha^{-1}((\N(w),(S_z,<_z,\rnk_z),i)))\Big) = \Big(v,\beta|_{S(y)}^{-1}\big((\N(w),(S_z,<_z,\rnk_z),i)\big)\Big) = (v,s) \in F$, where the membership holds since $v \in \N(w) = \N(y) = \N(s)$, and since $-1 \in \sigma_z(v) = \sigma_s(v)$. A similar argument holds when $((\N(w),(S_z,<_z,\rnk_z),i),v) \in F^{\prime} \cap (W^{\prime} \times V)$. Therefore $\alpha$ preserves arcs.

    For arc weight preservation, let $(v,w) \in F \cap (V \times W)$. There exists $x \in I$ such that $w \in S(x)$, so $\alpha(w) = \beta|_{\S(x)}(w) = (\N(x),(S_x,<_x,\rnk_x),i) \in W^{\prime}$ for some $i \in [\chi_{\N(x)}\big((S_x,<_x,\rnk_x)\big)]$. Then $F^{\prime} \owns \alpha\big((v,w)\big) = \big(v,\alpha(w)\big) = \big(v,(\N(x),(S_x,<_x,\rnk_x),i)\big)$, so $\omega^{\prime}\big(\alpha\big((v,w)\big)\big) = \mu^\Minus_x(v) = \mu^\Minus_w(v) = \omega(v,w)$, where the second equality holds since $w \in S(x)$ implies $(S_w,<_w,\rnk_w) = (S_x,<_x,\rnk_x)$. A similar argument holds when $(w,v) \in F \cap (W \times V)$. Therefore $\alpha$ preserves arc weights. Note also that $\Codom(\omega^{\prime}) = \Codom(\omega) = \AA$.

    For second-part weight preservation, since $\Dom(\delta) = W$, $\Dom(\delta^{\prime}) = W^{\prime}$, and $\alpha(W) = W^{\prime}$, we have $\alpha\big(\Dom(\delta)\big) = \Dom(\delta^{\prime})$. Note also that $\Codom(\delta^{\prime}) = \Codom(\delta) = \VV$. Now, if $w \in W$ then there exists $x \in I$ such that $w \in S(x)$, so $\alpha(w) = \beta|_{\S(x)}(w) = (\N(x),(S_x,<_x,\rnk_x),i) \in W^{\prime}$ for some $i \in [\chi_{\N(x)}\big((S_x,<_x,\rnk_x)\big)]$. Noting $(\N(w),(S_w,\allowbreak<_w,\rnk_w),i) = (\N(x),(S_x,<_x,\rnk_x),i)$, we have $\delta^{\prime}\big(\alpha(w)\big) = \delta^{\prime}\big((\N(w),(S_w,<_w,\rnk_w),i)\big) = \rho_w\big(\N(w)\big) = \delta(w)$. Therefore $\alpha$ preserves second-part weights.

    Finally, since $\alpha|_V \colon V \to V$ is the identity map the vertex labels in $V$ are preserved. We conclude that $\alpha$ is a type-1 isomorphism, hence $B \cong_1 B^{\prime}$ and $B \cong_0 B^{\prime}$.

    \vspace{0.2cm}

    Step~\ref{eqK}: Suppose $t \in \{0,1\}$. To show that $\widetilde{K} \colon \widetilde{\B} \to \widetilde{\Hh}$ is well-defined, suppose $B \defeq (V,W,F,(\omega,\AA),(\delta,\VV))$ and $B^{\prime} \defeq (V^{\prime},W^{\prime},F^{\prime},(\omega^{\prime},\AA^{\prime}),(\delta^{\prime},\VV^{\prime}))$ are in $\B$ and $\gamma \colon B \to B^{\prime}$ is a type-t directed bipartite graph isomorphism. Note, in particular, that we therefore have $\AA^{\prime} = \AA$ and $\VV^{\prime} = \VV$. Let $H \defeq K(B) = (V,E,\mu_E,Q,\chi)$ where $Q \defeq (S,<,\rnk)$ and $\chi \colon E \to \bigcup_{e \in E} \{ Q_e \rightarrow \NN \}$, and $H^{\prime} \defeq K(B^{\prime}) = (V^{\prime},E^{\prime},\mu^{\prime}_{E^{\prime}},Q^{\prime},\chi^{\prime})$ where $Q^{\prime} \defeq (S^{\prime},<^{\prime},\rnk^{\prime})$ and $\chi^{\prime} \colon E^{\prime} \to \bigcup_{e^{\prime} \in E^{\prime}} \{ Q^{\prime}_{e^{\prime}} \rightarrow \NN \}$. If $t=1$ then $V^{\prime} = V$, and it follows that $H^{\prime} = H$, so the systems hypergraphs $H$ and $H^{\prime}$ are (trivially) type-$1$ isomorphic under the identity maps on $V = V^{\prime}$ and $S = S^{\prime}$. We now consider the case when $t=0$. Let $\phi \colon V \to V^{\prime}$ be the bijection such that $\phi \defeq \gamma|_V$, and let $\psi \colon S \to S^{\prime}$ be the bijection such that $\psi(\widehat{\sigma}) = \widehat{\sigma}^{\prime}$, $\psi(\widehat{\mu}) = \widehat{\mu}^{\prime}$, $\psi(\widehat{\rho}) = \widehat{\rho}^{\prime}$, and $\psi(\widehat{\mu}^{\,\Minus / \Plus}) = \widehat{\mu}^{\,\Minus / \Plus \; \prime}$. We show that $(\phi,\psi)$ is a type-$0$ systems hypergraph isomorphism of $H$ and $H^{\prime}$.

    We begin by showing that $\phi \colon (V,E,\mu_E) \to (V^{\prime},E^{\prime},\mu^{\prime}_{E^{\prime}})$ is a multihypergraph isomorphism. For hyperedge preservation, let $e \subseteq V$. Then $e \in E$ implies $e = \N(w)$ for some $w \in W$, so $\phi(e) = \phi\big(\N(w)\big) = \gamma\big(\N(w)\big) = \N\big(\gamma(w)\big) \in E^{\prime}$. Conversely, if $\phi(e) \in E^{\prime}$ then $\phi(e) = \N\big(\gamma(x)\big)$ for some $x \in W$, so $\gamma(e) = \gamma\big(\N(x)\big)$, hence $e = \N(x) \in E$. So the bijection $\phi$ preserves hyperedges, hence is an isomorphism. For multihyperedge preservation, if $e \in E$ then there exists $w \in W$ such that $e = \N(w)$, and noting that $\gamma(e) = \N\big(\gamma(w)\big)$ and $\left|\T(w)\right| = \left|\T\big(\gamma(w)\big)\right|$ since $\gamma$ is an isomorphism, we have $\mu^{\prime}_{E^{\prime}}\big(\phi(e)\big) = \mu^{\prime}_{E^{\prime}}\big(\gamma(e)\big) = \mu^{\prime}_{E^{\prime}}\big(\N\big(\gamma(w)\big)\big) = \left|\T\big(\gamma(w)\big)\right| = \left|\T(w)\right| = \mu_E\big(\N(w)\big) = \mu_E(e)$. Therefore, $\phi$ is a multihypergraph isomorphism.

    We show that $\psi \colon S \to S^{\prime}$ is an order isomorphism of the graded posets $Q$ and $Q^{\prime}$. By definition, $\psi$ is a bijection. Further, since $Q$ and $Q^{\prime}$ are both AMT type, it follows that $\psi$ is order preserving, order reflecting, and rank preserving. Therefore $\psi$ is an order isomorphism.

    We show that the distribution functions are preserved. Note that the domain functions are preserved. So let $e \in E$. Consider $\widehat{\sigma} \in S$. Let $w \in W$ with $e = \N(w)$ and let $z \in \T(w)$, so that $\sigma_z \in \widehat{\sigma}_e^{-1} (\NN_+)$. We have $\Phi_{\hat{\sigma},e}(\sigma_z) = \sigma^{\prime}_{\gamma(z)}$, since $v \in \N(w)$ gives $\Phi_{\hat{\sigma},e}(\sigma_z)\big(\phi(v)\big) = \sigma_z \circ \phi^{-1} \big(\phi(v)\big) = \sigma_z(v) = \sigma^{\prime}_{\gamma(z)}\big(\gamma(v)\big) = \sigma^{\prime}_{\gamma(z)}\big(\phi(v)\big)$, where the first equality follows by the definition of $\Phi_{\hat{\sigma},e}$ and the third equality follows since $\gamma$ is an isomorphism. Then $\widehat{\sigma}^{\prime}_{\phi(e)}\big(\Phi_{\hat{\sigma},e}(\sigma_z)\big) = \widehat{\sigma}^{\prime}_{\phi(e)}(\sigma^{\prime}_{\gamma(z)}) = \widehat{\sigma}^{\prime}_{\N(\gamma(w))}(\sigma^{\prime}_{\gamma(z)}) = \left| \T_{\text{dir}}\big(\gamma(z)\big)\right| = \left|\T_{\text{dir}}(z)\right| = \widehat{\sigma}_e(\sigma_z)$. We conclude that $\widehat{\sigma}$ is preserved. The proofs for $\widehat{\mu}$ and $\widehat{\rho}$ in $S$ are similar so we omit the details. Consider $\widehat{\mu}^{\,\Minus / \Plus} \in S$. Let $w \in W$ with $e = \N(w)$ and let $z \in \T(w)$, so that $(\mu^\Minus_z,\mu^\Plus_z) \in \widehat{\mu}^{\,\Minus / \Plus \; -1}_e (\NN_+)$. Note that $\phi\big(\Dom(\mu^\Minus_z)\big) = \phi\big(\N_{\text{in}}(z)\big) = \gamma\big(\N_{\text{in}}(z)\big) = \N_{\text{in}}\big(\gamma(z)\big) = \Dom\big(\mu^{\Minus \; \prime}_{\gamma(z)}\big)$, and $\phi\big(\Dom(\mu^\Plus_z)\big) = \phi\big(\N_{\text{out}}(z)\big) = \gamma\big(\N_{\text{out}}(z)\big) = \N_{\text{out}}\big(\gamma(z)\big) = \Dom\big(\mu^{\Plus \; \prime}_{\gamma(z)}\big)$. We have $\Phi_{\hat{\mu}^{\,\Minus / \Plus},e}\big((\mu^\Minus_z,\mu^\Plus_z)\big) = (\mu^{\Minus \; \prime}_{\gamma(z)},\mu^{\Plus \; \prime}_{\gamma(z)})$, since $v \in \N(w)$ gives $\Phi_{\hat{\mu}^{\,\Minus / \Plus},e}\big((\mu^\Minus_z,\mu^\Plus_z)\big)\big(\phi(v)\big) = (\mu^\Minus_z \circ \phi^{-1},\mu^\Plus_z \circ \phi^{-1})\big(\phi(v)\big) = (\mu^\Minus_z,\mu^\Plus_z)(v) = (\mu^{\Minus \; \prime}_{\gamma(z)},\mu^{\Plus \; \prime}_{\gamma(z)})\big(\gamma(v)\big) = (\mu^{\Minus \; \prime}_{\gamma(z)},\mu^{\Plus \; \prime}_{\gamma(z)})\big(\phi(v)\big)$, where the first equality follows by the definition of $\Phi_{\hat{\mu}^{\,\Minus / \Plus},e}$ and the third equality follows since $\gamma$ is an isomorphism. Then $\widehat{\mu}^{\,\Minus / \Plus \; \prime}_{\phi(e)}\big(\Phi_{\hat{\mu}^{\,\Minus / \Plus},e}\big((\mu^\Minus_z,\mu^\Plus_z)\big)\big) = \widehat{\mu}^{\,\Minus / \Plus \; \prime}_{\phi(e)}\big((\mu^{\Minus \; \prime}_{\gamma(z)},\mu^{\Plus \; \prime}_{\gamma(z)})\big) = \widehat{\mu}^{\,\Minus / \Plus \; \prime}_{\N(\gamma(w))}\big((\mu^{\Minus \; \prime}_{\gamma(z)},\mu^{\Plus \; \prime}_{\gamma(z)})\big) = \left| \T_{\text{in}}\big(\gamma(z)\big) \cap \T_{\text{out}}\big(\gamma(z)\big) \right| = \left| \T_{\text{in}}(z) \cap \T_{\text{out}}(z) \right| = \widehat{\mu}^{\,\Minus / \Plus}_e\big((\mu^\Minus_z,\mu^\Plus_z)\big)$. We conclude that $\widehat{\mu}^{\,\Minus / \Plus}$ is preserved.

    We now show that the association function is preserved. Let $w \in W$ with $e = \N(w)$ and let $z \in \T(w)$, so that $(S_z,<_z,\rnk_z) \in \chi_e^{-1} (\NN_+)$, where $S_z \defeq \{ \sigma_z,\mu_z,\rho_z,(\mu_z^\Minus,\mu_z^\Plus) \}$, $\rnk_z \colon S_z \to \NN$ with $\rnk_z(\sigma_z) = \rnk_z(\mu_z) = \rnk_z(\rho_z) = 0$ and $\rnk_z((\mu_z^\Minus,\mu_z^\Plus)) = 1$, and the comparable attribute sequences are $\sigma_z <_z (\mu_z^\Minus,\mu_z^\Plus)$ and $\mu_z <_z (\mu_z^\Minus,\mu_z^\Plus)$. Since $\Phi_{\hat{\sigma},e}(\sigma_z) = \sigma^{\prime}_{\gamma(z)}$, $\Phi_{\hat{\mu},e}(\mu_z) = \mu^{\prime}_{\gamma(z)}$, $\Phi_{\hat{\rho},e}(\rho_z) = \rho^{\prime}_{\gamma(z)}$, and $\Phi_{\hat{\mu}^{\,\Minus / \Plus},e}\big((\mu^\Minus_z,\mu^\Plus_z)\big) = (\mu^{\Minus \; \prime}_{\gamma(z)},\mu^{\Plus \; \prime}_{\gamma(z)})$, we have $\Psi_e\big((S_z,<_z,\rnk_z)\big) = (S^{\prime}_{\gamma(z)},<^{\prime}_{\gamma(z)},\rnk^{\prime}_{\gamma(z)})$, where $S^{\prime}_{\gamma(z)} \defeq \{ \sigma^{\prime}_{\gamma(z)},\mu^{\prime}_{\gamma(z)},\rho^{\prime}_{\gamma(z)},(\mu_{\gamma(z)}^{\Minus \; \prime},\mu_{\gamma(z)}^{\Plus \; \prime}) \}$, $\rnk^{\prime}_{\gamma(z)} \colon S^{\prime}_{\gamma(z)} \to \NN$ with $\rnk^{\prime}_{\gamma(z)}(\sigma^{\prime}_{\gamma(z)}) = \rnk^{\prime}_{\gamma(z)}(\mu^{\prime}_{\gamma(z)}) = \rnk^{\prime}_{\gamma(z)}(\rho^{\prime}_{\gamma(z)}) = 0$ and $\rnk^{\prime}_{\gamma(z)}((\mu_{\gamma(z)}^{\Minus \; \prime},\mu_{\gamma(z)}^{\Plus \; \prime})) = 1$, and the comparable attribute sequences are $\sigma^{\prime}_{\gamma(z)} <^{\prime}_{\gamma(z)} (\mu_{\gamma(z)}^{\Minus \; \prime},\mu_{\gamma(z)}^{\Plus \; \prime})$ and $\mu^{\prime}_{\gamma(z)} <^{\prime}_{\gamma(z)} (\mu_{\gamma(z)}^{\Minus \; \prime},\mu_{\gamma(z)}^{\Plus \; \prime})$. Now, $\chi^{\prime}_{\phi(e)}\big(\Psi_e\big((S_z,<_z,\rnk_z)\big)\big) = \chi^{\prime}_{\phi(e)}\big((S^{\prime}_{\gamma(z)},<^{\prime}_{\gamma(z)},\rnk^{\prime}_{\gamma(z)})\big) = \chi^{\prime}_{\N(\gamma(w))}\big((S^{\prime}_{\gamma(z)},<^{\prime}_{\gamma(z)},\rnk^{\prime}_{\gamma(z)})\big) = \linebreak \left| \T_{\text{dir}}(\gamma(z)) \cap \T_{\text{awgt}}(\gamma(z)) \cap \T_{\text{vwgt}}(\gamma(z)) \cap \T_{\text{in}}(\gamma(z)) \cap \T_{\text{out}}(\gamma(z)) \right| = \left| \T_{\text{dir}}(z) \cap \T_{\text{awgt}}(z) \cap \T_{\text{vwgt}}(z) \cap \T_{\text{in}}(z) \cap \T_{\text{out}}(z) \right| = \chi_{\N(w)}\big((S_z,<_z,\rnk_z)\big) = \chi_e\big((S_z,<_z,\rnk_z)\big)$, as required.

    We conclude that $(\phi,\psi)$ is a type-$0$ systems hypergraph isomorphism of $H$ and $H^{\prime}$.

    \vspace{0.2cm}

    Step~\ref{eqKinv}: Suppose $t \in \{0,1\}$. To show that $\widetilde{L} \colon \widetilde{\Hh} \to \widetilde{\B}$ is well-defined, suppose $H \defeq (V,E,\mu_E,Q,\chi)$ with $S = \{(\widehat{\sigma},\I),(\widehat{\mu},\AA),(\widehat{\rho},\VV),(\widehat{\mu}^{\,\Minus / \Plus},\AA)\}$ and $H^{\prime} \defeq (V^{\prime},E^{\prime},\mu^{\prime}_{E^{\prime}},Q^{\prime},\chi^{\prime})$ with $S^{\prime} = \{(\widehat{\sigma}^{\prime},\I),(\widehat{\mu}^{\prime},\AA^{\prime}),(\widehat{\rho}^{\prime},\VV^{\prime}),(\widehat{\mu}^{\,\Minus / \Plus \; \prime},\AA^{\prime})\}$ are in $\Hh$, and $(\phi,\psi) \colon H \to H^{\prime}$ is a type-$t$ systems hypergraph isomorphism of $H$ and $H^{\prime}$. Note, in particular, that we therefore have $\AA^{\prime} = \AA$ and $\VV^{\prime} = \VV$. Let $B \defeq K_r(H) = (V,W,F,(\omega,\AA),(\delta,\VV))$ and $B^{\prime} \defeq K_r(H^{\prime}) = (V^{\prime},W^{\prime},F^{\prime},(\omega^{\prime},\AA),(\delta^{\prime},\VV))$. If $t=1$ then $[H]_t = [H]_1 = \{H\}$ is a singleton, and it follows that $\widetilde{L}$ is well-defined. We now consider the case when $t = 0$. Note that $E^{\prime} = \phi(E)$ and $S^{\prime} = \psi(S)$. Let $\tau \colon W \to W^{\prime}$ be the map such that $\tau\big(\big(e,(S_e,<_e,\rnk_e),i\big)\big) \defeq \big(\phi(e),\Psi_e \big((S_e,<_e,\rnk_e)\big),i\big)$ for $\big(e,(S_e,<_e,\rnk_e),i\big) \in W$. The map $\tau$ is well defined since, for $(S_e,<_e,\rnk_e) \in Q_e$, the preservation of the association function gives $\chi^{\prime}_{\phi(e)}\big(\Psi_e \big((S_e,<_e,\rnk_e)\big)\big) = \chi_e((S_e,<_e,\allowbreak\rnk_e))$, and it follows that $\big(\phi(e),\Psi_e \big((S_e,<_e,\rnk_e)\big),i\big) \in W^{\prime}$. Further, since $\phi$ and $\Psi_e$ are bijections it follows that $\tau$ is a bijection. We show that the bijection $\beta \colon V \cup W \to V^{\prime} \cup W^{\prime}$ such that $\beta|_V = \phi$ and $\beta|_W = \tau$ is a type-$0$ directed bipartite graph isomorphism.

    For arc preservation, let $(v,w) \in V \times W$ for some $w \defeq \big(e,(S_e,<_e,\rnk_e),i\big) \in W$. Then $(v,w) \in F$ if and only if $v \in e$ and $-1 \in \sigma(v)$, if and only if $\phi(v) \in \phi(e)$ and $-1 \in \Phi_{\hat{\sigma},e}(\sigma)\big(\phi(v)\big) = \sigma(v)$, if and only if $(\phi(v),\tau(w)) \in F^{\prime}$, if and only if $\beta\big((v,w)\big) = \big(\beta(v),\beta(w)\big) = \big(\phi(v),\tau(w)\big) \in F^{\prime}$. A similar argument holds when $(w,v) \in W\times V$. Therefore $\beta$ preserves arcs.

    For arc weight preservation, let $(v,w) \in F \cap (V \times W)$ for some $w \defeq \big(e,(S_e,<_e,\rnk_e),i\big) \in W$. Note that $\tau(w) = \big(\phi(e),\Psi_e \big((S_e,<_e,\rnk_e)\big),i\big)$ where $\Psi_e \big((S_e,<_e,\rnk_e)\big) \defeq (S^{\prime}_{\phi(e)},<^{\prime}_{\phi(e)},\rnk^{\prime}_{\phi(e)})$, and $\Phi_{\hat{\mu}^{\,\Minus / \Plus},e}\big((\mu^\Minus,\mu^\Plus)\big) = (\mu^\Minus \circ \phi^{-1},\mu^\Plus \circ \phi^{-1}) \in S^{\prime}_{\phi(e)}$ where $(\mu^\Minus,\mu^\Plus) \in S_e$. Then $\omega^{\prime}\big(\beta\big((v,w)\big)\big) = \omega^{\prime}\big((\phi(v),\tau(w))\big) = \mu^\Minus \circ \phi^{-1} (\phi(v)) = \mu^\Minus(v) = \omega\big((v,w)\big)$. A similar argument holds when $(w,v) \in F \cap (W \times V)$. Therefore $\beta$ preserves arc weights.

    For second-part weight preservation, note that $\Dom(\delta) = W$, $\Dom(\delta^{\prime}) = W^{\prime}$, and $\beta(W) = W^{\prime}$, so $\beta\big(\Dom(\delta)\big) = \Dom(\delta^{\prime})$. Let $w \defeq \big(e,(S_e,<_e,\rnk_e),i\big) \in W$, and note that $\tau(w) = \big(\phi(e),\Psi_e \big((S_e,<_e,\rnk_e)\big),i\big)$ where $\Psi_e \big((S_e,<_e,\allowbreak\rnk_e)\big) \defeq (S^{\prime}_{\phi(e)},<^{\prime}_{\phi(e)},\rnk^{\prime}_{\phi(e)})$, and $\Phi_{\hat{\rho},e}(\rho) = \rho \circ \phi^{-1} \in S^{\prime}_{\phi(e)}$ where $\rho \in S_e$. Then $\delta^{\prime}\big(\beta(w)\big) = \delta^{\prime}\big(\tau(w)\big) = \rho \circ \phi^{-1} \big(\phi(e)\big) = \rho(e) = \delta(w)$. Therefore $\beta$ preserves second-part weights.

    We conclude that $\beta$ is a type-$0$ directed bipartite graph isomorphism.

    \vspace{0.2cm}

    Step~\ref{eqKbij}: We show that $\widetilde{L} = \widetilde{K}^{-1}$. For $[H]_t \in \widetilde{\Hh}$ we have $\widetilde{K} \circ \widetilde{L}\big([H]_t\big) = \widetilde{K}\big([K_r(H)]_t\big) = [K \circ K_r(H)]_t = [H]_t$, since $K \circ K_r$ is the identity map on $\Hh$, and it follows that $\widetilde{K} \circ \widetilde{L}$ is the identity map on $\widetilde{\Hh}$. Further, for $[B]_t \in \widetilde{\B}$ we have $\widetilde{L} \circ \widetilde{K}\big([B]_t\big) = \widetilde{L}\big([K(B)]_t\big) = [K_r \circ K(B)]_t = [B]_t$, since $K_r \circ K(B) \cong_t B$ by Part~\ref{thm:KrK} of this Theorem, and it follows that $\widetilde{L} \circ \widetilde{K}$ is the identity map on $\widetilde{\B}$.
\end{proof}

It is important to note that Theorem~\ref{theorem:Corr-SH-DBGW} provides a symbolic association between multihypergraphs and bipartite graphs. This association does not imply an identification of the two types of graphs which have different relational structures; in particular relations (hyperedges) in multihypergraphs are viewed as objects (vertices) in bipartite graphs. While the correspondence means that we can construct one type of graph from the other, in applications we must also ensure logical consistency when transforming between the two structures. This is discussed in the main document in the context of hypostatic abstraction.

\begin{corollary}[Correspondence between multihypergraphs and bipartite graphs]\label{corr:multihypergraphs-bipartite}
    Let $t \in \{0,1\}$, let $\B^{\ast}$ be the set of all vertex-labelled (undirected) bipartite graphs with type-$t$ isomorphism classes $\widetilde{\B}^{\ast}$, and let $\Hh^{\ast}$ be the set of all multihypergraphs with type-$t$ isomorphism classes $\widetilde{\Hh}^{\ast}$.

    There exists a surjection $K^{\ast} \colon \B^{\ast} \to \Hh^{\ast}$ such that:
    \begin{enumerate}[label=(\arabic*),topsep=3pt,itemsep=1pt]
        \item For $(V,W,U) \in \B^{\ast}$ we have $K^{\ast}\big((V,W,U)\big) = (V,E,\mu_E) \in \Hh^{\ast}$, where:
            \begin{enumerate}[label=(\arabic{enumi}.\arabic*),topsep=0pt]
                \item $E \defeq \{\, \mathcal{N}(w) \mid w \in W \,\}$.
                \item $\mu_E \colon E \to \NN_+$ satisfies $\mu_E\big(\mathcal{N}(w)\big) \defeq \left| \T(w) \right|$ for $w \in W$.
            \end{enumerate} 
        \item A right inverse $K^{\ast}_r \colon \Hh^{\ast} \to \B^{\ast}$ of $K^{\ast}$ is given by $K^{\ast}_r\big((V,E,\mu_E)\big) = (V,W,U) \in \B^\ast$, for $(V,E,\mu_E) \in \Hh^\ast$ with $S = \{(\widehat{\sigma},\I),(\widehat{\mu},\AA),(\widehat{\rho},\VV),(\widehat{\mu}^{\,\Minus / \Plus},\AA)\}$, where:
            \begin{enumerate}[label=(\arabic{enumi}.\arabic*),topsep=0pt]
                \item $W \defeq \big\{\, \big(e,(S_e,<_e,\rnk_e),i\big) \mid \text{$e \in E$, $\chi_e^{-1}(\NN_+) = \{(S_e,<_e,\rnk_e)\}$, $i \in [\mu_E(e)]$} \,\big\}$.
                \item $U = \{\, \{v,w\} \mid \text{$v \in e$ and $w = (e,i) \in W$} \,\}$.
            \end{enumerate} 
        \item $K^{\ast}_r \circ K^{\ast} \colon \B^{\ast} \to \B^{\ast}$ satisfies $K^{\ast}_r \circ K^{\ast}(B) \cong_t B$ for $B \in \B^{\ast}$. 
        \item $K^{\ast}$ induces a bijection $\widetilde{K}^{\ast} \colon \widetilde{\B}^{\ast} \to \widetilde{\Hh}^{\ast}$ such that: \label{thm:eqKbijCorr}
            \begin{enumerate}[label=(\arabic{enumi}.\arabic*),topsep=0pt]
                \item $\widetilde{K}^{\ast}\big([B]_t\big) = [K^{\ast}(B)]_t$ for $[B]_t \in \widetilde{\B}^{\ast}$.
                \item $\widetilde{K}^{\ast \; -1}\big([H]_t\big) = [K^{\ast}_r(H)]_t$ for $[H]_t \in \widetilde{\Hh}^{\ast}$. 
            \end{enumerate}
    \end{enumerate}
\end{corollary}

\begin{proof}
As in Theorem~\ref{theorem:Corr-SH-DBGW}, let $\B$ be the set of all vertex-labelled directed bipartite graphs with arc weight functions and second-part weight functions, and type-$t$ isomorphism classes $\widetilde{\B}$, let $\Hh$ be the set of all AMT-type systems hypergraphs with type-$t$ isomorphism classes $\widetilde{\Hh}$, let $K \colon \B \to \Hh$ be the surjection with right inverse $K_r \colon \Hh \to \B$, and let $\widetilde{K} \colon \widetilde{\B} \to \widetilde{\Hh}$ be the bijection induced by $K$.

Now, let $\A \subseteq \B$ be the set of all vertex-labelled directed bipartite graphs $(V,W,F,(\omega,\AA),(\delta,\VV))$ such that: $(v,w) \in F \cap (V \times W)$ if and only if $(w,v) \in F \cap (W \times V)$; $\omega(f) = 1$ for all $f \in F$; and $\delta(w) = 1$ for all $w \in W$. Further, let $\W \subseteq \Hh$ be the subset of all AMT-type systems hypergraphs $(V,E,\mu_E,Q,\chi)$, where $Q = (S,<,\rnk)$ and $S = \{(\widehat{\sigma},\I),(\widehat{\mu},\AA),(\widehat{\rho},\VV),(\widehat{\mu}^{\,\Minus / \Plus},\AA)\}$, such that for each $e \in E$: $\widehat{\sigma}_e^{-1} (\NN_+)$ is a singleton where $\sigma \in \widehat{\sigma}_e^{-1} (\NN_+)$ is the constant function to $\{-1,+1\}$; $\widehat{\mu}_e^{-1} (\NN_+)$ is a singleton where $\mu \in \widehat{\mu}_e^{-1} (\NN_+)$ is the constant function to $2$; $\widehat{\rho}_e^{-1} (\NN_+)$ is a singleton where $\rho \in \widehat{\rho}_e^{-1} (\NN_+)$ is the constant function to $1$; $\widehat{\mu}^{\,\Minus / \Plus \; -1}_{\N(w)} (\NN_+)$ is a singleton where $(\mu^\Minus_z,\mu^\Plus_z) \in \widehat{\mu}^{\,\Minus / \Plus \; -1}_{\N(w)} (\NN_+)$ is a pair of constant total functions to $(1,1)$. Note that the choices for the constant attribute functions associated with $\widehat{\mu}$, $\widehat{\rho}$, and $\widehat{\mu}^{\,\Minus / \Plus}$ are not unique.

Note that $K(\A) \subseteq \W$ and $K_r(\W) \subseteq \A$. Since $K_r|_{\W} \colon \W \to \A$ is a right inverse of $K|_{\A} \colon \A \to \W$, $K|_{\A}$ is a surjection onto $\W$. Since $\A \subseteq \B$ it follows that $K_r|_{\W} \circ K|_{\A} \colon \A \to \A$ satisfies $K_r|_{\W} \circ K|_{\A} (A) = K_r \circ K(A) \cong_t A$ for $A \in \A$. Further, the type-$t$ isomorphism classes $\widetilde{\A}$ of $\A$ satisfy $\widetilde{\A} \subseteq \widetilde{\B}$, and the type-$t$ isomorphism classes $\widetilde{\W}$ of $\W$ satisfy $\widetilde{\W} \subseteq \widetilde{\Hh}$. Since $\widetilde{K}\big([A]_t\big) = [K(A)]_t = [K|_{\A}(A)]_t \in \widetilde{\W}$ for $[A]_t \in \widetilde{\A}$, and $\widetilde{K}^{-1}\big([W]_t\big) = [K_r(W)]_t = [K_r|_{\W}(W)]_t \in \widetilde{\A}$ for $[W]_t \in \widetilde{\W}$, $\widetilde{K}|_{\widetilde{\A}} \colon \widetilde{\A} \to \widetilde{\W}$ is a bijection induced by $K|_{\A}$.

Now, $\B^{\ast}$ is bijective with $\A$ under the map $T$ such that: $\B^{\ast} \owns (V,W,U) \mapsto (V,W,F,(\omega,\AA),(\delta,\VV)) \in \A$, where $\{v,w\} \in U$ with $v \in V$ and $w \in W$ if and only if $(v,w)$, $(w,v) \in F$, $\omega(f) = 1$ for all $f \in F$, and $\delta(w) = 1$ for all $w \in W$; $\A \owns (V,W,F,(\omega,\AA),(\delta,\VV)) \mapsto (V,W,U) \in \B^{\ast}$, where $\{v,w\} \in U$ if and only if $(v,w) \in F \cap (V \times W)$. Further, $\Hh^{\ast}$ is bijective with $\W$ under the map $Z$ such that $(V,E,\mu_E) \in \Hh^{\ast}$ maps to the unique $(V,E,\mu_E,Q,\chi) \in \W$ for some $Q$ and $\chi$.

Defining $K^{\ast} \colon \B^{\ast} \to \Hh^{\ast}$ by $K^{\ast} \defeq Z^{-1} \circ K|_{\A} \circ T$ and $K^{\ast}_r \colon \Hh^{\ast} \to \B^{\ast}$ by $T^{-1} \circ K_r|_{\W} \circ Z$ gives the required result.
\end{proof}

\section{Representation of stochastic Petri nets as SPN-type systems hypergraphs}

\begin{proposition}[Representation of stochastic Petri nets as SPN-type systems hypergraphs] \label{prop:SH-rep-SPN-supp}
Let $\AA \in \{\NN,\NN_+\}$, let $\Y$ be the set of all stochastic Petri nets $Y \defeq (P,T,F,(\omega,\AA),(r,\RR_{\ge 0}))$ with type-1 isomorphism classes $\widetilde{\Y}$, and let $\Hh_{\text{SPN}}$ be the set of all SPN($\AA$)-type systems hypergraphs. There exists a bijection $\Gamma \colon \widetilde{\Y} \to \Hh_{\text{SPN}}$ such that:
    \begin{enumerate}[label=(\arabic*)]
        \item For $(P,T,F,\omega,r) \in \Y$ we have $\Gamma\big([(P,T,F,\omega,r)]_1\big) = (P,E,\mu_E,Q,\chi)$, where:
            \begin{enumerate}[label=(\arabic{enumi}.\arabic*),topsep=0pt]
                \item $E \defeq \{\, \N(w) \mid w \in T \,\}$.
                \item $\mu_E \colon E \to \NN_+$ satisfies $\mu_E\big(\N(w)\big) \defeq \left| \T(w) \right|$ for $w \in T$.
                \item $\widehat{\sigma} \colon E \to \bigcup_{e \in E} \{ \{ e \rightarrow \I \} \rightarrow \NN \}$ satisfies $\widehat{\sigma}_{\N(w)}^{-1} (\NN_+) = \{\, \sigma_z \mid z \in \T(w) \,\}$ and $\widehat{\sigma}_{\N(w)} (\sigma_z) = \left| \T_{\text{dir}}(z) \right|$, for $w \in T$ and $\sigma_z \in {\widehat{\sigma}_{\N(w)}}^{-1} (\NN_+)$.
                \item $\widehat{\mu} \colon E \to \bigcup_{e \in E} \{ \{e \rightarrow \AA\} \rightarrow \NN \}$ satisfies $\widehat{\mu}_{\N(w)}^{-1} (\NN_+) = \{\, \mu_z \mid z \in \T(w) \,\}$ and $\widehat{\mu}_{\N(w)} (\mu_z) = \left| \T_{\text{awgt}}(z) \right|$, for $w \in T$ and $\mu_z \in {\widehat{\mu}_{\N(w)}}^{-1} (\NN_+)$.
                \item $\widehat{\rho} \colon E \to \bigcup_{e \in E} \{ \{ \{e\} \rightarrow \RR_{\ge 0} \} \rightarrow \NN \}$ satisfies $\widehat{\rho}_{\N(w)}^{-1} (\NN_+) = \{\, \rho_z \mid z \in \T(w) \,\}$ and $\widehat{\rho}_{\N(w)} (\rho_z) = \left| \T_{\text{vwgt}}(z) \right|$, for $w \in T$ and $\rho_z \in {\widehat{\rho}_{\N(w)}}^{-1} (\NN_+)$.
                \item $\widehat{\mu}^{\,\Minus / \Plus} \colon E \to \bigcup_{e \in E} \{ \{ e \rightharpoonup \AA \} \times \{ e \rightharpoonup \AA \} \rightarrow \NN \}$ satisfies $\widehat{\mu}^{\,\Minus / \Plus \; -1}_{\N(w)} (\NN_+) = \{\, (\mu^\Minus_z,\mu^\Plus_z) \mid z \in \T(w) \,\}$ and $\widehat{\mu}^{\,\Minus / \Plus}_{\N(w)} ((\mu^\Minus_z,\mu^\Plus_z)) = \left| \T_{\text{in}}(z) \cap \T_{\text{out}}(z) \right|$, for $w \in T$ and $(\mu^\Minus_z,\mu^\Plus_z) \in \widehat{\mu}^{\,\Minus / \Plus \; -1}_{\N(w)} (\NN_+)$.
                \item $\chi \colon E \to \bigcup_{e \in E} \{ Q_e \rightarrow \NN \}$ satisfies $\chi_{\N(w)}^{-1} (\NN_+) = \{\, (S_z,<_z,\rnk_z) \mid z \in \T(w) \,\}$ and $\chi_{\N(w)}\big((S_z,<_z,\rnk_z)\big) =\\ \left| \T_{\text{dir}}(z) \cap \T_{\text{awgt}}(z) \cap \T_{\text{vwgt}}(z) \cap \T_{\text{in}}(z) \cap \T_{\text{out}}(z) \right|$, for $w \in T$.
            \end{enumerate}
        \item For $(V,E,\mu_E,Q,\chi) \in \Hh_{\text{SPN}}$ we have $\Gamma^{-1}\big((V,E,\mu_E,Q,\chi)\big)
        = [(V,T,F,\omega,r)]_1$, where:
            \begin{enumerate}[label=(\arabic{enumi}.\arabic*),topsep=0pt]
                \item $T \defeq \big\{\, \big(e,(S_e,<_e,\rnk_e),i\big) \mid \text{$e \in E$, $(S_e,<_e,\rnk_e) \in \chi_e^{-1}(\NN_+)$, and $i \in [\chi_e\big((S_e,<_e,\rnk_e)\big)]$} \,\big\}$.
                \item $F \subseteq (V \times T) \sqcup (T \times V)$ such that $(v,(e,(S_e,<_e,\rnk_e),i)) \in F$ if and only if $v \in e$ and $-1 \in \sigma(v)$ where $\sigma \in S_e$, and $((e,(S_e,<_e,\rnk_e),i),v) \in F$ if and only if $v \in e$ and $+1 \in \sigma(v)$ where $\sigma \in S_e$.
                \item $\omega \colon F \to \AA$ is given by $\omega\big((v,(e,(S_e,<_e,\rnk_e),i))\big) = \mu^\Minus(v)$ if $(v,(e,(S_e,<_e,\rnk_e),i)) \in F$ and $\omega\big(((e,(S_e,<_e,\rnk_e),i),v)\big) = \mu^\Plus(v)$ if $((e,(S_e,<_e,\rnk_e),i),v) \in F$, where $(\mu^\Minus,\mu^\Plus) \in S_e$.
                \item $r \colon T \to \RR_{\ge 0}$ is given by $r\big((e,(S_e,<_e,\rnk_e),i)\big) = \rho(e)$ for all $(e,(S_e,<_e,\rnk_e),i\big) \in T$ where $\rho \in S_e$.
            \end{enumerate}
    \end{enumerate}
\end{proposition}

\begin{proof}
The result follows from Part~\ref{thm:eqKbijCorr} of Theorem~\ref{theorem:Corr-SH-DBGW} by setting $t = 1$, $\AA \in \{\NN,\NN_+\}$, $\VV \defeq \RR_{\ge 0}$, and noting that $[H]_1 = \{H\}$ for $[H]_1 \in \widetilde{\Hh_{\text{SPN}}}$ with $H \in \Hh_{\text{SPN}}$.
\end{proof}

\section{Representation of chemical reaction networks as CRN-type systems hypergraphs}

\begin{definition}[Twin-reactions maps]
    Suppose $Z \defeq (\S,\C,\R,k)$ is a connected chemical reaction network with a kinetics. The \emph{twin-reactions map} $\W \colon \R \to \powerset(\R)$ is given by $\W\big((f,g)\big) \defeq \{\, (f^{\prime},g^{\prime}) \in \R \mid \Dom(f^{\prime} \boxplus g^{\prime}) = \Dom(f \boxplus g) \,\}$ for $(f,g) \in \R$. The \emph{directed twin-reactions map} $\W_{\text{dir}} \colon \R \to \powerset(\R)$ is such that $\W_{\text{dir}}\big((f,g)\big) \defeq \{\, (f^{\prime},g^{\prime}) \in \R \mid \text{$\Dom(f^{\prime}) = \Dom(f)$ and $\Dom(g^{\prime}) = \Dom(g)$} \,\}$ for $(f,g) \in \R$. Note that $\W_{\text{dir}}\big((f,g)\big) \subseteq \W\big((f,g)\big)$ for $(f,g) \in \R$. The \emph{stoichiometry twin-reactions map} $\W_{\text{stoic}} \colon \R \to \powerset(\R)$ is given by $\W_{\text{stoic}}\big((f,g)\big) \defeq \{\, (f^{\prime},g^{\prime}) \in \W(r) \mid f^{\prime} \boxplus g^{\prime} = f \boxplus g \,\}$ for $(f,g) \in \R$. The \emph{reactant twin-reactions map} $\W_{\text{react}} \colon \R \to \powerset(\R)$ is given by $\W_{\text{react}}\big((f,g)\big) \defeq \{\, (f^{\prime},g^{\prime}) \in \W(r) \mid f^{\prime} = f \,\}$ for $(f,g) \in \R$, and the \emph{product twin-reactions map} $\W_{\text{prod}} \colon \R \to \powerset(\R)$ is given by $\W_{\text{prod}}\big((f,g)\big) \defeq \{\, (f^{\prime},g^{\prime}) \in \W(r) \mid g^{\prime} = g \,\}$ for $(f,g) \in \R$. The \emph{rate twin-reactions map} $\W_{\text{rate}} \colon \R \to \powerset(\R)$ is given by $\W_{\text{rate}}\big((f,g)\big) \defeq \{\, (f^{\prime},g^{\prime}) \in \W(r) \mid k\big((f^{\prime},g^{\prime})\big) = k\big((f,g)\big) \,\}$ for $(f,g) \in \R$.
\end{definition}

\begin{definition}[Orientation maps induced from chemical reaction networks]
    Let $Z \defeq (\S,\C,\R,k)$ be a connected chemical reaction network with a kinetics. For each $r \defeq (f,g) \in \R$ let $\sigma_r \colon \Dom(f \boxplus g) \to \I$ be the \emph{orientation map} such that, for $x \in \Dom(f \boxplus g)$:
        \begin{equation*}
            \sigma_r(x) \defeq \begin{cases} \{-1\} & \text{if $x \in \Dom(f) \setminus \Dom(g)$},\\ \{+1\} & \text{if $x \in \Dom(g) \setminus \Dom(f)$},\\ \{-1,+1\} & \text{if $x \in \Dom(f) \cap \Dom(g)$}.
            \end{cases}
        \end{equation*}
    \end{definition}

\begin{definition}[Multiset maps induced from chemical reaction networks]
    Let $Z \defeq (\S,\C,\R,k)$ be a connected chemical reaction network with a kinetics. For each $r \defeq (f,g) \in \R$ let:
    \begin{itemize}[topsep=3pt,itemsep=1pt,leftmargin=12pt]
        \item $\mu_r \colon \Dom(f \boxplus g) \to \NN_+$ be the \emph{multiset map} such that, for $x \in \Dom(f \boxplus g)$, $\mu_r(x) \defeq (f \boxplus g)(x)$.
        \item $\mu^\Minus_r \colon \Dom(f \boxplus g) \to \NN_+$ be the \emph{negative oriented-multiset map} such that $\Dom(\mu^\Minus_r) \defeq \Dom(f)$, hence $\mu^\Minus_r$ may be a partial map, and $\mu^\Minus_r|_{\Dom(f)} = f$.
        \item $\mu^\Plus_r \colon \Dom(f \boxplus g) \to \NN_+$ be the \emph{positive oriented-multiset map} such that $\Dom(\mu^\Plus_r) \defeq \Dom(g)$, hence $\mu^\Plus_r$ may be a partial map, and $\mu^\Plus_r|_{\Dom(g)} = g$.
    \end{itemize}
    Note that if $\Dom(f) \ne \emptyset$ (resp. $\Dom(g) \ne \emptyset$) and $\Dom(g) = \emptyset$ (resp. $\Dom(f) = \emptyset$) then $\mu^\Minus_r$ (resp. $\mu^\Plus_r$) is the empty partial function.
    \end{definition}

\begin{definition}[Hyperedge label maps induced from chemical reaction networks]
    Let $Z \defeq (\S,\C,\R,k)$ be a connected chemical reaction network with a kinetics. For each $r \defeq (f,g) \in \R$ let $\rho_r \colon \big\{\Dom(f \boxplus g)\big\} \to \RR_{\ge 0}$ be the \emph{hyperedge label map} such that $\rho_r\big(\Dom(f \boxplus g)\big) \defeq k\big((f,g)\big)$.
\end{definition}

While the notation we employ for orientation, multiset, and hyperedge label maps induced from chemical reaction networks is similar to the notation for the maps corresponding to stochastic Petri nets, no ambiguity is possible since the relevant system is clear from context.

\begin{proposition}[Representation of chemical reaction networks as CRN-type systems hypergraphs]\label{prop:SH-rep-CRN-supp}
Let $\Z$ be the set of all connected chemical reaction networks with a kinetics $Z \defeq (\S,\C,\R,k)$, and let $\Hh_{\text{CRN}}$ be the set of all CRN-type systems hypergraphs $H \defeq (V,E,\mu_E,Q,\chi)$. There exists a bijection $\Theta \colon \Z \to \Hh_{\text{CRN}}$ such that:
    \begin{enumerate}[label=(\arabic*)]
        \item For $(\S,\C,\R,k) \in \Z$ we have $\Theta\big((\S,\C,\R,k)\big) = (V,E,\mu_E,Q,\chi)$, where:
            \begin{enumerate}[label=(\arabic{enumi}.\arabic*),topsep=0pt]
                \item $E \defeq \{\, \Dom\big((f \boxplus g)\big) \mid (f,g) \in \R \,\}$.
                \item $\mu_E \colon E \to \NN_+$ satisfies $\mu_E\big(\Dom\big((f \boxplus g)\big)\big) \defeq \left| \W\big((f,g)\big) \right|$ for $r \defeq (f,g) \in \R$.
                \item $\widehat{\sigma} \colon E \to \bigcup_{e \in E} \{ \{ e \rightarrow \I \} \rightarrow \NN \}$ satisfies ${\widehat{\sigma}_{\Dom((f \boxplus g))}}^{-1} (\NN_+) = \{\, \sigma_s \mid s \in \W(r) \,\}$ and $\widehat{\sigma}_{\Dom((f \boxplus g))} (\sigma_s) = \left| \W_{\text{dir}}(s) \right|$, for $r \defeq (f,g) \in \R$ and $\sigma_s \in {\widehat{\sigma}_{\Dom((f \boxplus g))}}^{-1} (\NN_+)$.
                \item $\widehat{\mu} \colon E \to \bigcup_{e \in E} \{ \{e \rightarrow \NN_+\} \rightarrow \NN \}$ satisfies ${\widehat{\mu}_{\Dom((f \boxplus g))}}^{-1} (\NN_+) = \{\, \mu_s \mid s \in \W(r) \,\}$ and $\widehat{\mu}_{\Dom((f \boxplus g))} (\mu_s) = \left| \W_{\text{stoic}}(s) \right|$, for $r \defeq (f,g) \in \R$ and $\mu_s \in {\widehat{\mu}_{\Dom((f \boxplus g))}}^{-1} (\NN_+)$.
                \item $\widehat{\rho} \colon E \to \bigcup_{e \in E} \{ \{ \{e\} \rightarrow \RR_{\ge 0} \} \rightarrow \NN \}$ satisfies ${\widehat{\rho}_{\Dom((f \boxplus g))}}^{-1} (\NN_+) = \{\, \rho_s \mid s \in \W(r) \,\}$ and $\widehat{\rho}_{\Dom((f \boxplus g))} (\rho_s) = \left| \W_{\text{rate}}(s) \right|$, for $r \defeq (f,g) \in \R$ and $\rho_s \in {\widehat{\rho}_{\Dom((f \boxplus g))}}^{-1} (\NN_+)$.
                \item $\widehat{\mu}^{\,\Minus / \Plus} \colon E \to \bigcup_{e \in E} \{ \{ e \rightharpoonup \NN_+ \} \times \{ e \rightharpoonup \NN_+ \} \rightarrow \NN \}$ satisfies $\widehat{\mu}^{\,\Minus / \Plus \; -1}_{\Dom((f \boxplus g))} (\NN_+) = \{\, (\mu^\Minus_s,\mu^\Plus_s) \mid s \in \W(r) \,\}$ and  $\widehat{\mu}^{\,\Minus / \Plus}_{\Dom((f \boxplus g))} ((\mu^\Minus_s,\mu^\Plus_s)) = \left| \W_{\text{react}}(s) \cap \W_{\text{prod}}(s) \right|$, for $r \defeq (f,g) \in \R$ and $(\mu^\Minus_s,\mu^\Plus_s) \in \widehat{\mu}^{\,\Minus / \Plus \; -1}_{\Dom((f \boxplus g))} (\NN_+)$.
                \item $\chi \colon E \to \bigcup_{e \in E} \{ Q_e \rightarrow \NN \}$ satisfies $\chi_{\Dom((f \boxplus g))}^{-1} (\NN_+) = \{\, (S_s,<_s,\rnk_s) \mid s \in \W(r) \,\}$ and $\chi_{\Dom((f \boxplus g))}\big((S_t,<_t,\rnk_t)\big) = 1$, for $r \defeq (f,g) \in \R$ and $(S_t,<_t,\rnk_t) \in \chi_{\Dom((f \boxplus g))}^{-1} (\NN_+)$.
            \end{enumerate}
        \item For $(V,E,\mu_E,Q,\chi) \in \Hh_{\text{CRN}}$ we have $\Theta^{-1}\big((V,E,\mu_E,Q,\chi)\big) = (V,\C,\R,k)$, where:
            \begin{enumerate}[label=(\arabic{enumi}.\arabic*),topsep=0pt]
               \item $\R \defeq \{\, (f_x,g_x) \mid \text{$x = \big(e,(S_e,<_e,\rnk_e)\big)$ for $e \in E$ and $(S_e,<_e,\rnk_e) \in \chi_e^{-1}(\NN_+)$}\,\}$ such that, for each $x \defeq \big(e,(S_e,<_e,\rnk_e)\big)$ where $e \in E$ and $(S_e,<_e,\rnk_e) \in \chi_e^{-1}(\NN_+)$, we have $f_x$, $g_x \in \{V \rightharpoonup \NN_+\}$, $f_x(v) \defeq \mu^\Minus(v)$ for $v \in \Dom(f_x) \defeq \Dom(\mu^\Minus)$, and $g_x(v) \defeq \mu^\Plus(v)$ for $v \in \Dom(g_x) \defeq \Dom(\mu^\Plus)$, where $(\mu^\Minus,\mu^\Plus) \in S_e$.
                \item $\C \defeq \bigcup_{(f,g) \in \R} \{f,g\}$.
                \item $k \colon \R \to \RR_{\ge 0}$ is given by $k\big((f_x,g_x)\big) \defeq \rho(e)$ for all $x \defeq \big(e,(S_e,<_e,\rnk_e)\big)$ where $e \in E$ and $(S_e,<_e,\rnk_e) \in \chi_e^{-1}(\NN_+)$.
            \end{enumerate}
    \end{enumerate}
\end{proposition}

\begin{proof}
Let $\kappa \colon \Z \to \widetilde{\U}$ be the bijection in Proposition~\ref{prop:CRN-SPN-supp} where $\U$ is the set of all stochastic Petri nets $U \defeq (P,T,F,(\omega,\NN_+),(r,\RR_{\ge 0}))$ such that $\T_{\text{in}}(z) \cap \T_{\text{out}}(z) = \{z\}$ and $\dot{U}_{\text{in}}[z] \ne \dot{U}_{\text{out}}^{\ast}[z]$ for all $z \in T$ with respect to $\dot{U} \defeq (P,T,F,(\omega,\NN_+))$, and $\widetilde{\U}
$ is the set of all type-$1$ isomorphism classes of $\U$. Let $\Gamma \colon \widetilde{\Y} \to \Hh_{\text{SPN}}$ be the bijection in Proposition~\ref{prop:SH-rep-SPN-supp} with $\AA \defeq \NN_+$ where $\Y$ is the set of all stochastic Petri nets $Y \defeq (P,T,F,(\omega,\NN_+),(r,\RR_{\ge 0}))$, and $\widetilde{\Y}$ is the set of all type-$1$ isomorphism classes of $\Y$. We show that $\Gamma|_{\widetilde{\U}} \colon \widetilde{\U} \to \Hh_{\text{CRN}}$ is a well-defined bijection, and then $\Theta \defeq \Gamma|_{\widetilde{\U}} \circ \kappa$ is the required map.

To see that $\Gamma|_{\widetilde{\U}} \colon \widetilde{\U} \to \Hh_{\text{CRN}}$ is well defined, first note that $\U \subseteq \Y$ implies $\widetilde{\U} \subseteq \widetilde{\Y}$, so we need to show $\Gamma(\widetilde{\U}) \subseteq \Hh_{\text{CRN}}$. Let $U \defeq (P,T,F,\omega,r) \in \U$ and consider $H \defeq \Gamma\big([U]_1\big) = (P,E,\mu_E,Q,\chi) \in \Hh_{\text{SPN}}$. We show that $H \in \Hh_{\text{CRN}}$. Suppose $e \in E$ and $z \in T$ with $e = \N(z)$. First, if $\sigma \in \{e \rightarrow \I\}$, $\mu \in \{e \rightarrow \NN_+\}$, and $(\mu^\Minus,\mu^\Plus) \in \{ e \rightharpoonup \NN_+ \} \times \{ e \rightharpoonup \NN_+ \}$ then, noting that there is at most one $w \in \T(z)$ such that $(\sigma_w,\mu_w,(\mu^\Minus_w,\mu^\Plus_w)) = (\sigma,\mu,(\mu^\Minus,\mu^\Plus))$ since $\T_{\text{in}}(w) \cap \T_{\text{out}}(w) = \{w\}$, we have $\sum_{\{\, (S_e,<_e,\rnk_e) \in Q_e \mid \text{$\sigma$, $\mu$, $(\mu^\Minus,\mu^\Plus) \in S_e$} \,\}} \chi_e((S_e,<_e,\rnk_e)) = \sum_{\{\, x \in \T(z) \mid \text{$\sigma_x = \sigma$, $\mu_x = \mu$, $(\mu^\Minus_x,\mu^\Plus_x) = (\mu^\Minus,\mu^\Plus)$} \,\}} \chi_e((S_x,<_x,\rnk_x)) \le \chi_e((S_w,<_w,\rnk_w)) = \left| \T_{\text{dir}}(w) \cap \T_{\text{awgt}}(w) \cap \T_{\text{vwgt}}(w) \cap \T_{\text{in}}(w) \cap \T_{\text{out}}(w) \right| = 1$, where: the first equality holds since $\chi_e^{-1} (\NN_+) = \{\, (S_x,<_x,\rnk_x) \mid x \in \T(z) \,\}$, and the last equality holds since $\T_{\text{in}}(w) \cap \T_{\text{out}}(w) = \{w\}$. Second, if $(S_e,<_e,\rnk_e) \in Q_e$ and $\chi_e((S_e,<_e,\rnk_e)) > 0$ then $(S_e,<_e,\rnk_e) = (S_y,<_y,\rnk_y)$ for some $y \in \T(z)$, so $\mu^\Minus_y = \mu^\Plus_y$ implies $\N_{\text{in}}(y) = \Dom(\mu^\Minus_y) = \Dom(\mu^\Plus_y) = \N_{\text{out}}(y)$, equivalently $\N_{\text{in}}(y) = \N_{\text{out}}(y) = \N(y)$, and $\omega(v,y) = \mu^\Minus_y(v) = \mu^\Plus_y(v) = \omega(y,v)$ for all $v \in \N(y)$, from which it follows that $\dot{U}_{\text{in}}[y] = \dot{U}_{\text{out}}^{\ast}[y]$, however since $\dot{U}_{\text{in}}[y] \ne \dot{U}_{\text{out}}^{\ast}[y]$ we must have $\mu^\Minus_y \ne \mu^\Plus_y$. We conclude that $H \in \Hh_{\text{CRN}}$, and it follows that $\Gamma|_{\widetilde{\U}} \colon \widetilde{\U} \to \Hh_{\text{CRN}}$ is well defined.

To see that the injective map $\Gamma|_{\widetilde{\U}} \colon \widetilde{\U} \to \Hh_{\text{CRN}}$ is a bijection it suffices to show surjectivity, which we establish by showing that $\Hh_{\text{CRN}} \subseteq \Gamma (\widetilde{\U})$. So, suppose $H \defeq (V,E,\mu_E,Q,\chi) \in \Hh_{\text{CRN}} \subseteq \Hh_{\text{SPN}}$, noting that, for all $e \in E$, we have: $\sum_{\{\, (S_e,<_e,\rnk_e) \in Q_e \mid \text{$\sigma$, $\mu$, $(\mu^\Minus,\mu^\Plus) \in S_e$} \,\}} \chi_e((S_e,<_e,\rnk_e)) \in \{0,1\}$ for all $\sigma \in \{e \rightarrow \I\}$, $\mu \in \{e \rightarrow \NN_+\}$, and $(\mu^\Minus,\mu^\Plus) \in \{ e \rightharpoonup \NN_+ \} \times \{ e \rightharpoonup \NN_+ \}$; and if $(S_e,<_e,\rnk_e) \in Q_e$ with $\chi_e((S_e,<_e,\rnk_e)) > 0$ then $(\mu^\Minus,\mu^\Plus) \in S_e$ satisfies $\mu^\Minus \ne \mu^\Plus$. Consider $\Gamma^{-1}(H) = [Y]_1 \in \widetilde{\Y}$ where $Y \defeq (V,T,F,(\omega,\NN_+),(r,\RR_{\ge 0})) \in \Y$. We show that $Y \in \U$ by establishing that $\T_{\text{in}}(z) \cap \T_{\text{out}}(z) = \{z\}$ and $\dot{Y}_{\text{in}}[z] \ne \dot{Y}_{\text{out}}^{\ast}[z]$ for all $z \in T$ with respect to $\dot{Y} \defeq (V,T,F,(\omega,\NN_+))$. Let $z \defeq \big(e,(S_e,<_e,\rnk_e),i\big) \in T$ where $e \in E$, $(S_e,<_e,\rnk_e) \in \chi_e^{-1}(\NN_+)$, and $i \in [\chi_e\big((S_e,<_e,\rnk_e)\big)]$.

If $x \in \T_{\text{in}}(z) \cap \T_{\text{out}}(z)$ then $x = \big(e,(S^{\prime}_e,<^{\prime}_e,\rnk^{\prime}_e),j\big)$ for some $(S^{\prime}_e,<^{\prime}_e,\rnk^{\prime}_e) \in \chi_e^{-1}(\NN_+)$ and $j \in [\chi_e\big((S^{\prime}_e,<^{\prime}_e,\rnk^{\prime}_e)\big)]$ such that there exist $\sigma \in \{e \rightarrow \I\}$, $\mu \in \{e \rightarrow \NN_+\}$, and $(\mu^\Minus,\mu^\Plus) \in \{ e \rightharpoonup \NN_+ \} \times \{ e \rightharpoonup \NN_+ \}$ with $\{\sigma,\mu,(\mu^\Minus,\mu^\Plus)\} \subseteq S_e \cap S^{\prime}_e$, and since $\sum_{\{\, (S_f,<_f,\rnk_f) \in Q_e \mid \text{$\sigma$, $\mu$, $(\mu^\Minus,\mu^\Plus) \in S_f$} \,\}} \chi_e((S_f,<_f,\rnk_f)) \in \{0,1\}$ we must have $x=z$, so $\T_{\text{in}}(z) \cap \T_{\text{out}}(z) = \{z\}$. Further, if $\dot{Y}_{\text{in}}[z] = \dot{Y}_{\text{out}}^{\ast}[z]$ then, for $\big(e,(S_f,<_f,\rnk_f),j\big) \in T$ with $(\mu^\Minus,\mu^\Plus) \in S_f$, we have $\mu^\Minus(v) = \omega\big((v,(e,(S_f,<_f,\allowbreak\rnk_f),j))\big) = \omega\big(((e,(S_f,<_f,\rnk_f),j),v)\big) = \mu^\Plus(v)$ for all $v \in e$, hence $\mu^\Minus = \mu^\Plus$. So we must have $\dot{Y}_{\text{in}}[z] \ne \dot{Y}_{\text{out}}^{\ast}[z]$. Therefore $Y \in \U$, so $H = \Gamma \circ \Gamma^{-1}(H) = \Gamma\big([Y]_1\big) \in \Gamma(\widetilde{\U})$. We conclude that $\Hh_{\text{CRN}} \subseteq \Gamma(\widetilde{\U})$, hence $\Gamma(\widetilde{\U}) = \Hh_{\text{CRN}}$.

Finally, we can check that $\Theta = \Gamma|_{\widetilde{\U}} \circ \kappa$ as required.

\end{proof}

\end{appendices}

\end{document}